\documentclass[letterpaper, 12 pt, conference, draftclsnofoot, onecolumn]{ieeeconf}  

\IEEEoverridecommandlockouts                              

\usepackage{amssymb}  
\usepackage{amssymb}
\usepackage{mathtools}
\usepackage{xcolor}
 
\usepackage{graphicx}

\usepackage{algorithm}
\usepackage{algpseudocode}
\usepackage{adjustbox}
\usepackage[normalem]{ulem}
\usepackage{soul}

\usepackage{cite}

 \newtheorem{ass}{Assumption}[section]
 
 \newtheorem{lem}{Lemma}[section]
 \newtheorem{thm}{Theorem}[section]
 \newtheorem{coro}{Corollary}[section]
 
 \newtheorem{defn}{Definition}[section]
 \newtheorem{rem}{Remark}[section]

\title{\LARGE \bf
Regret Analysis of Learning-Based Linear Quadratic Gaussian Control with Additive Exploration*
}

\author{Archith Athrey$^{1}$, Othmane Mazhar$^{2}$, Meichen Guo$^{1}$, \\ Bart De Schutter$^{1}$ and Shengling Shi$^{1}$
\thanks{*This research has received funding from the
European Research Council (ERC) under the European Union’s Horizon
2020 research and innovation programme (Grant agreement No.101018826 - CLariNet).}
\thanks{$^{1}$Archith Athrey, Meichen Guo, Bart De Schutter, and Shengling Shi are with the Delft Center for Systems and Control,
        Delft University of Technology, Mekelweg 5, 2628 CD Delft, The Netherlands
        {\tt\small a.athrey@student.tudelft.nl}, {\tt\small \{m.guo, S.Shi-3, b.deschutter\}@tudelft.nl}}%
\thanks{$^{2}$Othamne Mazhar is with Université Paris Cité - Grands Moulins Campus, 5 Rue Thomas Mann, 75013 Paris, France
        {\tt\small  omazhar@lpsm.paris} }%
}

\begin{document}

\maketitle
\thispagestyle{plain}
\pagestyle{plain}

\begin{abstract}
In this paper, we analyze the regret incurred by a computationally efficient exploration strategy, known as naive exploration, for controlling unknown partially observable systems within the Linear Quadratic Gaussian (LQG) framework. We introduce a two-phase control algorithm called LQG-NAIVE, which involves an initial phase of injecting Gaussian input signals to obtain a system model, followed by a second phase of an interplay between naive exploration and control in an episodic fashion. We show that LQG-NAIVE achieves a regret growth rate of $\tilde{\mathcal{O}}(\sqrt{T})$, i.e., $\mathcal{O}(\sqrt{T})$ up to logarithmic factors after $T$ time steps, and we validate its performance through numerical simulations. Additionally, we propose LQG-IF2E, which extends the exploration signal to a `closed-loop' setting by incorporating the Fisher Information Matrix (FIM). We provide compelling numerical evidence of the competitive performance of LQG-IF2E compared to LQG-NAIVE.
\end{abstract}

\section{INTRODUCTION}
In optimal control, Linear Quadratic (LQ) control has been a benchmark for decades and is extensively used to control complex systems. Moreover, the insights gleaned from studying LQ control problems can be translated into a critical understanding of more complex control problems. In this work, we address adaptive control of \textit{unknown} partially observable linear dynamical systems in the Linear Quadratic Gaussian (LQG) setting. Adaptive control caters to the control of unknown systems, where the controller is updated online from the collected data, to optimize some performance measures \cite{matni2019self}. The LQG control problem is one of the key issues in adaptive control \cite{bertsekas2012dynamic}. The seemingly benign difference of not being able to measure the true states will in fact pose a significant challenge when controlling a system with unknown dynamics \cite{lale2021adaptive}. The errors in the state estimates due to approximate models could potentially accumulate to have a significant impact on the control performance. 

In adaptive control, there exists a conflict between learning and control, which inherently necessitates a trade-off between exploration (learning) and exploitation (control). A metric called regret is a way to quantify an adaptive controller's performance, i.e., how well the adaptive controller balances exploration and exploitation \cite{abbasi2011regret}. The regret quantifies the cumulative performance gap over a finite time horizon between the adaptive controller and the ideal controller having full knowledge of the true system dynamics. For the adaptive control in the LQG setting, several works have contributed to statistical guarantees on both learning and control \cite{mania2019certainty,lale2020logarithmic,lale2021adaptive,kargin2023thompson,ziemann2022regret}. 

The adaptive control algorithm in \cite{lale2021adaptive} uses an exploration approach called optimism in the face of uncertainty, which utilizes model parameter uncertainty to engender optimism in its deployed policy. Although this scheme is shown to guarantee a regret growth of $\tilde{\mathcal{O}}(\sqrt{T})$, i.e., $\mathcal{O}(\sqrt{T})$ after $T$ time steps up to logarithmic factors in $T$, finding optimistic model parameters involves non-convex optimization. In \cite{mania2019certainty}, the performance of the Certainty Equivalence Controller (CEC) is analyzed, but it is analyzed without any exploration or online model updates. A Thompson-sampling-based approach is adopted in \cite{kargin2023thompson}, which exploits parameter uncertainty and promises computational efficiency. This approach is also shown to guarantee a regret growth of $\tilde{\mathcal{O}}(\sqrt{T})$. The results in \cite{ziemann2022regret} show that the best regret upper bound that one can achieve is $\Tilde{\mathcal{O}}(\sqrt{T})$, for the LQ setting. In \cite{lale2020logarithmic}, under an additional assumption that the optimal controller persistently excites the true underlying system, a convex reparametrization of a linear dynamical controller is considered, which guarantees a polylogarithmic regret upper bound. However, this assumption can be restrictive since the optimal controller in the LQ setting typically cannot persistently excite the true underlying system \cite{ziemann2022regret}.

In the Linear Quadratic Regulator (LQR) setting, it has been shown that naive exploration, which involves a simple CEC with an additive excitation signal, whose covariance diminishes at a rate $\mathcal{O}(1/\sqrt{t})$ for intermediate time step $t \leq T$, is sufficient to guarantee a regret growth of $\tilde{\mathcal{O}}(\sqrt{T})$ \cite{mania2019certainty, simchowitz2020naive}. Whereas, naive-exploration-based control with regret guarantee in the LQG setting is still an open problem.

In the present work, we investigate naive exploration in the LQG setting. We propose two adaptive control algorithms, LQG-NAIVE and LQG-IF2E, which operate in an episodic fashion and conduct exploration by injecting additive Gaussian signals. While the covariance of the Gaussian exploration signal in LQG-NAIVE decreases over episodes, the covariance of the exploration signal in LQG-IF2E adjusts adaptively to the data informativity by exploiting the Fisher Information Matrix (FIM). The latter exploration strategy is inspired by the approach in \cite{colin2022regret} designed for the LQR setting.

The structure of the adaptive control algorithms proposed in this work is similar to the one in \cite{lale2021adaptive}; however, the main difference lies in the exploration strategy: the algorithm in \cite{lale2021adaptive} employs optimism in the face of uncertainty and thus requires solving non-convex optimization problems online for exploration, whereas in this work, we consider an additive Gaussian signal for exploration. This additive exploration avoids solving optimization problems online and is thus much more computationally efficient. This difference in the exploration strategy poses a major challenge to the regret analysis of the proposed algorithms. In this work, we establish a regret growth rate of $\tilde{\mathcal{O}}(\sqrt{T})$ for LQG-NAIVE. This is achieved by exploiting and extending the analysis techniques for the naive exploration in the LQR setting \cite{simchowitz2020naive} and the techniques for analyzing the CEC in the LQG setting \cite{mania2019certainty}. Moreover, the performance of LQG-NAIVE and LQG-IF2E is validated in numerical simulations.

The contributions of this work can be summarized as follows:
\begin{itemize}
    \item A novel regret guarantee of $\tilde{\mathcal{O}}(\sqrt{T})$ is established for a naive-exploration-based adaptive control algorithm in the LQG setting.

    \item A novel adaptive control algorithm that exploits the FIM for exploration is proposed for the LQG setting and is validated in simulations.
\end{itemize}

\section{Preliminaries}

\subsection{Notations}
Let $\mathbb{Z}^+$ denote the set of non-negative integers. The Euclidean norm of a vector $x$ is denoted by $||x||$. For a matrix $X \in \mathbb{R}^{n \times m}$, $||X||$ denotes the spectral norm, $\rho(X)$ denotes the spectral radius, $||X||_\mathrm{F}$ denotes the Frobenius norm, $X^\top$ denotes its transpose, $X^\dagger$ denotes the Moore-Penrose inverse, and $\text{Tr}(X)$ denotes the trace. The determinant of a matrix $X$ is denoted by $\text{det}(X)$. The $j^\text{th}$ singular value of a matrix $X$ is denoted by $\sigma_j(X)$, where $\sigma_\text{max}(X) \coloneqq \sigma_1(X) \geq \sigma_2(X) \geq ... \geq \sigma_\text{min}(X) \coloneqq \sigma_{\min(n,m)}(X) > 0$. Similarly, $\lambda_\text{min}(X)$ and $\lambda_\text{max}(X)$ have analogous meanings for the eigenvalues of $X$. The identity matrix with the appropriate dimension is denoted by $I$ and similarly, $0$ is a matrix or a vector of $0$'s with appropriate dimensions. Further, $\mathcal{N}(\mu, \Sigma)$ denotes a multivariate normal distribution with a mean vector $\mu$ and a covariance matrix $\Sigma$. The expectation operator is denoted by $\mathbb{E}$, and $\mathbb{P}$ denotes the probability of an event occurring. The inequality $f \lesssim g$ denotes $f \leq C g$ for a universal constant $C$, and $f \lessapprox g$ denotes informal inequality. The informal inequality is used when it is required to hide some of the terms in $g$. The Kronecker product is denoted by $\otimes$, \verb|vec| denotes the vectorization operator. Further, \verb|D|$_\theta$ denotes Jacobian, \verb|d|$_\theta$ denotes differential, and $\nabla_\theta$ denotes gradient, with respect to $\theta$.

The notation $\text{diag}(\cdot)$ denotes a block diagonal matrix with the arguments as the blocks along the main diagonal. Further, $\mathcal{N}(\mu, \Sigma)$ denotes a multivariate normal distribution with a mean vector $\mu$ and a covariance matrix $\Sigma$. The expectation operator is denoted by $\mathbb{E}$, and $\mathbb{P}$ denotes the probability of an event occurring.  Given two functions $f(\cdot)$ and $g(\cdot)$, whose domain and co-domain are subsets of non-negative real numbers, we write $f(x) = \mathcal{O}(g(x))$ if $\exists$ $c>0$ and $\Tilde{x} \geq 0$ such that $f(x) \leq c g(x) $ for all $ x \geq \Tilde{x}$. We write $f(x) = \Omega(g(x))$ if $\exists$ $c>0$ and $\Tilde{x} \geq 0$ such that $f(x) \geq c g(x)$  for all $x \geq \Tilde{x}$. The notations $\mathcal{\Tilde{O}(\cdot)}$ and $\Tilde{\Omega}(\cdot)$ ignore constants that depend on the true system parameters and poly-logarithmic factors that depend on the number of time steps $T$. The inequality $f(x) \lesssim g(x)$ denotes $f(x) \leq c g(x)$ for a universal constant $c > 0$. The notation $\mathrm{poly}(\cdot)$ denotes a polynomial function. In this paper, $\hat{X}$ is used to denote an estimate of the true quantity $X$. Further, at time step $t$, $\hat{X}_t$ is used to denote an estimate that is available for the true quantity $X$.

\subsection{Problem setting}
A discrete-time Linear Time-Invariant (LTI) system is characterized by the state-space equation
\vspace{0mm}
\begin{equation}
\begin{split}
    x_{t+1} & = Ax_t + Bu_t + w_{t},\hspace{2mm} w_t \sim \mathcal{N}(0, \sigma_w^2I), \\
    y_t & = Cx_t + z_{t},\hspace{2mm} z_t \sim \mathcal{N}(0, \sigma_z^2I),
\end{split}
\label{eq: LQG state space equation}
\end{equation}
for $t = 0,1,2,\dots$, $A \in \mathbb{R}^{n_x\times n_x}$, $B \in \mathbb{R}^{n_x \times n_u}$, and $C \in \mathbb{R}^{n_y \times n_x}$. At time step $t$, $u_t \in \mathbb{R}^{n_u}$ is the input, $x_t \in \mathbb{R}^{n_x}$ is the state, $w_{t} \in \mathbb{R}^{n_x}$ is the process noise, $y_t \in \mathbb{R}^{n_y}$ is the system output, and $z_{t} \in \mathbb{R}^{n_y}$ is the measurement noise. Let the model parameter of the true system be $\Theta = (A, B, C)$. To measure the performance of a controller, the cost $c_t$ incurred at time step $t$ is defined as
\begin{equation*}
    c_t = y_t^\top Qy_t + u_t^\top Ru_t,
    \label{eq: LQG cost at time step t}
\end{equation*}
where $Q \in \mathbb{R}^{n_y \times n_y}$ and $R \in \mathbb{R}^{n_u \times n_u}$ are positive-definite. In this work, the infinite-horizon setting is considered, wherein the goal is to design an input signal such that the long-term average expected cost is minimized. The long-term average expected cost in this setting is given by
\vspace{0mm}
\begin{equation*}
J = \lim_{T \rightarrow \infty} \frac{1}{T} \mathbb{E} \left[ \sum_{t = 0}^{T-1} c_t \right]\text{ s.t. } \eqref{eq: LQG state space equation}.
\label{eq: average expected cost, LQG infinite horizon}
\end{equation*}
Let $J_*$ denote the optimal long-term average expected cost for the true system with parameter $\Theta$. The true system is assumed to be controllable and observable to ensure that $J_*$ exists \cite{bertsekas2012dynamic}. A linear system with parameter $\Theta$ is controllable if the controllability matrix denoted by $\mathbf{C}(A,B,n_x)$, where 
\vspace{0mm}
\begin{equation*}
    \mathbf{C}(A,B,n_x) \coloneqq 
    \begin{bmatrix}
        B & AB& {A}^2B & . & .& . & {A}^{n_x-1}B
    \end{bmatrix},
    \label{eq: controllability matrix}
\end{equation*}
has full row rank. Similarly, a linear system with parameter $\Theta$ is observable if the observability matrix denoted by $\mathbf{O}(A,C,n_x)$, where 
\vspace{0mm}
\begin{equation*}
    \mathbf{O}(A,C,n_x) \coloneqq 
    \begin{bmatrix}
        C \\ CA\\ C{A}^2 \\ .\\.\\. \\ C{A}^{n_x-1}
    \end{bmatrix},
    \label{eq: observability matrix}
\end{equation*}
has full column rank. An optimal feedback control law minimizing $J$ has the following form: 
 \begin{equation}
    u_t = -K\hat{x}_{t|t, \Theta},
    \label{eq: optimal control law for LQG}
\end{equation}
where $\hat{x}_{t|t,\Theta}$ is the state estimate given the true parameter value $\Theta$ and the observations until time step $t$, and
 \begin{equation*}
    K = ({B}^\top P B + R)^{-1}{B}^\top P A,
    \label{eq: optimal control gain}
\end{equation*}
with $P$ solving the following Discrete-Time Algebraic Riccati Equation (DARE) \cite{bertsekas2012dynamic}:
\begin{equation}
    P = C^\top Q C + {A}^\top P A - {A}^\top P B ({B}^\top P B + R )^{-1}{B}^\top P A.
    \label{eq: DARE to compute K}
\end{equation}
The state estimate $\hat{x}_{t|t,\Theta}$ in \eqref{eq: optimal control law for LQG} is obtained from the Kalman filter:
\begin{equation}
    \begin{split}
    \hat{x}_{t|t,\Theta} & = (I - LC)\hat{x}_{t|t-1,\Theta} + Ly_t,\\
    \hat{x}_{t+1|t,\Theta} & = A\hat{x}_{t|t,\Theta} + Bu_t,\\
    L & = \Sigma {C}^\top ( C \Sigma {C}^\top + \sigma_z^2I )^{-1},
    \end{split}
    \label{eq: measurement and time update}
\end{equation}
where $\hat{x}_{0|-1,\Theta} = 0$, $L$ is the Kalman gain, and $\Sigma$ is the solution to the following DARE:
\begin{equation*}
    \Sigma = \sigma_w^2I + A \Sigma A^\top - {A} \Sigma {C}^\top (C \Sigma {C}^\top + \sigma_z^2I )^{-1}{C} \Sigma {A}^\top.
    \label{eq: DARE for Kalman filter}
\end{equation*}
Further, the system is assumed to start at the steady state, i.e., $x_0 \sim \mathcal{N}(0,\Sigma)$. In \eqref{eq: measurement and time update}, the two expressions concerning $\hat{x}_{t|t,\Theta}$ and $\hat{x}_{t+1|t,\Theta}$ can be combined to obtain the innovation form \cite{verhaegen2007filtering}:
\begin{equation}
    \begin{split}
        \hat{x}_{t+1|t,\Theta} & =  A\hat{x}_{t|t-1,\Theta} + Bu_t + Fe_t,\\
        e_t & =  C\left(x_t - \hat{x}_{t|t-1,\Theta}\right) + z_t,\\
        e_t & \sim \mathcal{N}(0, \Sigma_e),
    \end{split}
    \label{eq: Innovation model LQG}
\end{equation}
    where $\Sigma_e = C \Sigma {C}^\top  + \sigma_z^2I$ and $F = AL$ is the Kalman gain in the innovations form. The innovations form \eqref{eq: Innovation model LQG} can also be expanded to obtain the one-step-ahead prediction model \cite{verhaegen2007filtering}:
\vspace{0mm}
\begin{equation}
    \begin{split}
        \hat{x}_{t+1|t,\Theta} & = (A - FC)\hat{x}_{t|t-1,\Theta} + Bu_t + Fy_t,\\
        \hat{y}_{t+1|t,\Theta} & =  C\hat{x}_{t+1|t,\Theta},
    \end{split}
    \label{eq: prediction model LQG}
\end{equation}
where the Kalman gain here ensures that $A-FC$ is asymptotically stable. There exists a closed-form expression for the optimal long-term average expected cost when applying the optimal feedback control law \eqref{eq: optimal control law for LQG} \cite{lale2020regret_prior}:
\vspace{0mm}
\begin{equation}
    \begin{split}
        J_* = \text{Tr}\left(C^\top Q C \Bar{\Sigma}\right) + \sigma^2_z \text{Tr}\left(Q\right) + \text{Tr}\left(P(\Sigma - \Bar{\Sigma}) \right),
    \end{split}
    \label{eq: optimal LT avg exp cost}
\end{equation}
where
\vspace{0mm}
\begin{equation*}
    \Bar{\Sigma} = \Sigma - \Sigma {C}^\top \left( C \Sigma {C}^\top + \sigma^2_zI \right)^{-1} C \Sigma.
    \label{eq: covariance of xt - x_hat_tt}
\end{equation*}

In this work, the true model parameter $\Theta$ is unknown, whereas $Q$ and $R$ are user-defined (known). The main problem considered in this work is to design a controller that computes an input $u_t$ based on the past observations $\mathcal{I}_t$:
\begin{equation}
   \mathcal{I}_t = \{ (y_k, u_k) \mid k=0,1,\dots, t-1\} \cup \{y_t\}.
   \label{eq: observations I_t}
\end{equation}
Moreover, the designed controller should perform optimally with respect to specific metrics. Following the literature \cite{abbasi2011regret}, we consider the notion of regret as our metric. Given a finite time horizon of length $T$, the cumulative regret $ \mathrm{Regret}(T)$ is given by 
\vspace{0mm}
\begin{equation}
     \mathrm{Regret}(T) = \sum_{t = 0}^{T-1} (c_t - J_*),
    \label{eq: cumulative regret}
\end{equation}
where $c_t$ is the cost incurred by the controller at time step $t$. It is desired to have a controller whose $ \mathrm{Regret}(T)$ grows sub-linearly, i.e., $\lim_{T\to \infty} \mathrm{Regret}(T)/T \rightarrow 0$. In this case, the average performance of the adaptive controller converges to the optimal average performance $J_\star$.

As $\Theta$ is unknown, given another parameter value $\Theta'= (A',B',C')$, the notations $K(\Theta')$, $L(\Theta')$, and $P(\Theta')$ denote the control gain in \eqref{eq: optimal control law for LQG}, the Kalman gain in \eqref{eq: measurement and time update}, and the solution of the DARE in \eqref{eq: DARE to compute K} respectively, obtained from the parameter value $\Theta'$. The main assumptions of this work are summarized as follows and are considered to hold throughout the entire paper.
\begin{ass}
$Q$, $R$ are positive-definite, $x_0 \sim \mathcal{N}(0,\Sigma)$, and $\hat{x}_{0|-1,\Theta} = 0$. The state dimension $n_x$ is known. 
\end{ass}

The positive definiteness of $Q$ is required to quantify the sub-optimality in the long-term average cost \cite[Th. 3]{mania2019certainty}. Since the convergence of the Kalman filter gain is exponentially fast, the assumption on $x_0$ is not restrictive \cite{lale2021adaptive}. 

\ass \label{assumption:setS} The unknown model parameter $\Theta$ is an element in a set $\mathcal{S}$ satisfying
    \[
    \mathcal{S} \subseteq \left\{ \Theta'= (A',B',C') \left | 
    \begin{alignedat}{2}
        & \rho(A') < 1, \\
        & (A', B') \text{ is controllable,}\\
        & (A', C') \text{ is observable,}\\
        & (A', F') \text{ is controllable.}\\
    \end{alignedat} \right.
        \right\}. \]
        
Assumption \ref{assumption:setS} is standard in the literature of finite-sample system identification and regret minimization \cite{lale2021adaptive, oymak2019non,mania2019certainty,sarkar2019nonparametric, tsiamis2019finite}. The stability of the open-loop plant is assumed in Assumption~\ref{assumption:setS} to avoid explosive behavior during the initial system identification phase \cite{lale2021adaptive}. 

\defn[\cite{kargin2023thompson}] Given an invertible matrix $\mathbf{T} \in \mathbb{R}^{n_x \times n_x}$ and $\hat{\Theta}_t = (\hat{A}_t, \hat{B}_t, \hat{C}_t)$, the estimated model parameters at time step $t$, we define the following model mismatch pseudo-metric:
\[
 \tau(\hat{\Theta}_t, \Theta) \coloneqq  \min_{\mathbf{T} \in \mathrm{GL}_n} \max \left\{ 
\begin{alignedat}{2}
    & || \hat{A}_t - \mathbf{T}^\top A\mathbf{T}||,\\
    & || \hat{B}_t - \mathbf{T}^\top B||,\\
    &  || \hat{C}_t - C\mathbf{T}||\\
\end{alignedat} 
    \right\}, \]
which is invariant under similarity transformations.

\subsection{Closed-loop system identification}
In this work, we adopt a model-based control approach, in which an estimate of the unknown parameter $\Theta$ is obtained and continuously updated online using a system identification technique. We specifically use the subspace identification \cite{lale2020logarithmic}. We consider the predictor form in \eqref{eq: prediction model LQG}, and for the sake of brevity, we introduce the notation $\Bar{A} = (A - FC)$. At time step $t$, we examine the system's evolution over the last $H$ time steps, with the condition that $t \geq H.$ Then we obtain 
\vspace{0mm}
\begin{equation}
\begin{split}
    y_t & = \mathbf{M}\phi_t + e_t + C\Bar{A}^H\hat{x}_{t-H|t-H-1,\Theta}, 
\end{split}
\label{eq: y_t expansion under predictor form}
\end{equation}
where
\vspace{0mm}
\begin{equation}
    \begin{split}
        \mathbf{M} & \coloneqq 
        \begin{bmatrix}
             M^{(0)} & \cdots & M^{(H-1)} 
        \end{bmatrix} \in \mathbb{R}^{n_y \times (n_y + n_u)H},\\
    \end{split}
    \label{eq: defining M}
\end{equation}
with $M^{(i)} \coloneqq C\Bar{A}^i[F\hspace{1mm}B]$, and $\phi_t \in \mathbb{R}^{(n_y + n_u)H}$ is defined as
\begin{equation}
    \begin{split}
        \phi_t & \coloneqq
        \begin{bmatrix}
            y_{t-1}^\top &\cdots& y_{t-H}^\top & u_{t-1}^\top & \cdots & u_{t-H}^\top
        \end{bmatrix}^\top.
    \end{split}
\end{equation}
Since $\Bar{A}$ is stable, the last term in \eqref{eq: y_t expansion under predictor form} becomes negligible for a large enough $H$, specifically for $H \geq \Bar{H}$ with some $\Bar{H} = \Omega(\log T)$ \cite{lale2020logarithmic}. The exact expression of $\Bar{H}$ can be found in \eqref{eq: value of H for CL}. Now with $\{y_i\}_{i = 0}^{t}$ and $\{u_i\}_{i = 0}^{t-1}$, we have
\vspace{0mm}
\begin{align}
    Y_{t} & = \Phi_{t} {\mathbf{M}}^\top  + E_{t} + N_t \nonumber \\ 
    \implies Y_{t} & \approx \Phi_{t} {\mathbf{M}}^\top  + E_{t},
    \label{eq: input-output trajectory representation}
\end{align}
where
\vspace{0mm}
\begin{equation*}
\begin{split}
    Y_{t} & = 
\begin{bmatrix}
    y_H & y_{H+1} & ... & y_{t}
\end{bmatrix}^\top  \in \mathbb{R}^{N \times n_y},\\
\Phi_{t} & = 
\begin{bmatrix}
    \phi_H & \phi_{H+1} & ... & \phi_{t}
\end{bmatrix}^\top  \in \mathbb{R}^{N \times (n_y + n_u)H},\\
E_{t} & =
\begin{bmatrix}
    e_H & e_{H+1} & ... & e_{t}
\end{bmatrix}^\top  \in \mathbb{R}^{N \times n_y},\\
N_{t} & = 
\begin{bmatrix}
    C\Bar{A}^H\hat{x}_{0|-1,\Theta} & C\Bar{A}^H\hat{x}_{1|0,\Theta} & ... & C\Bar{A}^H\hat{x}_{t -H|t - H -1,\Theta}
\end{bmatrix}^\top  \in \mathbb{R}^{N \times n_y},
\end{split}
\end{equation*}
where $N = t - H + 1$. The approximation in \eqref{eq: input-output trajectory representation} comes from the fact that $N_t$ becomes negligible for a large enough $H$. Therefore, from \eqref{eq: input-output trajectory representation}, the Markov parameters $\mathbf{M}$ of the unknown true system can be estimated using regularized least squares \cite{lale2021adaptive}:
\vspace{0mm}
\begin{equation}
\begin{split}
    \hat{\mathbf{M}}_t^\top  & = (\Phi_{t}^\top \Phi_{t} + \lambda I)^{-1}\Phi_{t}^\top  Y_{t},
\end{split}
\label{eq: RLS to estimate Markov parameters in closed-loop}
\end{equation}
where $\lambda > 0$ is a regularization parameter. Define $V_t = \Phi_{t}^\top \Phi_{t} + \lambda I$. 

Following \cite{lale2020logarithmic}, from $\hat{\mathbf{M}}_t$, a subroutine called SYSID will be deployed in the control algorithms of this work to obtain an estimate of the model parameters $\hat{A}_t, \hat{B}_t, \hat{C}_t, \hat{L}_t$. This subroutine is a variation of the classical Ho-Kalman realization algorithm \cite{ho1966effective}, and details of this identification approach can be found in Algorithm \ref{algorithm: SYSID}, which is presented for the sake of completeness in the Appendix.

\section{Adaptive control with additive exploration}
\subsection{Naive exploration}
This section presents LQG-NAIVE (Algorithm \ref{algorithm: LQG-NAIVE}), which provides a computationally efficient method to address regret minimization in the LQG setting. Overall, the algorithm consists of two phases: the warm-up phase and the adaptive-control phase. 

\textbf{Warm-up phase:} To obtain an initial CEC that can stabilize the unknown true system, an initial model parameter estimate is obtained. This is achieved through pure exploration by injecting Gaussian input signals for $T_\text{w}$ time steps to effectively excite the system and then conducting system identification. The length $T_\text{w}$ of this phase depends on how accurate the initial estimate needs to be \cite{lale2021adaptive}. Moreover, we let $T_\mathrm{w} \geq H$.

\textbf{Adaptive-control phase:} Following the warm-up phase, the algorithm proceeds in an episodic fashion. The number of time steps $l_k$ of the $k^{\text{th}}$ episode satisfies $l_k = 2^k T_\text{w}$, for $k = 0,1,2,\dots$ It holds that the time step at the beginning of the $k^{\text{th}}$ episode equals $l_k$. Since $l_{k+1} = 2l_k$, the total number of episodes $k_\text{fin}$ within a time horizon of length $T$ is $\lfloor \log_2(T/T_{\text{w}})\rfloor$, where $\lfloor \cdot \rfloor$ denotes the floor function. 

At the beginning of the $k^{\text{th}}$ episode, LQG-NAIVE updates the parameter estimate to $\hat{\Theta}_{l_k}$. Then, within this episode, the corresponding CEC with an additive Gaussian excitatory signal is deployed:
\vspace{0mm}
\begin{equation}
    \begin{split}
        u_t & = -K(\hat{\Theta}_{l_k}) \hat{x}_{t|t,\hat{\Theta}_{l_k}} + \eta_t,\\
        \eta_t & \sim \sigma_{\eta_t}\mathcal{N}(0, I),\\
        \sigma^2_{\eta_t} & = \left(\gamma/\sqrt{l_k} \right),
    \end{split}
    \label{eq: LBC policy naive}
\end{equation}
where $\gamma > 0$ is a tuning parameter, and $K(\hat{\Theta}_{l_k})$ is the optimal feedback gain computed from $\hat{\Theta}_{l_k} = ( \hat{A}_{l_k}, \hat{B}_{l_k}, \hat{C}_{l_k})$. 

\begin{algorithm}
\caption{LQG-NAIVE}
\begin{algorithmic}[1]
\State Initialize $Q, R, \gamma > 0$, $H$, $T_\text{w}$, $n_x$, $n_y$, $n_u$, $\sigma_u^2$
\Procedure{Warm-up}{}
\For{$t = 0,1,...,T_\text{w}-1$}
\State Inject $u_t \sim \mathcal{N}(0,\sigma^2_u I)$ 
\EndFor
\State Store $\{(y_t,u_t)\}_{t=0}^{T_\text{w}-1}$
\EndProcedure
\Procedure{Adaptive Control}{}
\For{$k = 0,1,...$}
\State Let $l_k = 2^k T_\text{w}$
\State Calculate $\mathbf{\hat{M}}_{l_k}$ using $ \mathcal{I}_{l_k}$ and \eqref{eq: RLS to estimate Markov parameters in closed-loop}
\State Perform SYSID \cite{lale2020logarithmic} to obtain $\hat{A}_{l_k}, \hat{B}_{l_k}, \hat{C}_{l_k}, \hat{L}_{l_k}$
\For{$t = 2^k T_\text{w},...,2^{k+1} T_\text{w}-1$}
\State Inject $u_t$ as in \eqref{eq: LBC policy naive}  
\EndFor
\EndFor
\EndProcedure
\end{algorithmic}
\label{algorithm: LQG-NAIVE}
\end{algorithm}

\subsection{FIM-based exploration (LQG-IF2E)}
The adaptive control algorithm based on the FIM employs the same structure as in Algorithm 1, with the only difference at step $14$ for exploration. To present the algorithm, called the LQG-IF2E, we first provide an intuition for using the FIM. The recent work \cite{colin2022regret} uses the FIM to explicitly design the exploration signal for the LQR setting. This is motivated by \cite{ziemann2022regret} and \cite{ziemann2021uninformative}, where it is shown that the regret bounds the FIM. This motivates the use of the FIM to influence the rate of regret growth. Moreover, the FIM reflects the informativity of data and thus can be exploited to generate informative data for model re-estimation.
\vspace{0mm}
\defn[\cite{ziemann2022regret}] For a family of parameterized probability densities $\{p_\theta, \theta \in 
\Bar{\mathcal{S} } \}$ of a random variable $x \in \mathbb{R}^n$, where $\Bar{\mathcal{S}} \subseteq  \mathbb{R}^d$, the FIM $\Bar{I}_p(\theta) \in \mathbb{R}^{d \times d}$ is given by
\vspace{0mm}
\begin{equation}
    \Bar{I}_p(\theta) = \int_{\mathbb{R}^n} \nabla_\theta \log p_\theta(x) \left( \nabla_\theta \log p_\theta(x) \right)^\top p_\theta(x) dx,
\end{equation}
whenever the integral exists. 

In the present work, the FIM is constructed based on the approximate model \eqref{eq: input-output trajectory representation} and the output measurements. In this case, the FIM under any policy $\pi$ after collecting the observations $\{y_i\}_{i = 0}^{t}$ and $\{u_i\}_{i = 0}^{t-1}$ for $t \geq H$, is given by 
\vspace{0mm}
\begin{equation}
    I_{H,t} = \sum_{i = H}^{t} \mathbb{E}\left[ \phi_i \phi_i^\top \otimes \Sigma_e^{-1} \right].
    \label{eq: FIM equation}
\end{equation}
The proof of \eqref{eq: FIM equation} is presented in Lemma \ref{Lemma: FIM} in Appendix F for the sake of completeness. The FIM cannot be constructed for the first $H$ time steps since the $\phi_t$ vector is not defined during this period. This is acceptable because after the warm-up phase, sufficient data is collected, i.e., $T_\text{w} \geq H$, to construct the FIM, which is then used in the adaptive control phase. 

There is however a caveat in using the FIM: the FIM requires knowledge of the unknown true parameter $\Theta$, as in \eqref{eq: FIM equation}. To circumvent this issue, we evaluate the FIM instead at $\hat{\Theta}_{l_k}$. Even if $\hat{\Theta}_{l_k}$ can only converge to a similarity transformation of $\Theta$, the eigenvalues of a matrix are preserved under similarity transformation, and thus one can evaluate the FIM with $\hat{\Theta}_{l_k}$. For the simplicity of notations, we use $\hat{\Theta}_t$ here to denote the estimated parameter at time step $t$ to estimate the FIM, and note that in Algorithm~1, $\hat{\Theta}_t = \hat{\Theta}_{l_k}$ when $t \in [l_k, \hspace{1mm}l_{k+1})$. This holds for other estimates as well, e.g., $\mathbf{\hat{M}}_t = \mathbf{\hat{M}}_{l_k}$ when $t \in [l_k, \hspace{1mm}l_{k+1})$. Therefore, we can estimate the `true' FIM as 
\vspace{0mm}
\begin{equation}
    \hat{I}_{H,t} = \sum_{i = H}^{t} \phi_i \phi_i^\top \otimes \hat{\Sigma}_{e,i}^{-1},
    \label{eq: estimated FIM}
\end{equation}
and
\begin{equation*}
    \hat{\Sigma}_{e,i} = \frac{1}{i+1}\sum_{j = 0}^i \left( y_j -  \hat{y}_{j|j-1,\hat{\Theta}_{j-1}}\right) \left( y_j -  \hat{y}_{j|j-1,\hat{\Theta}_{j-1}}\right)^\top,
\end{equation*}
with $\hat{y}_{j|j-1,\hat{\Theta}_{j-1}} = \hat{C}_{j-1}\hat{x}_{j|j-1,\hat{\Theta}_{j-1}}$.

To ensure that the FIM is not ill-conditioned, the exploration strategy in \eqref{eq: LBC policy naive} is used until $\lambda_\text{min}\left(\hat{I}_{H,t}\right)$ is larger than some tolerance value. After achieving this tolerance, the FIM-based exploration strategy is deployed. That is, given $c_\text{tol} >0$, if $\lambda_\text{min}\left(\hat{I}_{H,t}\right) \geq c_\text{tol}$,
\vspace{0mm}
\begin{equation}
\begin{split}
    u_t & = - K(\hat{\Theta}_{l_k}) \hat{x}_{t|t, \hat{\Theta}_{l_k}} + \eta_t,\\
    \eta_t & \sim \left( \alpha/\lambda_\text{min}\left(\hat{I}_{H,t}\right) \right)^{1/2} \mathcal{N}( 0, I ),
\end{split}
\label{eq: LBC policy FIM}
\end{equation}
where $\alpha > 0$ is a tuning parameter. Given that $\hat{I}_{H,t}$ depends on past inputs and outputs through the vector $\phi_t$, the FIM-based exploration strategy is a type of `closed-loop' exploration strategy capable of adaptively changing the magnitude of the exploration signal to the `degree' of informativity.

\section{ Regret guarantee}

To provide a finite-time regret guarantee when deploying LQG-NAIVE, several auxiliary results are required. Since there is an aspect of learning and control in the paradigm of adaptive control, it is imperative to provide corresponding guarantees. From a broader perspective, we need to show that the model parameter estimation error is monotonically decreasing and the input-output signals of the system remain bounded during the adaptive-control phase (Lemma \ref{Lemma: bound on the states and outputs during LBC phase}). Providing a guarantee on learning requires showing that the Markov parameters estimation error is monotonically decreasing (Lemma \ref{Lemma: Markov parameter estimation error during LBC phase}), which requires showing that a `persistence of excitation' condition is satisfied (Lemma \ref{Lemma: PE LBC}). 

Now, since the model parameter estimation error can indeed be shown to decrease monotonically, the corresponding estimation error bound during the adaptive control period can be upper-bounded by the corresponding estimation error bound after the warm-up period (Corollary \ref{Corollary: confidence bound on the system parameters after warm-up}). This analysis approach is extensively exploited to simplify the analyses. The following provides a list of the results involved in providing a  finite-time regret guarantee:

\begin{enumerate}
    \item Bounds on the input and output signals during the adaptive control phase (Lemma \ref{Lemma: bound on the states and outputs during LBC phase}).
    \item Persistence of excitation during the adaptive control phase (Lemma \ref{Lemma: PE LBC}).
    \item Bounds on the  Markov parameter estimation error during the adaptive control phase (Lemma \ref{Lemma: Markov parameter estimation error during LBC phase}).
\end{enumerate}

\subsection{Warm-up phase:} 
The modeling error of the initial parameter estimate after the warm-up phase can be bounded as shown in \cite{lale2021adaptive,lale2020regret}. From \cite[Lemma 3.1]{lale2021adaptive}, the input-output data persistently excites the underlying system during the warm-up period, i.e., $\sigma_\text{min}(V_{T_\text{w}}) = \Omega(T_\text{w})$. Further, from \cite[Th. 3.3]{lale2021adaptive}, it holds that
\begin{equation}
||\mathbf{\hat{M}}_{T_\text{w}} - \mathbf{M}|| \leq \frac{\beta_{T_\text{w}}}{\sqrt{\sigma_\text{min}(V_{T_\text{w}})}} = \tilde{\mathcal{O}}\left( \frac{1}{\sqrt{T_\text{w}}} \right),
\label{eq: Markov param error after warm up}
\end{equation}
with a probability of at least $1 - \delta$, where
\begin{equation*}
    \beta_{T_\text{w}} \coloneqq \sqrt{n_y||\Sigma_e||\log\left( \frac{\text{det}(V_{T_\text{w}})^{1/2}}{\delta \text{det}(\lambda I)^{1/2}} \right)} + ||\mathbf{M}||_\mathrm{F} \sqrt{\lambda} + \frac{\sqrt{H}}{T_\text{w}}
\end{equation*}
for $\delta \in (0,1)$ and $T_\text{w} \geq H \geq \Bar{H}$. 

Then based on \cite[Th.~3.4]{lale2021adaptive}, it holds that if $H \geq \Bar{H}$ and the data is informative, then
\begin{equation}
    \tau(\hat{\Theta}_{t}, \Theta) = \mathcal{O}\left(||\mathbf{\hat{M}}_{t} - \mathbf{M}||\right).
    \label{eq: connecting Markov param error with model param error}
\end{equation}
Combining the above result with \eqref{eq: Markov param error after warm up} shows that $\tau(\hat{\Theta}_{T_\text{w}}, \Theta) = \Tilde{\mathcal{O}}(1/\sqrt{T_\text{w}})$. The results corresponding to the warm-up phase are presented in this work in Appendix D for the sake of completeness. 

\subsection{Adaptive Control Period}

During the adaptive control period, it is imperative to guarantee that the input and output signals remain bounded to ensure the safe operation of the closed-loop system. Such guarantee is provided with LQG-NAIVE, as shown in the following lemma:

\lem \label{Lemma: bound on the states and outputs during LBC phase} For all $t \geq T_\text{w}$ with  $T_\text{w} \geq \Bar{T}_\text{w}$, where $\Bar{T}_\text{w}$ is as defined in \eqref{eq: definition for LB on T_w}, LQG-NAIVE satisfies the following with a probability of at least $1 - \delta$ for $\delta \in (0,1)$:
\vspace{0mm}
\begin{equation}
    \begin{split}
        ||\hat{x}_{t|t,\hat{\Theta}_t}|| \leq \Bar{\mathcal{X}}, \hspace{2mm} ||\hat{x}_{t|t-1,\hat{\Theta}_{t-1}}|| \leq X_\text{est,ac},\\
        ||y_t|| \leq Y_\text{ac}, \hspace{2mm} ||u_t|| \leq U_\text{ac}, \hspace{2mm} ||x_t|| \leq X_\text{ac},
    \end{split}
\end{equation}
for some $\Bar{\mathcal{X}}, X_\text{est,ac}, U_\text{ac}, Y_\text{ac}, X_\text{ac} = \mathcal{O}(\sqrt{\log (T/\delta)})$.

It is important to show that the model parameter estimation error is monotonically decreasing in the adaptive control phase. A critical piece to ensure that lies in guaranteeing the persistence of excitation, which ensures the estimation accuracy of the Markov parameters. Essentially, the persistence of the excitation ensures that the cumulative sum of the covariates $\left( \sum_{i = T_\text{w}}^{t} \phi_i \phi_i^\top \right)$ is positive definite. The guarantee for the persistence of excitation is presented in the following result:
\lem \label{Lemma: PE LBC} We have the following with probability of at least $1 - \delta$ for $\delta \in (0,1)$: for all $t \geq T_\text{ac} + T_\text{w}$, where $T_\text{w} \geq \Bar{T}_{\text{w}}$ for $\Bar{T}_{\text{w}}$ as defined in \eqref{eq: definition for LB on T_w}, and for some constant $\sigma_\text{c} > 0$,
\vspace{0mm}
\begin{equation}
 \sigma_{\text{min}}\left(\sum_{i = T_\text{w}}^{t} \phi_i\phi_i^\top \right) \geq (t - T_\text{w} + 1)\frac{\sigma_\text{c}^2\text{min}\{\sigma_w^2,\sigma_z^2,\sigma_{\eta_{t-1}}^2\}}{8},
 \label{eq: PE LBC}
\end{equation}
where
\vspace{0mm}
\begin{equation}
    T_\text{ac} = \frac{ 512 {{\Upsilon_\text{ac}}}^4 H^2 \text{log}\left(\frac{2H(n_y + n_u)}{\delta}\right)}{\sigma_\text{c}^4\text{min}\{\sigma_w^4,\sigma_z^4,\sigma_{\eta_{t-1}}^4\}},
\end{equation} 
with ${\Upsilon_\text{ac}} = Y_\text{ac} + U_\text{ac}$.

From the persistence of excitation property in \eqref{eq: PE LBC}, we can now provide a bound on the parameter estimation error during the adaptive control phase. 

\lem \label{Lemma: Markov parameter estimation error during LBC phase} For any $t \geq \max \left\{ T_\text{ac} + T_\text{w}, 2T_\text{w}\right\}$, where $T_\text{ac}$ is as defined in Lemma \ref{Lemma: PE LBC} and $T_\text{w} \geq \Bar{T}_{\text{w}}$ for $\Bar{T}_{\text{w}}$ as defined in \eqref{eq: definition for LB on T_w}, the estimate of the Markov parameters, $\mathbf{\hat{M}}_t$, obeys the following bound with a probability of at least $1 - \delta$ for $\delta \in (0,1)$:
\vspace{0mm}
\begin{equation}
    ||\mathbf{\hat{M}}_t - \mathbf{M}|| \leq \frac{\Bar{\beta}_\text{ac}}{\sqrt{\sigma_\text{min}(V_t)}} = \tilde{\mathcal{O}}(1/\sqrt{t}),
\end{equation}
for some $\Bar{\beta}_\text{ac} = \text{poly}(n_y, \Sigma_e, \delta, Y_\text{ac}, U_\text{ac}).$

The proofs of Lemmas \ref{Lemma: bound on the states and outputs during LBC phase} - \ref{Lemma: Markov parameter estimation error during LBC phase} are jointly proven in Appendix A. From \eqref{eq: connecting Markov param error with model param error} and Lemma \ref{Lemma: Markov parameter estimation error during LBC phase}, we have $\tau(\hat{\Theta}_t, \Theta) = \tilde{\mathcal{O}}(1/\sqrt{t})$. Therefore, the model parameter estimation error is monotonically decreasing.

The final piece in establishing the regret upper bound requires bounding the sub-optimality gap $\Delta_{\hat{\Theta}_t} \coloneqq J(\hat{\Theta}_t) - J_*$. This inherently requires a way to represent $J(\hat{\Theta}_t)$. It is a standard procedure to write the long-term average expected cost as a function of the solution to a Lyapunov equation \cite{simchowitz2020naive}, and thus we connect $J(\hat{\Theta}_t)$ to a Lyapunov equation as follows.

Consider the true system \eqref{eq: LQG state space equation}, and another model parameter $\tilde{\Theta} = (\Tilde{A}, \Tilde{B}, \Tilde{C} ) \in \mathcal{S}$. Let $\Tilde{K} = K(\Tilde{\Theta})$ and $\Tilde{L} = L(\Tilde{\Theta})$. Now, define an alternative formulation of the LQG cost function as
\vspace{0mm}
\begin{equation}
\begin{split}
        J_s(\Tilde{\Theta})  =& \lim_{T \rightarrow \infty} \frac{1}{T}\mathbb{E}\left[\sum_{t = 0}^{T-1} x_t^\top Q_c x_t + \hat{x}_{t|t,\Tilde{\Theta}}^\top \Tilde{K}^\top R \Tilde{K} \hat{x}_{t|t,\Tilde{\Theta}} \right],\\
         & \text{ s.t. } \eqref{eq: LQG state space equation},\\
        & \hat{x}_{t|t,\Tilde{\Theta}} = (I - \Tilde{L}\Tilde{C})\hat{x}_{t|t-1,\Tilde{\Theta}} + \Tilde{L}y_t,\\
        & \hat{x}_{t+1|t,\Tilde{\Theta}} = \Tilde{A}\Tilde{x}_{t|t,\Tilde{\Theta}} + \Tilde{B}u_t,\\
        & u_t = -\Tilde{K} \hat{x}_{t|t,\Tilde{\Theta}},\\
\end{split}
\label{eq: LQG control problem with quadratic state}
\end{equation}
where $Q_c = {C}^\top Q C$, $\Tilde{K}$ stabilises the true system, and $\Tilde{A} - \Tilde{F}\Tilde{C}$ is asymptotically stable. This alternative formulation of the quadratic cost shows up when upper bounding the cumulative cost in the regret analysis. Further, consider the following closed-loop state-space equation with extended states:
\vspace{0mm}
 \begin{equation*}
 \begin{split}
        \begin{bmatrix}
            x_t \\ \hat{x}_{t|t,\Tilde{\Theta}}
        \end{bmatrix}
        &
        = \mathbf{\Tilde{G}_{1}}
        \begin{bmatrix}
            x_{t-1} \\ \hat{x}_{t-1|t-1,\Tilde{\Theta}} 
        \end{bmatrix} + \mathbf{\Tilde{G}_{2}}
        \begin{bmatrix}
            w_{t-1} \\ z_t \\ 
        \end{bmatrix},    
 \end{split}
 \end{equation*}
where 
\vspace{0mm}
\begin{equation*}
    \begin{split}
    \mathbf{\Tilde{G}_{1}} = 
            \begin{bmatrix}
            A & -B \Tilde{K}\\
            \Tilde{L} C A & \left(I - \Tilde{L} \Tilde{C} \right) \left( \Tilde{A} - \Tilde{B} \Tilde{K} \right) - \Tilde{L} C B \Tilde{K} 
        \end{bmatrix}, \hspace{4mm}
    \mathbf{\Tilde{G}_{2}} = 
        \begin{bmatrix}
                I & 0\\
                \Tilde{L} C & \Tilde{L} 
        \end{bmatrix}.
    \end{split}
\end{equation*}
Consider the discrete Lyapunov equation $\Tilde{S} = \mathbf{\Tilde{G}_{1}}^\top \Tilde{S} \mathbf{\Tilde{G}_{1}} + \text{diag}(Q_c, \Tilde{K}^\top R \Tilde{K}) $ with $\Tilde{S} \geq 0$ being its solution. Then, we have 
\vspace{0mm}
\begin{equation}
    J_s(\Tilde{\Theta}) = \text{Tr}\left(\mathbf{\Tilde{G}_{2}}^\top \Tilde{S} \mathbf{\Tilde{G}_{2}} \text{diag}(\sigma^2_w I, \sigma^2_z I)\right).
    \label{eq: result on LT avg cost - Lyapunov equation}
\end{equation}
Moreover, it holds that $ J(\Tilde{\Theta})=J_s(\Tilde{\Theta})+ \text{Tr}(Q\sigma_z^2I)$ \cite[Th. 3]{mania2019certainty}. This property can aid in quantifying the sub-optimality gap $\Delta_{\hat{\Theta}_t}$. Now, we are ready to state the regret upper bound.

\thm \label{Theorem: regret LBC} Let $\Bar{T}_\text{w}$ be as defined in \eqref{eq: definition for LB on T_w} and $T_\text{ac}$ be as defined in Lemma \ref{Lemma: PE LBC}. If $T_\text{w} \geq \Bar{T}_\text{w}$, with a probability of at least $1 - \delta$ for $\delta \in (0,1)$, we have for any $T \geq \max\{T_\text{ac} + T_\text{w}, 2T_\text{w}\}$ that the regret of LQG-NAIVE is bounded as
\begin{align}
        &  \mathrm{Regret}(T) \lesssim \sum_{k=0}^{k_\text{fin}-1} l_k \left(J_s(\hat{\Theta}_{l_k}) -  J_*\right) + l_k n_y\sigma^2_z \text{Tr} \left( Q\right) \nonumber \\
        & + l_k \sigma_{\eta_{l_k}}^2 \mathrm{poly} \left( \tau(\hat{\Theta}_{l_k}, \Theta)\right) \nonumber \\
        & + \sqrt{l_k} \mathrm{poly} \left( \tau(\hat{\Theta}_{l_k}, \Theta), X_\text{ac}, \Bar{\mathcal{X}}, ||Q||, ||R|| \right)   \label{eq: approximate regret upper bound} \\
        &  = \Tilde{\mathcal{O}}(\sqrt{T}) \nonumber
\end{align}
where $l_k$ is the number of time steps in the $k^\text{th}$ episode and $k_\text{fin}$ is the total number of episodes.  

The result in Theorem \ref{Theorem: regret LBC} confirms that a naive-exploration-based adaptive control strategy is sufficient to guarantee a $\sqrt{T}$-regret growth, which is the optimal rate of regret growth up to logarithmic terms in the control of unknown partially observable linear systems. The proof is presented in Appendix B. Now, an intuition is provided on how the regret bound is derived. As Algorithm \ref{algorithm: LQG-NAIVE} operates in an episodic fashion, the regret is also analyzed episode-wise. First, an upper bound on the cumulative cost incurred by LQG-NAIVE in any arbitrary episode is obtained. From this, we can obtain an upper bound on the regret for any episode. This episode-wise regret bound is then summed over the number of episodes to obtain the final regret upper bound incurred by LQG-NAIVE during the adaptive control phase as shown in \eqref{eq: approximate regret upper bound}. 

In \eqref{eq: approximate regret upper bound}, the sub-optimality gap $J_s(\hat{\Theta}_{l_k}) -  J_*$ and the exploration cost $\sigma_{\eta_{l_k}}^2 \mathrm{poly} \left( \tau(\hat{\Theta}_{l_k}, \Theta)\right)$ have significant contributions towards the regret, as they are linearly dependent on $l_k$. To provide a bound on the sub-optimality gap, we exploit an earlier result \cite[Th. 4]{mania2019certainty}, which essentially bounds the contribution from the sub-optimality gap by $\tilde{\mathcal{O}}(\log_2(T/T_\text{w}))$. On the other hand, the exploration cost is bounded by $\tilde{\mathcal{O}}(\sqrt{T})$ as $\sigma^2_{\eta_{l_k}} = \gamma/\sqrt{l_k}$. We can further see from the $3^\text{rd}$ and $4^\text{th}$ terms in \eqref{eq: approximate regret upper bound} that the model parameter estimation error along with the established bounds on the state and its estimate, also influence the regret upper bound. From \eqref{eq: connecting Markov param error with model param error} and Lemma \ref{Lemma: Markov parameter estimation error during LBC phase}, we have that the model parameter estimation errors are monotonically decreasing, and as a consequence, $\tau(\hat{\Theta}_t, \Theta) \leq \tau(\hat{\Theta}_{T_\text{w}}, \Theta)$ with a high probability. This result is used in deriving the regret upper bound \eqref{eq: approximate regret upper bound}.

\section{Numerical simulations}

In this section, we validate the performance of LQG-NAIVE and LQG-IF2E through numerical simulations. For the simulation, we consider a linearized version of the web server control problem \cite{aastrom2021feedback}. Different from \cite{aastrom2021feedback}, we consider the partial observability case, i.e., the inclusion of the $C$ matrix and the measurement noise. The true system under consideration is given by
\vspace{0mm}
\begin{equation*}
\begin{split}
    x_{t+1} & = \begin{bmatrix} 0.54 & -0.11\\ -0.026 & 0.63 \end{bmatrix} x_t + \begin{bmatrix} -85 &	4.4\\ -2.5 & 2.8 \end{bmatrix}u_t + w_{t},\\
    y_t & = \begin{bmatrix} 0.2 & 0.3 \\ 0.3 & 0.2  \end{bmatrix}x_t + z_{t},
\end{split}
\label{eq: Web server system}
\end{equation*}
where $w_t, z_t \sim \mathcal{N}(0, 0.01I)$. The cost matrices for the control problem are given by \cite{aastrom2021feedback}:
\vspace{0mm}
\begin{equation*}
            Q = \begin{bmatrix} 5 & 0\\ 0 & 1 \end{bmatrix}, \hspace{4mm}
            R = \begin{bmatrix} \frac{1}{50^2} & 0\\ 0 & \frac{1}{10^6} \end{bmatrix}.
\end{equation*}
The optimal long-term average expected cost calculated from \eqref{eq: optimal LT avg exp cost} is $0.0707$. 

To implement the adaptive control algorithm, the length of the warm-up phase is set to $T_\text{w} = 25$. During the warm-up phase, Gaussian excitatory signals are injected, where $u_t \sim \mathcal{N}(0, 0.1I)$. For the adaptive control phase, the number of episodes is taken to be $k_\text{fin} = 11$. The hyper-parameters for the adaptive control policies \eqref{eq: LBC policy naive} and \eqref{eq: LBC policy FIM} are $\gamma = \frac{\sqrt{T_\text{w}}}{10}$ and $\alpha = 1$ respectively. To avoid ill-conditioned FIM, we select $c_\text{tol} = 1$.  Finally, the length of the input-output data for constructing the $\phi$ vector in system identification is $H = 12$. Each of the algorithms LQG-NAIVE and LQG-IF2E are run 100 times to report the mean and the standard deviation of the regret growth. 

\begin{figure}
\centering
    \includegraphics[width=0.75\textwidth]{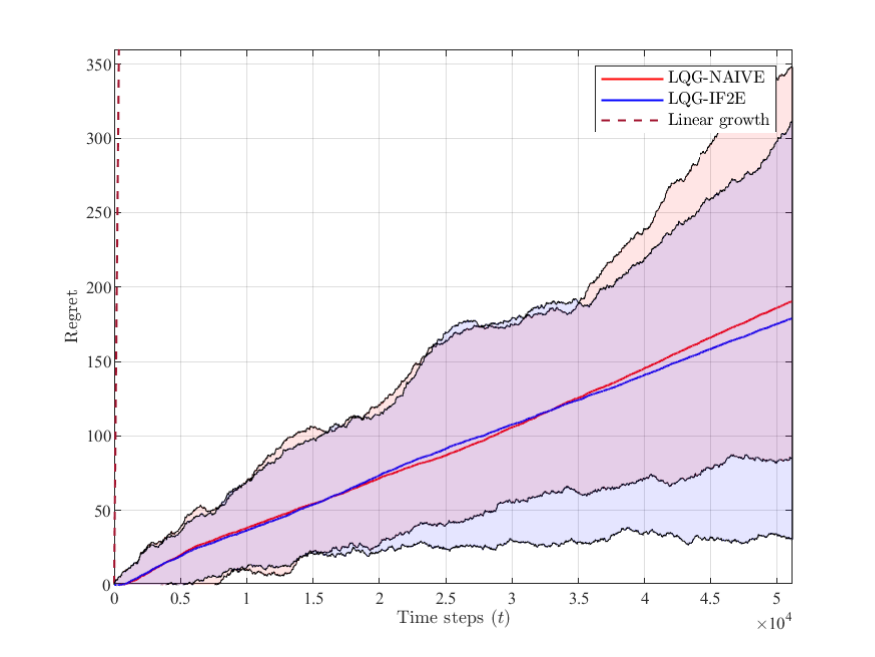}
    \caption{Regret growth of LQG-NAIVE and LQG-IF2E.}
    \label{fig: Regret growth: LQG-NAIVE and LQG-IF2E}
\end{figure}

Fig.~\ref{fig: Regret growth: LQG-NAIVE and LQG-IF2E} shows the regret growth of the $100$ simulations. The bold red line represents the mean regret of LQG-NAIVE, whereas the bold blue line represents the mean regret of LQG-IF2E. LQG-NAIVE incurs a long-term average cost of $0.0744$ and LQG-IF2E incurs a long-term average cost of $0.0742$, averaged over the $100$ simulations. The LQG-IF2E algorithm switches to the FIM-based exploration strategy at the $35^{\text{th}}$ time step, on average. This means that with a delay of approximately one episode, the algorithm is able to deploy the FIM-based exploration strategy. The hyper-parameter $\alpha$ is chosen such that LQG-NAIVE and LQG-IF2E have similar behavior for the regret growth, which is evident from Fig. \ref{fig: Regret growth: LQG-NAIVE and LQG-IF2E}. An intuitive way to understand this similarity in regret growth is by plotting the evolution of the minimum eigenvalue of the FIM.

\begin{figure}[H]
\centering
    \includegraphics[width=0.75\textwidth]{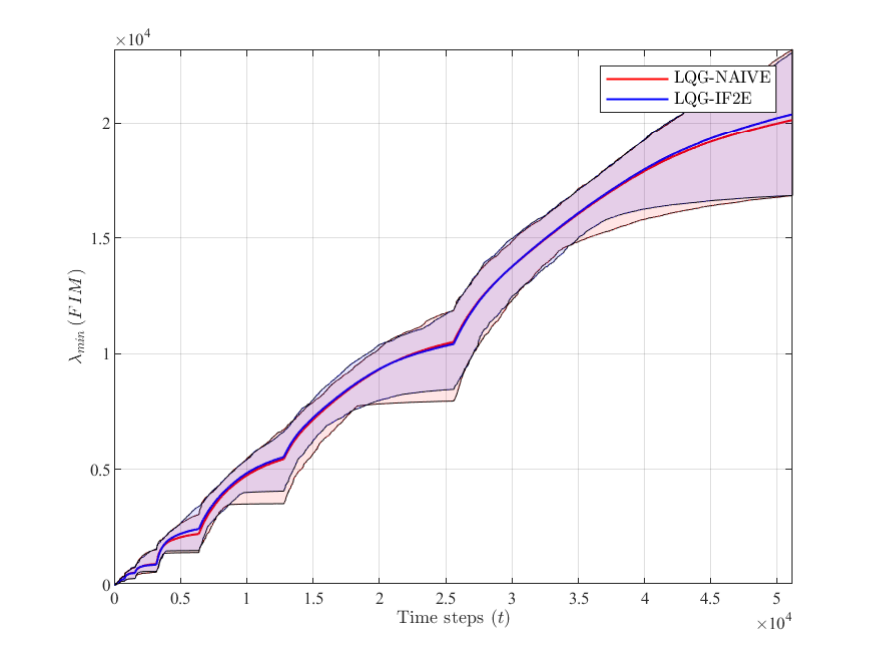}
    \caption{Growth of the minimum eigenvalue of the estimated FIM.}
    \label{fig: min eig FIM: LQG-IF2E and LQG-NAIVE}
\end{figure} 

Fig. \ref{fig: min eig FIM: LQG-IF2E and LQG-NAIVE} shows how the minimum eigenvalue of the FIM varies over the time steps. The bold blue line represents the mean growth of $\lambda_\text{min}\left(\hat{I}_{H,t}\right)$ of LQG-IF2E, whereas the bold red line represents the mean growth of $\lambda_\text{min}\left(\hat{I}_{H,t}\right)$ of LQG-NAIVE. Based on Fig. \ref{fig: min eig FIM: LQG-IF2E and LQG-NAIVE}, the behaviors of the FIMs are also similar between the two algorithms. Since the FIM captures the informativity of the data, which influences the parameter learning rate and hence the growth of the regret, one can expect two algorithms to have similar regret growth if their corresponding FIMs behave similarly. The `bumps' that are observed in Fig. \ref{fig: min eig FIM: LQG-IF2E and LQG-NAIVE} correspond to the time steps where the model parameter estimate $\hat{\Theta}_t$ was updated. The length of the `bumps' corresponds approximately to the length of the episodes.

\section{Conclusions}

We have focused on control of unknown partially observable LTI systems in an LQG setting. We have developed two computationally efficient adaptive control algorithms: LQG-NAIVE and LQG-IF2E. The LQG-NAIVE algorithm, based on naive exploration, is more computationally efficient than optimism-in-the-face-of-uncertainty-based exploration and Thompson-sampling-based exploration. It also has a guaranteed regret growth of $\Tilde{\mathcal{O}}(\sqrt{T})$. 

On the other hand, LQG-IF2E extends the `open-loop' additive excitation signal in LQG-NAIVE to a `closed-loop' additive excitation by incorporating FIM in designing the covariance of the exploration signal. However, providing finite-time regret guarantees for LQG-IF2E is significantly more challenging, as the additive excitation signal is not i.i.d. This is because the FIM depends on all the previous observations. Both algorithms have been validated in numerical simulations and show competitive performance. Deriving finite-time regret guarantees for FIM-based adaptive control strategies such as LQG-IF2E is a topic for future work.





\newpage
\bibliographystyle{IEEEtran} 
\bibliography{ECClqg.bib}

\newpage
\section*{APPENDIX}

Firstly, we will provide the overall organization of the Appendix following which, we will provide the lower bound on the warm-up duration, the system identification algorithm SYSID as well as present a critical component in establishing the persistence of excitation condition in Lemma \ref{Lemma: PE LBC}.

\subsection*{Organization of the Appendices:}
The Appendix is divided into three main segments: in the first segment, we have Appendices A and B, where we present the proofs of Lemmas \ref{Lemma: bound on the states and outputs during LBC phase} - \ref{Lemma: Markov parameter estimation error during LBC phase} and Theorem \ref{Theorem: regret LBC}, respectively. In the second segment, we have Appendices C - D, which present results corresponding to the warm-up phase and the construction of the FIM for the FIM-based algorithm LQG-IF2E, respectively. Finally, in the third segment, we have Appendix E, which presents the technical tools that are utilized in the proofs of Lemmas \ref{Lemma: bound on the states and outputs during LBC phase} - \ref{Lemma: Markov parameter estimation error during LBC phase} and Theorem \ref{Theorem: regret LBC}.

\rem This paper focuses on probabilistic guarantees which depend on $\delta \in (0,1)$. Throughout this paper, we re-parameterize $\delta$ to values (e.g., $c\delta, ct\delta$, or $c\delta/t$ for some $c \in \mathbb{R}$) that aid in simplifying the exposition of the proofs. This re-parameterization has only a logarithmic effect on the established probabilistic guarantees and therefore does not affect the overall analysis of the regret guarantee.

\subsection*{Warm-up duration:}

We let $T_\text{w} \geq \Bar{T}_\text{w}$, where
\begin{equation}
    \Bar{T}_\text{w} \coloneqq \max\{H, T_\mathbf{M}, T_A, T_B, T_o, T_K, T_\text{cl}, T_N, T_a, T_f, T_{\mathcal{G}}, T_J, T_\text{ac}(t)\}.
    \label{eq: definition for LB on T_w}
\end{equation}
The time step $T_{\mathbf{M}}$ is defined in Lemma \ref{Lemma: Markov parameter estimation error after warm-up}. The time steps $T_K$, $T_\text{cl}$, $T_N$, $T_a$ and $T_f$ are defined in the proof of Lemma \ref{Lemma: bound on the states and outputs during LBC phase}, i.e., in Appendix B, $T_\text{o}$ is defined in \eqref{eq: To for warm-up}, $T_A$ and $T_B$ are defined in the proof of \cite[Th. 3.4]{lale2021adaptive}, $T_\mathcal{G}$ and $T_\text{ac}(t)$ are defined in the proof of Lemma \ref{Lemma: PE LBC}, i.e., in Appendix A, and finally, $T_J$ is defined in the proof of Theorem \ref{Theorem: regret LBC}, i.e., in Appendix C. 

We also define other lower bounds for $T_\text{w}$:
\begin{equation}
\begin{split}
    \Bar{T}^\text{alt,1}_{\text{w}} & \coloneqq \max\{H, T_{\mathbf{M}}, T_A, T_B, T_\text{o}\},\\
    \Bar{T}^\text{alt,2}_{\text{w}} & \coloneqq \max\{H, T_{\mathbf{M}}, T_A, T_B, T_\text{o}, T_\text{ac}(t), T_{\mathcal{G}}\},
\end{split}
    \label{eq: Alt LB on T_w}
\end{equation}
which are used for the sake of brevity in the proofs of Lemmas \ref{Lemma: bound on the states and outputs during LBC phase} - \ref{Lemma: Markov parameter estimation error during LBC phase}.

\subsection*{System identification algorithm SYSID:}

Before we present the SYSID algorithm, some of the relevant terms will be defined.

At time step $t$, we have 
\vspace{0mm}
\begin{equation*}
    \hat{\mathbf{M}}_t = 
        \begin{bmatrix}
            \hat{C}_t\hat{F}_t & ... & \hat{C}_t\hat{\Bar{A}}_t^{H-1}\hat{F}_t & \hat{C}_t\hat{B}_t & ... & \hat{C}_t\hat{\Bar{A}}_t^{H-1}\hat{B}_t 
        \end{bmatrix},
\end{equation*}
where $\hat{\Bar{A}}_t = \hat{A}_t - \hat{F}_t \hat{C}_t$. Now, define the Hankel matrix $\mathcal{H}_{\hat{\mathbf{F}}_t}$ as
\vspace{0mm}
\begin{equation}
    \mathcal{H}_{\hat{\mathbf{F}}_t} \coloneqq
    \begin{bmatrix}
        \hat{C}_t\hat{F}_t & \hat{C}_t\hat{\Bar{A}}_t\hat{F}_t & ... & \hat{C}_t\hat{\Bar{A}}_t^{d_2}\hat{F}_t\\
        \hat{C}_t\hat{\Bar{A}}_t\hat{F}_t & \hat{C}_t\hat{\Bar{A}}_t^{2}\hat{F}_t & ... & \hat{C}_t\hat{\Bar{A}}_t^{d_2+1}\hat{F}_t\\
        .\\
        .\\
        .\\
        \hat{C}_t\hat{\Bar{A}}_t^{d_1-1}\hat{F}_t & \hat{C}_t\hat{\Bar{A}}_t^{d_1}\hat{F}_t & ...& \hat{C}_t\hat{\Bar{A}}_t^{H-1}\hat{F}_t
    \end{bmatrix} \in \mathbb{R}^{d_1 n_y \times (d_2 + 1)n_y}.
    \label{eq: Hankel matrix definition}
\end{equation}
Analogously, $\mathcal{H}_{\hat{\mathbf{G}}_t}$ has a similar definition as above for $\{ \hat{C}_t\hat{\Bar{A}}_t^i\hat{B}_t \}_{i = 0}^{H-1}$. For the sake of completeness, we present SYSID here.
\begin{algorithm}[H]
\caption{SYSID \cite{lale2020logarithmic}}
\begin{algorithmic}[1]
\State \textbf{Input:} $\mathbf{\hat{M}}_t$, $H$, $n_x$, $n_y$, $n_u$, $d_1 \geq n_x$, $d_2 \geq n_x$ such that $d_1 + d_2 + 1 = H$
\State Construct two Hankel matrices $\mathcal{H}_{\hat{\mathbf{F}}_t} \in \mathbb{R}^{d_1 n_y \times (d_2 + 1)n_y}$ and $\mathcal{H}_{\hat{\mathbf{G}}_t} \in \mathbb{R}^{d_1 n_y \times (d_2 + 1)n_u}$ from $\hat{\mathbf{F}}_t$ and $\hat{\mathbf{G}}_t$ respectively. Let $ \hat{\mathcal{H}}_t = 
        \begin{bmatrix}
            \mathcal{H}_{\hat{\mathbf{F}}_t} & \mathcal{H}_{\hat{\mathbf{G}}_t}
        \end{bmatrix}.$
\State Obtain $\hat{\mathcal{H}}_t^-$ by discarding $(d_2 + 1)^{\text{th}}$ and $(2d_2 + 2)^{\text{th}}$ block columns of $\hat{\mathcal{H}}_t$.
\State Perform SVD on $\hat{\mathcal{H}}_t^-$, and then obtain $\hat{\mathcal{Q}}_t$, the best $n_x$-rank approximation by setting all but the first $n_x$ singular values to zero.
\State Obtain $\mathbf{U}_t, \mathbf{Z}_t, \mathbf{V}_t = \text{SVD}(\hat{\mathcal{Q}}_t)$.
\State Construct $\mathbf{O}(\hat{\Bar{A}}_t, \hat{C}_t, d_1) = \mathbf{U}_t \mathbf{Z}_t^{1/2}$. \Comment{$\hat{\Bar{A}}_t = \hat{A}_t - \hat{F}_t\hat{C}_t$}
\State Construct $[\mathbf{C}(\hat{\Bar{A}}_t, \hat{F}_t, d_2 + 1), \mathbf{C}(\hat{\Bar{A}}_t, \hat{B}_t, d_2 + 1)] = \mathbf{Z}_t^{1/2}\mathbf{V}_t$.
\State Obtain $\hat{C}_t$ from the first $n_y$ rows of $\mathbf{O}(\hat{\Bar{A}}_t, \hat{C}_t, d_1)$.
\State Obtain $\hat{B}_t$ from the first $n_u$ columns of $\mathbf{C}(\hat{\Bar{A}}_t, \hat{B}_t, d_2 + 1)$.
\State Obtain $\hat{F}_t$ from the first $n_y$ columns of $\mathbf{C}(\hat{\Bar{A}}_t, \hat{F}_t, d_2 + 1)$.
\State Obtain $\hat{\mathcal{H}}_t^+$ by discarding $1^\text{st}$ and $(d_2 + 2)^\text{th}$ block columns of $\hat{\mathcal{H}}_t$.
\State Obtain $\hat{\Bar{A}}_t = \mathbf{O}^\dagger(\hat{\Bar{A}}_t, \hat{C}_t, d_1)\hspace{1mm}\hat{\mathcal{H}}_t^+ \hspace{1mm}[\mathbf{C}(\hat{\Bar{A}}_t, \hat{F}_t, d_2 + 1), \mathbf{C}(\hat{\Bar{A}}_t, \hat{B}_t, d_2 + 1)]^\dagger$.
\State Obtain $\hat{A}_t = \hat{\Bar{A}}_t + \hat{F}_t\hat{C}_t$.
\State Obtain $\hat{L}_t$ from the first $n_x \times n_y$ block of $\hat{A}_t^\dagger \mathbf{O}^\dagger(\hat{\Bar{A}}_t, \hat{C}_t, d_1)\hat{\mathcal{H}}_t^-$.
\State \textbf{Return:} $\hat{A}_t$, $\hat{B}_t$, $\hat{C}_t$, and $\hat{L}_t$.
\end{algorithmic}
\label{algorithm: SYSID}
\end{algorithm}

\subsection*{Truncated closed-loop noise evolution parameter:} 

The truncated closed-loop noise evolution parameter, which captures the effect of noises and excitation signals on the vector $\phi_t$, will play an important role in establishing the persistence of excitation during the adaptive control period. Consider a permutation of the vector $\phi_t$, denoted by $\Bar{\phi}_t$, where
\vspace{0mm}
\begin{equation*}
    \Bar{\phi}_t = 
    \begin{bmatrix}
        y_{t-1}^\top & u_{t-1}^\top &. &. &. & y_{t-H}^\top & u_{t-H}^\top
    \end{bmatrix}^\top \in \mathbb{R}^{(n_y + n_u)H}.
\end{equation*}
Following the warm-up period, recall that the adaptive control policy $u_t = -\hat{K}_t \hat{x}_{t|t,\hat{\Theta}_t} + \eta_t$ is deployed, where $\eta_t \sim \sigma_{\eta_t}\mathcal{N}(0, I)$ and $\sigma^2_{\eta_t} = \left(\gamma/\sqrt{l_k} \right)$ with $\gamma > 0$. Here, $l_k$ is the number of time steps in the $k^\text{th}$ episode. Recall the following relation from \eqref{eq: measurement and time update}:
\vspace{0mm}
\begin{equation*}
    \begin{split}
        \hat{x}_{t|t-1, \hat{\Theta}_{t-1}} & = \left( \hat{A}_{t-1} - \hat{B}_{t-1} \hat{K}_{t-1} \right) \hat{x}_{t-1|t-1, \hat{\Theta}_{t-1}} + \hat{B}_{t-1}\eta_{t-1},\\
        \hat{x}_{t|t, \hat{\Theta}_{t}} & = \hat{x}_{t|t-1, \hat{\Theta}_{t-1}} + \hat{L}_t \left( y_t - \hat{C}_t \hat{x}_{t|t-1,\hat{\Theta}_{t-1}} \right)\\
        & = \left( \hat{A}_{t-1} - \hat{B}_{t-1} \hat{K}_{t-1} \right) \hat{x}_{t-1|t-1, \hat{\Theta}_{t-1}} + \hat{B}_{t-1}\eta_{t-1}\\
        & \hspace{4mm} + \hat{L}_t \Bigg(Cx_t + z_t - \hat{C}_t \left(  \left( \hat{A}_{t-1} - \hat{B}_{t-1} \hat{K}_{t-1} \right) \hat{x}_{t-1|t-1, \hat{\Theta}_{t-1}} + \hat{B}_{t-1}\eta_{t-1} \right) \Bigg)\\
    \end{split}
\end{equation*}
\begin{equation*}
    \begin{split}
        & = \left(I - \hat{L}_t \hat{C}_t \right) \Bigg( \left( \hat{A}_{t-1} - \hat{B}_{t-1} \hat{K}_{t-1} \right) \hat{x}_{t-1|t-1, \hat{\Theta}_{t-1}} + \hat{B}_{t-1}\eta_{t-1} \Bigg)\\
        & \hspace{4mm} + \hat{L}_t \Bigg( C \left( Ax_{t-1} - B \hat{K}_{t-1} \hat{x}_{t-1|t-1, \hat{\Theta}_{t-1}} + B\eta_{t-1}  + w_{t-1} \right) + z_t \Bigg).
    \end{split}
\end{equation*}
Now,
\vspace{0mm}
\begin{equation*}
    \begin{split}
        \begin{bmatrix}
            x_t \\ \hat{x}_{t|t, \hat{\Theta}_t} 
        \end{bmatrix}
        &
        =
        \underbrace{
        \begin{bmatrix}
            A & -B \hat{K}_{t-1}\\
            \hat{L}_t C A & \left(I - \hat{L}_t \hat{C}_t \right) \left( \hat{A}_{t-1} - \hat{B}_{t-1} \hat{K}_{t-1} \right) - \hat{L}_t C B \hat{K}_{t-1} 
        \end{bmatrix}}_{\mathbf{\hat{G}_2}^{(t)}}
        \begin{bmatrix}
            x_{t-1} \\ \hat{x}_{t-1|t-1, \hat{\Theta}_{t-1}} 
        \end{bmatrix}
        \\
        & \hspace{4mm}
        +
        \underbrace{
        \begin{bmatrix}
            I & 0 & B \\
            \hat{L}_t C & \hat{L}_t & \left( I - \hat{L}_t \hat{C}_t \right)\hat{B}_{t-1} + \hat{L}_t C B 
        \end{bmatrix}}_{\mathbf{\hat{G}_3}^{(t)}}
        \begin{bmatrix}
            w_{t-1} \\ z_t \\ \eta_{t-1} 
        \end{bmatrix}.
    \end{split}
\end{equation*}
Let $f_t = \begin{bmatrix} y_t \\ u_t \end{bmatrix}$. Now, we can express $f_t$ as:
\vspace{0mm}
\begin{equation*}
    \begin{split}
        f_t & = 
        \begin{bmatrix}
            CA & -CB \hat{K}_{t-1}\\
            -\hat{K}_t\hat{L}_t C A & -\hat{K}_t\left(I - \hat{L}_t \hat{C}_t \right) \left( \hat{A}_{t-1} - \hat{B}_{t-1} \hat{K}_{t-1} \right) + \hat{K}_t\hat{L}_t C B \hat{K}_{t-1} 
        \end{bmatrix}
        \begin{bmatrix}
            x_{t-1} \\ \hat{x}_{t-1|t-1, \hat{\Theta}_{t-1}} 
        \end{bmatrix}\\
        & \hspace{4mm} + 
        \begin{bmatrix}
            C & I & CB & 0\\
            -\hat{K}_t\hat{L}_t C & -\hat{K}_t\hat{L}_t & -\hat{K}_t\left( I - \hat{L}_t \hat{C}_t \right)\hat{B}_{t-1}  -\hat{K}_t\hat{L}_t C B & I 
        \end{bmatrix}
        \begin{bmatrix}
            w_{t-1} \\ z_t \\ \eta_{t-1} \\ \eta_t
        \end{bmatrix}\\
        & = 
        \underbrace{
        \begin{bmatrix}
            C & 0\\
            0 & -\hat{K}_t
        \end{bmatrix}}_{\mathbf{\hat{\Psi}}_t}
        \mathbf{\hat{G}_2}^{(t)}
        \begin{bmatrix}
            x_{t-1} \\ \hat{x}_{t-1|t-1, \hat{\Theta}_{t-1}} 
        \end{bmatrix} 
        +
        \underbrace{
        \begin{bmatrix}
            C & 0\\
            0 & -\hat{K}_t
        \end{bmatrix}}_{\mathbf{\hat{\Psi}}_t}
        \mathbf{\hat{G}_3}^{(t)}
        \begin{bmatrix}
            w_{t-1} \\ z_t \\ \eta_{t-1} 
        \end{bmatrix}
        +
        \begin{bmatrix}
            0 & I & 0 & 0\\
            0 & 0 & 0 & I 
        \end{bmatrix}
        \begin{bmatrix}
            w_{t-1} \\ z_t \\ \eta_{t-1} \\ \eta_t
        \end{bmatrix}\\
        & = \mathbf{\hat{\Psi}}_t \mathbf{\hat{G}_2}^{(t)}
        \begin{bmatrix}
            x_{t-1} \\ \hat{x}_{t-1|t-1, \hat{\Theta}_{t-1}} 
        \end{bmatrix}
        +
        \mathbf{\hat{\Psi}}_t\mathbf{\hat{G}_3}^{(t)}
        \begin{bmatrix}
            w_{t-1} \\ z_t \\ \eta_{t-1} 
        \end{bmatrix}
        + 
        \begin{bmatrix}
            0 & I & 0 & 0\\
            0 & 0 & 0 & I 
        \end{bmatrix}
        \begin{bmatrix}
            w_{t-1} \\ z_t \\ \eta_{t-1} \\ \eta_t
        \end{bmatrix}.\\
    \end{split}
\end{equation*}
For the ease of representation, let us denote $\Bar{x}_t = \begin{bmatrix} x_t^\top & \hat{x}_{t|t, \hat{\Theta}_t}^\top \end{bmatrix}^\top$ and $\eta^e_{t} =  \begin{bmatrix}
        w_{t}^\top & z_{t+1}^\top & \eta_{t}^\top 
    \end{bmatrix}^\top$. Rolling back in time $\Bar{x}_t$ for $H$ time steps, we obtain the following:
\vspace{0mm}
\begin{equation*}
\begin{split}
    \Bar{x}_t = \prod_{i = 0}^H \mathbf{\hat{G}_2}^{(t - i)} \Bar{x}_{t-H-1} + \sum_{i = 1}^{H+1} \left( \prod_{j=2}^i \mathbf{\hat{G}_2}^{(t-j+2)} \right) \mathbf{\hat{G}_3}^{(t-i+1)}\eta^e_{t-i}.
\end{split}
\label{eq: truncated CL noise evolution - eq1}
\end{equation*}
Now, let us expand $\Bar{x}_{t-H-1}$.
\vspace{0mm}
\begin{equation*}
    \begin{split}
        \Bar{x}_{t-H-1} & = \prod_{i = H+1}^{t-1} \mathbf{\hat{G}_2}^{(t - i)} 
        \begin{bmatrix}
        x_0\\ \hat{L}_0 C x_0 + \hat{L}_0 z_0
        \end{bmatrix}
        + \sum_{i = H+2}^{t} \left( \prod_{j=H+3}^i \mathbf{\hat{G}_2}^{(t-j+2)} \right) \mathbf{\hat{G}_3}^{(t-i+1)} \eta^e_{t-i}.\\
    \end{split}
    \label{eq: truncated CL noise evolution - eq2}
\end{equation*}
The equality comes from the assumption that $\hat{x}_{0|-1,\hat{\Theta}_{-1}} = 0$. Further, note that $x_0 \sim \mathcal{N}(0, \Sigma)$. Therefore, $\Bar{x}_{t-H-1}$ represents the effect of $\begin{bmatrix} w_{i-1}^\top & z_i^\top & \eta_{i-1}^\top \end{bmatrix}^\top$ for $0 \leq i \leq t - H -2$, which are independent with respect to the time. Now $f_t$ can be rolled back $H$ time steps backwards, as follows:
\vspace{0mm}
\begin{equation*}
    \begin{split}
        f_t & = \mathbf{\hat{\Psi}}_t\Bar{x}_t + 
        \begin{bmatrix}
            0 & I & 0 & 0\\
            0 & 0 & 0 & I 
        \end{bmatrix}
        \begin{bmatrix}
            w_{t-1} \\ z_t \\ \eta_{t-1} \\ \eta_t
        \end{bmatrix}\\
        & = \mathbf{\hat{\Psi}}_t \left(\prod_{i = 0}^H \mathbf{\hat{G}_2}^{(t - i)} \Bar{x}_{t-H-1} + \sum_{i = 1}^{H+1} \left( \prod_{j=2}^i \mathbf{\hat{G}_2}^{(t-j+2)} \right) \mathbf{\hat{G}_3}^{(t-i+1)}\eta^e_{t-i} \right) + 
        \begin{bmatrix}
            0 & I & 0 & 0\\
            0 & 0 & 0 & I 
        \end{bmatrix}
        \begin{bmatrix}
            w_{t-1} \\ z_t \\ \eta_{t-1} \\ \eta_t
        \end{bmatrix}\\
        & = \mathbf{\hat{\Psi}}_t \sum_{i = 2}^{H+1} \left( \prod_{j=2}^i \mathbf{\hat{G}_2}^{(t-j+2)} \right) \mathbf{\hat{G}_3}^{(t-i+1)}\eta^e_{t-i}\\
        & \hspace{4mm}+ 
        \underbrace{\begin{bmatrix}
            C & I & CB & 0\\
            -\hat{K}_t\hat{L}_t C & -\hat{K}_t\hat{L}_t & -\hat{K}_t\left( I - \hat{L}_t \hat{C}_t \right)\hat{B}_{t-1}  -\hat{K}_t\hat{L}_t C B & I 
        \end{bmatrix}}_{\mathbf{\hat{G}_1}^{(t)}}
        \begin{bmatrix}
            w_{t-1} \\ z_t \\ \eta_{t-1} \\ \eta_t
        \end{bmatrix} + \mathbf{r}_t^\mathbf{c},
    \end{split}
\end{equation*}
where $\mathbf{r}^\mathbf{c}_t$ is the residual vector that represents the effect of $\begin{bmatrix} w_{i-1}^\top & z_i^\top & \eta_{i-1}^\top \end{bmatrix}^\top$ for $0 \leq i \leq t - H -2$,  which are independent with respect to the time. For ease of representation, let $\varrho_{t} = \begin{bmatrix} w_{t}^\top & z_{t+1}^\top & \eta_{t}^\top & \eta_{t+1}^\top \end{bmatrix}^\top$. Now $f_t$ can be compactly represented as such:
\vspace{0mm}
\begin{equation*}
    \begin{split}
        f_t = \mathbf{\Bar{G}}_{t}
        \begin{bmatrix}
            {\varrho_{t-1}}^\top & {\eta^e_{t-2}}^\top & \dots & {\eta^e_{t-H-1}}^\top
        \end{bmatrix}^\top
        + \mathbf{r}^\mathbf{c}_t,
    \end{split}
\end{equation*}
where
\vspace{0mm}
\begin{equation*}
\begin{split}
    \mathbf{\Bar{G}}_{t} & = 
    \begin{bmatrix}
        \mathbf{\hat{G}_1}^{(t)} & \mathbf{\hat{\Psi}}_t \mathbf{\hat{G}_2}^{(t)} \mathbf{\hat{G}_3}^{(t-1)} & \mathbf{\hat{\Psi}}_t \mathbf{\hat{G}_2}^{(t)} \mathbf{\hat{G}_2}^{(t-1)} \mathbf{\hat{G}_3}^{(t-2)} & ... & \mathbf{\hat{\Psi}}_t \mathbf{\hat{G}_2}^{(t)} \mathbf{\hat{G}_2}^{(t-1)} ... \mathbf{\hat{G}_2}^{(t-H+1)}\mathbf{\hat{G}_3}^{(t-H)} 
    \end{bmatrix}\\
    & \in \mathbb{R}^{(n_y + n_u) \times (H+1)(n_x + n_y)}.
\end{split}
\end{equation*}
We can now represent $\Bar{\phi}_t$ as follows:
\vspace{0mm}
\begin{equation} \label{eq:phi_def_shi}
    \begin{split}
        \Bar{\phi}_t & = 
        \begin{bmatrix}
            f_{t-1}\\.\\.\\.\\f_{t-H}
        \end{bmatrix}
        =
        \mathcal{G}^\text{cl}_t 
        \begin{bmatrix}
            \varrho_{t-2}\\ \eta^e_{t-3} \\.\\.\\.\\ \eta^e_{t-2H-1}
        \end{bmatrix}
        +
        \begin{bmatrix}
            \mathbf{r}_{t-1}^\mathbf{c}\\.\\.\\.\\\mathbf{r}_{t-H}^\mathbf{c}
        \end{bmatrix},
    \end{split}
\end{equation}
where 
\begin{equation*}
    \mathcal{G}^\text{cl}_t = 
    \begin{bmatrix}
        \mathbf{\Bar{G}}_{t-1} & 0 & 0 & 0 & ...\\
        0 & \mathbf{\Bar{G}}_{t-2} & 0 & 0 & ...\\
        .& & & &\\
        .& & & &\\
        .& & & &\\
        0 &  0 &  0 & ...& \mathbf{\Bar{G}}_{t-H}
    \end{bmatrix}.
    \label{eq: closed loop noise evolution parameter}
\end{equation*}
If the true model parameter $\Theta$ is known, the optimal control policy can be deployed. Then, $\mathcal{G}^{cl}$ captures the effect of the process and measurement noises as well as excitation signals on $\Bar{\phi}_t$ while using the optimal control policy:
\begin{equation*}
    \mathcal{G}^\text{cl} = 
    \begin{bmatrix}
        \mathbf{\Bar{G}} & 0 & 0 & 0 & ...\\
        0 & \mathbf{\Bar{G}} & 0 & 0 & ...\\
        .& & & &\\
        .& & & &\\
        .& & & &\\
        0 &  0 &  0 & ...& \mathbf{\Bar{G}}
    \end{bmatrix},
    \label{eq: closed loop noise evolution true parameter}
\end{equation*}
where
\begin{equation*}
    \mathbf{\Bar{G}} = 
    \begin{bmatrix}
        \mathbf{G_1} & \mathbf{\Psi} \mathbf{G_2} \mathbf{G_3} & \mathbf{\Psi} \mathbf{G_2}^2 \mathbf{G_3} & ... & \mathbf{\Psi} \mathbf{G_2}^{H}\mathbf{G_3} 
    \end{bmatrix} \in \mathbb{R}^{(n_y + n_u) \times (H+1)(n_x + n_y)}
\end{equation*}
with
\begin{equation*}
\begin{split}
    & \mathbf{G_1} = 
    \begin{bmatrix}
        C & I & CB & 0\\
        -KL C & -KL & -K\left( I - L C \right)B  -K L C B & I 
    \end{bmatrix},
    \hspace{4mm}
    \mathbf{\Psi} = 
    \begin{bmatrix}
        C & 0\\
        0 & -K
    \end{bmatrix},\\
    & \mathbf{G_2} = 
    \begin{bmatrix}
        A & -B K\\
        L C A & \left(I - L C \right) \left( A - B K \right) - L C B K 
    \end{bmatrix}, \hspace{4mm} \mathbf{G_3} = 
    \begin{bmatrix}
        I & 0 & B\\
        L C & L & \left( I - L C \right)B + L C B
    \end{bmatrix}.
\end{split}
\end{equation*}
With $H$ chosen such that $\mathbf{\Bar{G}}$ is full-row rank, $\mathcal{G}^{cl}$ can be shown to have full-row rank with QR decomposition \cite{lale2020regret}. Therefore, $\sigma_\text{min}\left( \mathcal{G}^\text{cl} \right) \geq \sigma_\text{c}$ for some positive constant $\sigma_\text{c}$.

\subsection{Proofs of Lemmas \ref{Lemma: bound on the states and outputs during LBC phase} - \ref{Lemma: Markov parameter estimation error during LBC phase}:} 

Recall that LQG-NAIVE (refer to Algorithm \ref{algorithm: LQG-NAIVE}) proceeds in an episodic fashion. Therefore, the proofs of Lemmas \ref{Lemma: bound on the states and outputs during LBC phase} - \ref{Lemma: Markov parameter estimation error during LBC phase} are also analyzed episode-wise. The proofs of the three lemmas are proven jointly based on an inductive argument, where we will first analyze the $0^\text{th}$ episode, i.e., $t \in [T_\text{w}, 2T_\text{w}-1]$, which consists of $T_\text{w}$ time steps. During the $0^\text{th}$, it must be noted that $\hat{\Theta}_t = \hat{\Theta}_{T_\text{w}}$. Therefore, the event $\mathcal{E}_{\mathbf{M}_{T_\text{w}}}$ holds, where 
\begin{equation} \mathcal{E}_{\mathbf{M}_{T_\text{w}}} \coloneqq \left\{||\hat{\mathbf{M}}_{t} - \mathbf{M}|| \leq ||\hat{\mathbf{M}}_{T_\text{w}} - \mathbf{M}|| \leq 1, \text{for } t\in [T_\text{w}, 2T_\text{w}-1] \text{ and }T_\text{w} \geq \Bar{T}^\text{alt,1}_\text{w} \right\}, \label{eq:event_def_shi}
\end{equation}
which holds with a probability of at least $1 - \delta$ and $\Bar{T}^\text{alt,1}_\text{w}$ is as defined in \eqref{eq: Alt LB on T_w}. This event is a consequence of Lemma \ref{Lemma: Markov parameter estimation error after warm-up} for the warm-up phase. Given this event, we have

\begin{enumerate}
    \item $\Theta \in \mathcal{C}_{A}(t) \times \mathcal{C}_{B}(t) \times \mathcal{C}_{C}(t) \times \mathcal{C}_{L}(t)$ for all $t \in [T_\text{w}, 2T_\text{w}-1]$, where $\mathcal{C}_{A}(t), \mathcal{C}_{B}(t), \mathcal{C}_{C}(t) \text{ and } \mathcal{C}_{L}(t)$ are defined in Lemma \ref{Lemma: confidence bound on the system parameters}.
    \item $|| \hat{C}_t - C||, || \hat{B}_t -  B ||, || \hat{F}_t -  F || \leq  \beta_B(T_\text{w}) = 1$ when $T_\text{w} \geq T_B$.
    \item  $ \left| \left| \hat{A}_t -  A  \right| \right| \leq \beta_A(T_\text{w}) =\sigma_{n_x}(A)/2$ when $T_\text{w} \geq T_A$.
    \item $ \left| \left| \hat{L}_t -  L  \right| \right| \leq \beta_L(T_\text{w})$ for some $\beta_L(T_\text{w}) > 0$.
\end{enumerate}
Since the input-output behavior of a system with parameter $\Theta$ and its similarity transformation are the same, without loss of generality, we assume that the similarity transformation matrix $\mathbf{T} = I$. The above bounds are a consequence of Corollary \ref{Corollary: confidence bound on the system parameters after warm-up}.

Firstly, we will show that the system signals are `well-controlled', i.e., the bound on the norm of the system signals grows at most polynomial with respect to the time step $t$.

\subsection*{Proof of Lemma \ref{Lemma: bound on the states and outputs during LBC phase} during the $0^\text{th}$ episode:}

This proof is an extension of an earlier result \cite[Lemma 4.1]{lale2021adaptive}, where the extension requires accounting for the additive Gaussian excitatory signals. Some of the steps in this proof are identical to the proof of \cite[Lemma 4.1]{lale2021adaptive} and those identical steps are presented here for the sake of completeness with citations. 

Based on \eqref{eq: measurement and time update}, consider the following decomposition of $\hat{x}_{t|t,\hat{\Theta}_t}$:
\begin{equation*}
    \begin{split}
        \hat{x}_{t|t,\hat{\Theta}_t} & = \hat{x}_{t|t-1,\hat{\Theta}_{t-1}} + \hat{L}_t(y_t - \hat{C}_t\hat{x}_{t|t-1,\hat{\Theta}_{t-1}})\\
        & = \hat{A}_{t-1}\hat{x}_{t-1|t-1,\hat{\Theta}_{t-1}} - \hat{B}_{t-1}\hat{K}_{t-1}\hat{x}_{t-1|t-1,\hat{\Theta}_{t-1}}\\
        & \hspace{4mm} + \hat{B}_{t-1}\eta_{t-1} + \hat{L}_t\Bigg(y_t - \hat{C}_t\left(\hat{A}_{t-1}\hat{x}_{t-1|t-1,\hat{\Theta}_{t-1}} - \hat{B}_{t-1}\hat{K}_{t-1}\hat{x}_{t-1|t-1,\hat{\Theta}_{t-1}} + \hat{B}_{t-1}\eta_{t-1}\right)\Bigg)\\
        & = \left(I - \hat{L}_t\hat{C}_t\right)\left( \hat{A}_{t-1} - \hat{B}_{t-1}\hat{K}_{t-1}\right)\hat{x}_{t-1|t-1,\hat{\Theta}_{t-1}} + \left(I - \hat{L}_t\hat{C}_t\right)\hat{B}_{t-1}\eta_{t-1}\\
        & \hspace{4mm} + \hat{L}_t \left( Cx_t + z_t - C\hat{x}_{t|t-1,\hat{\Theta}_{t-1}} + C\hat{x}_{t|t-1,\hat{\Theta}_{t-1}}\right)\\
        & = \left(I - \hat{L}_t\hat{C}_t\right)\left( \hat{A}_{t-1} - \hat{B}_{t-1}\hat{K}_{t-1}\right)\hat{x}_{t-1|t-1,\hat{\Theta}_{t-1}} + \left(I - \hat{L}_t\hat{C}_t\right)\hat{B}_{t-1}\eta_{t-1}\\
        & \hspace{4mm} + \hat{L}_t \left( Cx_t - C\hat{x}_{t|t-1,\hat{\Theta}_{t-1}} + C\left( \hat{A}_{t-1} - \hat{B}_{t-1}\hat{K}_{t-1} \right)\hat{x}_{t-1|t-1,\hat{\Theta}_{t-1}} + C\hat{B}_{t-1}\eta_{t-1} + z_t\right)\\
        & = \Bigg( \hat{A}_{t-1} - \hat{B}_{t-1}\hat{K}_{t-1}  - \hat{L}_t \left( \hat{C}_t\hat{A}_{t-1} - \hat{C}_t\hat{B}_{t-1}\hat{K}_{t-1} - C\hat{A}_{t-1} + C\hat{B}_{t-1}\hat{K}_{t-1} \right)\Bigg)\hat{x}_{t-1|t-1,\hat{\Theta}_{t-1}}\\
        & \hspace{4mm} + \left(I - \hat{L}_t\hat{C}_t\right)\hat{B}_{t-1}\eta_{t-1} + \hat{L}_{t}C\left( x_t - \hat{x}_{t|t-1,\hat{\Theta}_{t-1}} + \hat{x}_{t|t-1,\Theta} - \hat{x}_{t|t-1,\Theta} \right) + \hat{L}_tC\hat{B}_{t-1}\eta_{t-1} + \hat{L}_tz_t\\
    \end{split}
\end{equation*}
\begin{equation}
    \begin{split}
        & = \underbrace{\Bigg( \hat{A}_{t-1} - \hat{B}_{t-1}\hat{K}_{t-1}  - \hat{L}_t \left( \hat{C}_t\hat{A}_{t-1} - \hat{C}_t\hat{B}_{t-1}\hat{K}_{t-1} - C\hat{A}_{t-1} + C\hat{B}_{t-1}\hat{K}_{t-1} \right)\Bigg)}_{\text{Can be thought of as the dynamics}}\hat{x}_{t-1|t-1,\hat{\Theta}_{t-1}}\\
        & \hspace{-4mm} + \underbrace{\hat{L}_{t}C\left( x_t - \hat{x}_{t|t-1,\Theta} \right) + \hat{L}_{t}C\left( \hat{x}_{t|t-1,\Theta} - \hat{x}_{t|t-1,\hat{\Theta}_{t-1}} \right) + \hat{B}_{t-1}\eta_{t-1} + \hat{L}_tz_t + \hat{L}_t\left( C - \hat{C}_t \right)\hat{B}_{t-1}\eta_{t-1}}_{\text{Can be thought of as a process noise}}.\\
    \end{split}
    \label{eq: expanding x_hat_t}
\end{equation}
We will upper bound $||\hat{x}_{t|t,\hat{\Theta}_t}||$ by upper bounding the norm of each of the above terms in the decomposition. Under the event $\mathcal{E}_{\mathbf{M}_{T_\text{w}}}$, we have the following bound, which holds for all $t \in [T_\text{w}, 2T_\text{w}-1]$:
\vspace{0mm}
\begin{equation}
    \begin{split}
        ||\hat{L}_t|| & \leq ||\hat{L}_t - L|| + ||L||\\
        & \leq \beta_L(T_\text{w}) + ||L|| \coloneqq \zeta.
    \end{split}
    \label{eq: bound on L}
\end{equation}
From Theorem \ref{Theorem: perturbation bound on K}, we have that there exists a time step $T_K > H$ such that for $T_\text{w} \geq T_K$, $\max\{||\hat{A}_t - A||, ||\hat{B}_t - B||\} \leq \epsilon_K$ for all $t \in [T_\text{w}, 2T_\text{w}-1]$. Then, $||\hat{K}_t - K|| \leq \beta_K(T_\text{w})$ for all $t \in [T_\text{w}, 2T_\text{w}-1]$, where $\beta_K(T_\text{w})$ is a constant dependent on the solution of DARE \eqref{eq: DARE to compute K}. Therefore, we have the following bound, which holds for all $t \in [T_\text{w}, 2T_\text{w}-1]$ with $T_\text{w} \geq T_K$ under the event $\mathcal{E}_{\mathbf{M}_{T_\text{w}}}$:
\begin{equation}
    \begin{split}
        ||\hat{K}_t|| & \leq ||\hat{K}_t - K|| + ||K||\\
        & \leq \beta_K(T_\text{w}) + ||K|| \coloneqq \Gamma.
    \end{split}
    \label{eq: bound on K}
\end{equation}

\subsection*{Bounding the norm of the `dynamics' term:}
Let
\vspace{0mm}
\begin{equation*}
    \begin{split}
        N_t = \hat{A}_{t-1} - \hat{B}_{t-1}\hat{K}_{t-1} - \hat{L}_t \left( \hat{C}_t - C \right) \left( \hat{A}_{t-1} - \hat{B}_{t-1}\hat{K}_{t-1} \right).
    \end{split}
\end{equation*}
From Lemma \ref{Lemma: decay of closed loop matrix}, there exists a time step $T_\text{cl} > H$ such that for $T_\text{w} \geq T_\text{cl}$, $\max\{||\hat{A}_t - A||, ||\hat{B}_t - B||\} \leq \epsilon_N$ for all $t \in [T_\text{w}, 2T_\text{w}-1]$. Then, under the event $\mathcal{E}_{\mathbf{M}_{T_\text{w}}}$, $||(\hat{A}_t - \hat{B}_t \hat{K}_t)^\tau|| \leq \kappa_\text{cl} \varphi^\tau_\text{cl}$ for any $\tau \in \mathbb{Z}^+$ and for all $t \in [T_\text{w}, 2T_\text{w}-1]$ and $T_\text{w} \geq T_\text{cl}$, where $\kappa_\text{cl} \geq 1$ and $\varphi_\text{cl} \in (0,1)$ with $\varphi_\text{cl} \geq \rho(\hat{A}_t - \hat{B}_t \hat{K}_t)$. Now from Lemma \ref{Lemma: general matrix perturbation bound}, we have 
\begin{equation}
    \begin{split}
        \left| \left| N_t^\tau \right| \right| & \leq \kappa_\text{cl} \left( \kappa_\text{cl} \left| \left| \hat{L}_t \left( \hat{C}_t - C \right) \left( \hat{A}_{t-1} - \hat{B}_{t-1}\hat{K}_{t-1} \right) \right| \right| + \varphi_\text{cl} \right)^\tau\\
        & \leq \kappa_\text{cl} \varphi_N^\tau,
    \end{split}
    \label{eq: bound on the dynamics term of x_hat_t}
\end{equation}
where $\varphi_N \coloneqq \kappa_\text{cl}^2 \zeta \beta_C(T_\text{w}) \varphi_\text{cl} + \varphi_\text{cl}$ and $\tau \in \mathbb{Z}^+$. There exists a time step $T_N > H$ such that, for all $t \in [T_\text{w}, 2T_\text{w}-1]$ and $T_\text{w} \geq T_N$, $\varphi_N \in (0,1)$.

\subsection*{Bounding the norm of the `process noise' term:}

The term $\hat{L}_{t}C\left( x_t - \hat{x}_{t|t-1,\Theta} \right) + \hat{L}_tz_t$ is a $\zeta \left( ||C|| ||\Sigma||^{1/2} + \sigma_z \right)$ - sub-Gaussian random variable \cite[Lemma 4.1]{lale2020regret}. Therefore, from Lemma \ref{Lemma B.1}, we get:
\vspace{0mm}
\begin{equation}
    \begin{split}
        \left| \left| \hat{L}_{t}C\left( x_t - \hat{x}_{t|t-1,\Theta} \right) + \hat{L}_tz_t \right| \right| \leq \zeta \left( ||C|| ||\Sigma||^{1/2} + \sigma_z \right) \sqrt{2n_x \log \left(\frac{2n_xt}{\delta}\right)},
    \end{split}
    \label{eq: bound on process noise term of  x_hat_t - 1}
\end{equation}
which holds with a probability of at least $1- \delta/t$. The constant $\zeta$ comes from \eqref{eq: bound on L}. Re-parameterise $\delta/t \rightarrow \delta$. Now, under $\mathcal{E}_{\mathbf{M}_{T_\text{w}}}$, we have:
\vspace{0mm}
\begin{equation*}
    \begin{split}
        \left| \left| \hat{B}_{t-1}\eta_{t-1} + \hat{L}_t\left( C - \hat{C}_t \right)\hat{B}_{t-1}\eta_{t-1} \right| \right| & \leq \left( \left| \left| B \right| \right| + \left| \left| \hat{B}_{t-1} - B \right| \right| \right) ||\eta_{t-1}|| \left( 1 + \left| \left| \hat{L}_t (\hat{C}_t - C)  \right| \right| \right)\\
        & \leq \left(  \left| \left| B \right| \right| + 1 \right) ||\eta_{t-1}|| \left( 1 + \zeta \right).
     \end{split}
\end{equation*}
Recall from \eqref{eq: LBC policy naive} that $\eta_t \sim \sigma_{\eta_t}\mathcal{N}\left(0, I \right)$, where $\sigma^2_{\eta_t} = \frac{\gamma}{\sqrt{l_k}}$ with $\gamma > 0$. Therefore, from Lemma \ref{Lemma B.1}, we have
\vspace{0mm}
\begin{equation}
    \begin{split}
        \left| \left| \hat{B}_{t-1}\eta_{t-1} + \hat{L}_t\left( C - \hat{C}_t \right)\hat{B}_{t-1}\eta_{t-1} \right| \right| & \leq \sigma_{\eta_t} \left(  \left| \left| B \right| \right| + 1 \right) \left( 1 + \zeta \right) \sqrt{2n_u\log \left( \frac{2n_ut}{\delta}\right)}\\
        & \leq \sqrt{\gamma} \left(  \left| \left| B \right| \right| + 1 \right) \left( 1 + \zeta \right) \sqrt{2n_u\log \left( \frac{2n_ut}{\delta}\right)},\\
     \end{split}
     \label{eq: bound on process noise term of  x_hat_t - 2}
\end{equation}
which holds with a probability of at least $1 - \delta/t$ under the event $\mathcal{E}_{\mathbf{M}_{T_\text{w}}}$. Re-parameterise $\delta/t \rightarrow \delta$.  Finally, we bound the spectral norm of the term $\Delta_t = \left( \hat{x}_{t|t-1,\Theta} - \hat{x}_{t|t-1,\hat{\Theta}_{t-1}} \right)$.

\subsection*{Bounding $||\Delta_t||$:}

From \eqref{eq: measurement and time update}, we have the following decompositions:
\begin{equation*}
    \begin{split}
        \hat{x}_{t+1|t, \Theta} & = A \hat{x}_{t|t, \Theta} + B u_t\\
        & = A \hat{x}_{t|t, \Theta} - B \hat{K}_t \hat{x}_{t|t, \hat{\Theta}_t} + B \eta_t\\
        & = A \hat{x}_{t|t, \Theta} - B \hat{K}_t \hat{x}_{t|t, \Theta} + B \hat{K}_t \hat{x}_{t|t, \Theta} - B \hat{K}_t \hat{x}_{t|t, \hat{\Theta}_t} + B \eta_t\\
        & = \left( A - B \hat{K}_t\right)\hat{x}_{t|t, \Theta} - B \hat{K}_t \left(\hat{x}_{t|t, \hat{\Theta}_t} - \hat{x}_{t|t, \Theta} \right) + B \eta_t.\\
        \hat{x}_{t+1|t, \hat{\Theta}_t} & = \hat{A}_t \hat{x}_{t|t, \hat{\Theta}_t} + \hat{B}_t u_t\\
        & = \hat{A}_t \hat{x}_{t|t, \hat{\Theta}_t} - \hat{B}_t \hat{K}_t \hat{x}_{t|t, \hat{\Theta}_t} + \hat{B}_t \eta_t\\
        & = \left( \hat{A}_t + A - A \right) \hat{x}_{t|t, \hat{\Theta}_t} + \left( - \hat{B}_t \hat{K}_t + B \hat{K}_t - B \hat{K}_t \right)\hat{x}_{t|t, \hat{\Theta}_t} + \hat{B}_t \eta_t\\
        & = \left( \hat{A}_t + A - A - \hat{B}_t \hat{K}_t + B \hat{K}_t - B \hat{K}_t \right) \hat{x}_{t|t, \hat{\Theta}_t} + \hat{B}_t \eta_t\\
        & = \left( \hat{A}_t + A - A - \hat{B}_t \hat{K}_t + B \hat{K}_t - B \hat{K}_t \right) \left( \hat{x}_{t|t, \hat{\Theta}_t} -\hat{x}_{t|t, \Theta} + \hat{x}_{t|t, \Theta} \right) + \hat{B}_t \eta_t\\
        & = \underbrace{\left( \hat{A}_t - A - \hat{B}_t \hat{K}_t + B \hat{K}_t \right)}_{\delta_{\hat{\Theta}_t}}\hat{x}_{t|t, \Theta} + \left(A - B\hat{K}_t \right) \hat{x}_{t|t, \Theta}\\
        & \hspace{4mm} +  \left( \hat{A}_t - A - \hat{B}_t \hat{K}_t + B \hat{K}_t \right) \left(\hat{x}_{t|t, \hat{\Theta}_t} - \hat{x}_{t|t, \Theta} \right) + \left(A - B\hat{K}_t \right)\left(\hat{x}_{t|t, \hat{\Theta}_t} - \hat{x}_{t|t, \Theta} \right)  + \hat{B}_t \eta_t.\\
    \end{split}
\end{equation*}
Thus $\Delta_{t+1}$ is,
\vspace{0mm}
\begin{equation*}
    \begin{split}
        & \Delta_{t+1} = \hat{x}_{t+1|t, \Theta} - \hat{x}_{t+1|t, \hat{\Theta}_t}\\
        & = \left( A - B \hat{K}_t\right)\hat{x}_{t|t, \Theta} - B \hat{K}_t \left(\hat{x}_{t|t, \hat{\Theta}_t} - \hat{x}_{t|t, \Theta} \right) + B \eta_t - \delta_{\hat{\Theta}_t}\hat{x}_{t|t, \Theta} - \delta_{\hat{\Theta}_t}\left(\hat{x}_{t|t, \hat{\Theta}_t} - \hat{x}_{t|t, \Theta} \right)\\
        & \hspace{4mm} -\left(A - B\hat{K}_t \right)\left(\hat{x}_{t|t, \hat{\Theta}_t} - \hat{x}_{t|t, \Theta} \right) - \left(A - B\hat{K}_t \right) \hat{x}_{t|t, \Theta} - \hat{B}_t \eta_t\\
        & = A\left( \hat{x}_{t|t, \Theta} - \hat{x}_{t|t, \hat{\Theta}_t} \right)  - \delta_{\hat{\Theta}_t}\hat{x}_{t|t, \Theta} + \delta_{\hat{\Theta}_t}\left( \hat{x}_{t|t, \Theta} -\hat{x}_{t|t, \hat{\Theta}_t} \right) + B \eta_t - \hat{B}_t \eta_t\\
        & = \left( A +  \delta_{\hat{\Theta}_t}\right)\left( \hat{x}_{t|t, \Theta} - \hat{x}_{t|t, \hat{\Theta}_t} \right) - \delta_{\hat{\Theta}_t}\hat{x}_{t|t, \Theta} + \left(B - \hat{B}_t\right) \eta_t.\\
    \end{split}
\end{equation*}
We will now decompose the term $\left( \hat{x}_{t|t, \Theta} - \hat{x}_{t|t, \hat{\Theta}_t} \right)$. From \eqref{eq: measurement and time update}, we have:
\vspace{0mm}
\begin{equation*}
    \begin{split}
        \hat{x}_{t|t, \Theta} - \hat{x}_{t|t, \hat{\Theta}_t} & = \left(I - \hat{L}_t \hat{C}_t \right) \Delta_t + \left( L - \hat{L}_t \right)e_t + \hat{L}_t \left( \hat{C}_t - C \right)\hat{x}_{t|t-1, \Theta}.
    \end{split}
\end{equation*}
\vspace{0mm}
Now, substituting the above expansion into $\Delta_{t+1}$, we get
\begin{equation*}
    \begin{split}
        & \Delta_{t+1}\\
        & = \left( A +  \delta_{\hat{\Theta}_t}\right) \Bigg( \left(I - \hat{L}_t \hat{C}_t \right) \Delta_t + \left( L - \hat{L}_t \right)e_t + \hat{L}_t \left( \hat{C}_t - C \right)\hat{x}_{t|t-1, \Theta} \Bigg) - \delta_{\hat{\Theta}_t}\hat{x}_{t|t, \Theta} + \left(B - \hat{B}_t\right) \eta_t\\
        & = \left( A +  \delta_{\hat{\Theta}_t}\right) \left(I - \hat{L}_t \hat{C}_t \right) \Delta_t + \left( A +  \delta_{\hat{\Theta}_t}\right)\left( L - \hat{L}_t \right)e_t + \left( A +  \delta_{\hat{\Theta}_t}\right) \hat{L}_t \left( \hat{C}_t - C \right)\hat{x}_{t|t-1, \Theta}\\
        &\hspace{4mm}  - \delta_{\hat{\Theta}_t}\hat{x}_{t|t, \Theta} + \left(B - \hat{B}_t\right) \eta_t\\
        & = \left( A +  \delta_{\hat{\Theta}_t}\right) \left(I - \hat{L}_t \hat{C}_t \right) \Delta_t + \left( A +  \delta_{\hat{\Theta}_t}\right)\left( L - \hat{L}_t \right)e_t + \left( A +  \delta_{\hat{\Theta}_t}\right) \hat{L}_t \left( \hat{C}_t - C \right)\hat{x}_{t|t-1, \Theta}\\
        &\hspace{4mm} - \delta_{\hat{\Theta}_t}\hat{x}_{t|t-1, \Theta} - \delta_{\hat{\Theta}_t}Le_t + \left(B - \hat{B}_t\right) \eta_t\\
    \end{split}
\end{equation*}
\vspace{0mm}
\begin{equation}
    \begin{split}
        & = \left( A +  \delta_{\hat{\Theta}_t}\right) \left(I - \hat{L}_t \hat{C}_t \right) \Delta_t + \left(AL - A\hat{L}_t + \delta_{\hat{\Theta}_t}L - \delta_{\hat{\Theta}_t} \hat{L}_t - \delta_{\hat{\Theta}_t}L \right)e_t\\
        & \hspace{4mm} + \Bigg( \left( A +  \delta_{\hat{\Theta}_t}\right) \hat{L}_t \left( \hat{C}_t - C \right) - \delta_{\hat{\Theta}_t} \Bigg)\hat{x}_{t|t-1, \Theta} + \left(B - \hat{B}_t\right) \eta_t\\
         & = \left( A +  \delta_{\hat{\Theta}_t}\right) \left(I - \hat{L}_t \hat{C}_t \right) \Delta_t + \Bigg(A \left( L - \hat{L}_t \right) - \delta_{\hat{\Theta}_t}\hat{L}_t \Bigg)e_t + \Bigg( \left( A +  \delta_{\hat{\Theta}_t}\right) \hat{L}_t \left( \hat{C}_t - C \right) - \delta_{\hat{\Theta}_t} \Bigg)\hat{x}_{t|t-1, \Theta}\\
         & \hspace{4mm} + \left(B - \hat{B}_t\right) \eta_t\\
        & = \sum_{i = 0}^t \prod_{j = 0}^{t-i-1} \Bigg( \left( A +  \delta_{\hat{\Theta}_{t-j}}\right) \left(I - \hat{L}_{t-j} \hat{C}_{t-j} \right) \Bigg) \Bigg(A \left( L - \hat{L}_i \right) - \delta_{\hat{\Theta}_i}\hat{L}_i \Bigg)e_i\\
        & \hspace{4mm} + \sum_{i = 1}^t \prod_{j = 0}^{t-i-1} \Bigg( \left( A +  \delta_{\hat{\Theta}_{t-j}}\right) \left(I - \hat{L}_{t-j} \hat{C}_{t-j} \right) \Bigg) \Bigg( \left( A +  \delta_{\hat{\Theta}_i}\right) \hat{L}_i \left( \hat{C}_i - C \right) - \delta_{\hat{\Theta}_i} \Bigg)\hat{x}_{i|i-1, \Theta}\\
        & \hspace{4mm} + \sum_{i = 0}^t \prod_{j = 0}^{t-i-1} \Bigg( \left( A +  \delta_{\hat{\Theta}_{t-j}}\right) \left(I - \hat{L}_{t-j} \hat{C}_{t-j} \right) \Bigg) \left(B - \hat{B}_i\right) \eta_i. \\
    \end{split}
    \label{eq: state estimation and output bound - 1}
\end{equation}
The last equality in \eqref{eq: state estimation and output bound - 1} comes from the assumption that $\hat{x}_{0|-1,\Theta} = \hat{x}_{0|-1,\hat{\Theta}_{-1}} = 0$. Let us now decompose $\hat{x}_{i|i-1, \Theta}$:
\vspace{0mm}
\begin{equation*}
    \begin{split}
        \hat{x}_{i|i-1, \Theta} & = A\hat{x}_{i-1|i-1, \Theta} - B\hat{K}_{i-1}\hat{x}_{i-1|i-1, \hat{\Theta}_{i-1}}  + B\eta_{i-1}\\
        & = A\hat{x}_{i-1|i-1, \Theta} - B\hat{K}_{i-1}\hat{x}_{i-1|i-1, \Theta} - B\hat{K}_{i-1} \left( \hat{x}_{i-1|i-1, \hat{\Theta}_{i-1}} - \hat{x}_{i-1|i-1, \Theta} \right) + B\eta_{i-1}\\
        & = \left( A - B\hat{K}_{i-1}\right)\hat{x}_{i-1|i-1, \Theta} + B\hat{K}_{i-1} \left( \hat{x}_{i-1|i-1, \Theta} - \hat{x}_{i-1|i-1, \hat{\Theta}_{i-1}}\right) + B\eta_{i-1}\\
        & = \left( A - B\hat{K}_{i-1}\right) \left(\hat{x}_{i-1|i-2, \Theta} + Le_{i-1}\right)\\
        & \hspace{-4mm} + B\hat{K}_{i-1} \Bigg( \left(I - \hat{L}_{i-1} \hat{C}_{i-1} \right) \Delta_{i-1} + \left( L - \hat{L}_{i-1} \right)e_{i-1} + \hat{L}_{i-1} \left( \hat{C}_{i-1} - C \right)\hat{x}_{i-1|i-2, \Theta} \Bigg) + B\eta_{i-1}\\
    \end{split}    
\end{equation*}
\begin{equation}
    \begin{split}
        & = \Bigg( A - B\hat{K}_{i-1}\left(I - \hat{L}_{i-1} \left( \hat{C}_{i-1} - C \right) \right) \Bigg) \hat{x}_{i-1|i-2, \Theta} + B\hat{K}_{i-1} \left(I - \hat{L}_{i-1} \hat{C}_{i-1} \right) \Delta_{i-1}\\
        & \hspace{4mm} + \Bigg( \left(  A - B\hat{K}_{i-1} \right) L +  B\hat{K}_{i-1} \left( L - \hat{L}_{i-1} \right) \Bigg)e_{i-1} + B\eta_{i-1}\\
        & = \sum_{j = 0}^{i-1} \prod_{k = 0}^{i- 2 - j} \Bigg( A - B\hat{K}_{i-1-k} + B\hat{K}_{i-1-k}\hat{L}_{i-1-k} \left( \hat{C}_{i-1-k} - C \right) \Bigg)\\
        & \hspace{4mm} \Bigg( B\hat{K}_{j} \left(I - \hat{L}_{j} \hat{C}_{j} \right) \Delta_{j} + \left(\left( A - B\hat{K}_{j} \right) L +  B\hat{K}_{j} \left( L - \hat{L}_{j} \right)\right) e_j +  B\eta_{j}\Bigg).
    \end{split}
    \label{eq: state estimation and output bound - 2}
\end{equation}
For ease of representation, we can define the following terms:
\vspace{0mm}
\begin{equation}
    \begin{split}
        & a_i =  \left( A +  \delta_{\hat{\Theta}_{i}}\right) \left(I - \hat{L}_{i} \hat{C}_{i} \right), \hspace{2mm} b_i = \left(A \left( L - \hat{L}_i \right) - \delta_{\hat{\Theta}_i}\hat{L}_i \right),\hspace{2mm} c_i = \left( A +  \delta_{\hat{\Theta}_i}\right) \hat{L}_i \left( \hat{C}_i - C \right) - \delta_{\hat{\Theta}_i},\\
        & d_i =  \left(B - \hat{B}_i\right),\hspace{2mm} f_i =  A - B\hat{K}_{i} + B\hat{K}_{i}\hat{L}_{i} \left( \hat{C}_{i} - C \right),\hspace{2mm} g_i = B\hat{K}_{i} \left(I - \hat{L}_{i} \hat{C}_{i} \right),\\
        & h_i = \left( A - B\hat{K}_{i} \right) L +  B\hat{K}_{i} \left( L - \hat{L}_{i} \right).
    \end{split}
    \label{eq: state estimation and output bound - brevity terms}
\end{equation}
Finally, using \eqref{eq: state estimation and output bound - 1} with the equality given in \eqref{eq: state estimation and output bound - 2}, and with the definitions in \eqref{eq: state estimation and output bound - brevity terms}, we get
\vspace{0mm}
\begin{equation}
    \begin{split}
        \Delta_{t+1} & = \sum_{i = 0}^t \prod_{j = 0}^{t-i-1} a_{t-j} b_i e_i +  \sum_{i = 1}^t \prod_{j = 0}^{t-i-1} a_{t-j} c_i \left( \sum_{j = 0}^{i-1} \prod_{k = 0}^{i- 2 - j} f_{i-1-k} g_j \Delta_j \right)\\
        & \hspace{4mm} + \sum_{i = 1}^t \prod_{j = 0}^{t-i-1} a_{t-j} c_i \left(  \sum_{j = 0}^{i-1} \prod_{k = 0}^{i- 2 - j} f_{i-1-k} h_j e_j \right) + \sum_{i = 1}^t \prod_{j = 0}^{t-i-1} a_{t-j} c_i\left( \sum_{j = 0}^{i-1} \prod_{k = 0}^{i- 2 - j} f_{i-1-k} B \eta_j \right)\\
        & \hspace{4mm} + \sum_{i = 0}^t \prod_{j = 0}^{t-i-1} a_{t-j} d_i \eta_i.
    \end{split}
    \label{eq: expansion of Delta t plus 1}
\end{equation}
The above expression is similar to the one in the proof of \cite[Lemma 4.1]{lale2020regret}. The main difference lies in accounting for the additive excitation signal, i.e., the last two terms in the above expression. By following a similar analysis in the proof of \cite[Lemma 4.1]{lale2020regret}, we can upper bound the norm of each of the above terms in the decomposition by the corresponding bound at the end of the warm-up under the event $\mathcal{E}_{\mathbf{M}_{T_\text{w}}}$. However, in this paper, we have relaxed the (restrictive) assumption in \cite[Assumption 2.3]{lale2020regret}, which leads to a more involved analysis to provide a bound for the first four terms in \eqref{eq: expansion of Delta t plus 1}. 

\subsection*{Bounding the norm of $\sum_{i = 1}^t \prod_{j = 0}^{t-i-1} a_{t-j} c_i\left( \sum_{j = 0}^{i-1} \prod_{k = 0}^{i- 2 - j} f_{i-1-k} B \eta_j \right)$:}
Expanding the term, we see that
\begin{equation}
    \begin{split}
        & \sum_{i = 1}^t \prod_{j = 0}^{t-i-1} a_{t-j} c_i \left(  \sum_{j = 0}^{i-1} \prod_{k = 0}^{i- 2 - j} f_{i-1-k} B \eta_j \right) \\
        & = \left(a_t a_{t-1} ... a_2 c_1 \right)\left(B\eta_0\right)  + \left( a_t a_{t-1} ... a_3 c_2 \right)\left(f_{1}B\eta_0 + B\eta_1 \right) + \left( a_t a_{t-1} ... a_4 c_3 \right)\left(f_2 f_{1}B\eta_0 + f_2B\eta_1 + B\eta_2 \right) \\
        & \hspace{4mm} + \left( a_t a_{t-1} ... a_5 c_4 \right)\left(f_3 f_2 f_{1}B\eta_0 + f_2 f_1 B\eta_1 + f_1 B\eta_2 + B\eta_3\right) + ... \\
        & \hspace{4mm} + c_t\left(f_{t-1} f_{t-2} ... f_1B\eta_0 + f_{t-1} f_{t-2} ...f_2 B\eta_1 + ... + B\eta_{t-1}\right)\\  
        & = \left( a_t a_{t-1} ... a_2 c_1 + a_t a_{t-1} ... a_3 c_2 f_1 + a_t a_{t-1} ... a_4 c_3 f_2 f_1 + ... + c_t f_{t-1} f_{t-2} ... f_1 \right)B\eta_0\\
        & \hspace{4mm} + \left( a_t a_{t-1} ... a_3 c_2 + a_t a_{t-1} ... a_4 c_3 f_2 + a_t a_{t-1} ... a_5 c_4 f_2 f_1 + ... + c_t f_{t-1} f_{t-2} ... f_2\right) B\eta_1\\
        & \hspace{4mm} + ... + c_t B\eta_{t-1}.\\
    \end{split}
    \label{eq: expanding the 4th term in Delta_t}
\end{equation}
We will first start with bounding the norm of $a_t$. Recalling \eqref{eq: state estimation and output bound - brevity terms}, we have
\vspace{0mm}
\begin{equation*}
    \begin{split}
    a_t & = \left( A +  \delta_{\hat{\Theta}_{t}}\right) \left(I - \hat{L}_{t} \hat{C}_{t} \right)\\
    & = \left(A + \hat{A}_t - A - \hat{B}_t \hat{K}_t + B \hat{K}_t  \right) \left(I - \hat{L}_{t} \hat{C}_{t} \right)\\
    & = \left( \hat{A}_t - \left(\hat{B}_t - B \right)\hat{K}_t \right)\left(I - \hat{L}_{t} \hat{C}_{t} \right)\\
    & = \left(\hat{A}_t - \hat{A}_t \hat{L}_t \hat{C}_t \right) - \left(\hat{B}_t - B \right)\hat{K}_t\left(I - \hat{L}_{t} \hat{C}_{t} \right).\\
    \end{split}
\end{equation*}
With a similar argument to the bound of $||N_t^\tau||$ in \eqref{eq: bound on the dynamics term of x_hat_t}, where $\tau \in \mathbb{Z}^+$, there exists some time step $T_a > H$ such that $||a_t^\tau|| \leq \kappa_a \varphi_a^\tau$ under the event $\mathcal{E}_{\mathbf{M}_{T_\text{w}}}$ for all $t \in [T_\text{w}, 2T_\text{w}-1]$ with $T_\text{w} \geq T_a$, where $\varphi_a \in (0,1)$. Recalling from \eqref{eq: state estimation and output bound - brevity terms}, we have
\vspace{0mm}
\begin{equation*}
    \begin{split}
        f_t & = A - B\hat{K}_{t} + B\hat{K}_{t}\hat{L}_{t} \left( \hat{C}_{t} - C \right).
    \end{split}
\end{equation*}
From Theorem \ref{Theorem: perturbation bound on K}, we have $\left| \left| A - B\hat{K}_{t} - \left( A - BK \right) \right| \right| \leq \left| \left| B \left(K - \hat{K}_t \right) \right| \right| < \frac{1}{5||P||^{3/2}}$ \cite{simchowitz2020naive}. Furthermore, there exists an upper bound on $||A - BK||$ \cite{shi2023suboptimality}, which can be used to upper bound $|| A - B\hat{K}_t||$. Therefore, with a similar argument to the bound of $||N_t^\tau||$ in \eqref{eq: bound on the dynamics term of x_hat_t}, where $\tau \in \mathbb{Z}^+$, there exists some time step $T_f > H$ such that $||f_t^\tau|| \leq \kappa_f \varphi_f^\tau$ under the event $\mathcal{E}_{\mathbf{M}_{T_\text{w}}}$ for all $t \in [T_\text{w}, 2T_\text{w}-1]$ and $T_\text{w} \geq T_f$, where $\varphi_f \in (0,1)$. Let $\kappa_{(a,f)} \coloneqq \max\{\kappa_a, \kappa_f\}$ and $\varphi_{(a,f)} \coloneqq \max\{\varphi_a, \varphi_f\}$. Now, we proceed to provide a bound for the norm of $c_t$. Under the event $\mathcal{E}_{\mathbf{M}_{T_\text{w}}}$, the following bound exists \cite[Lemma 4.1]{lale2020regret}:
\vspace{0mm}
\begin{equation}
    \begin{split}
        ||c_t|| & \leq 2\Bigg( \Xi(A) + \beta_A(T_\text{w}) + \Gamma \beta_B(T_\text{w}) \Bigg) \zeta \beta_C(T_\text{w}) + 2\Bigg( \beta_A(T_\text{w}) + \Gamma \beta_B(T_\text{w}) \Bigg) \coloneqq \Bar{c}.\\
    \end{split}
    \label{eq: bound on c_t}
\end{equation}
Therefore, from \eqref{eq: expanding the 4th term in Delta_t} and using the above bounds, we have the following under the event $\mathcal{E}_{\mathbf{M}_{T_\text{w}}}$ for all $t \in [T_\text{w}, 2T_\text{w}-1]$:
\begin{equation*}
    \begin{split}
        &\left| \left| \sum_{i = 1}^t \prod_{j = 0}^{t-i-1} a_{t-j} c_i \left(  \sum_{j = 0}^{i-1} \prod_{k = 0}^{i- 2 - j} f_{i-1-k} B \eta_j \right) \right| \right| \\
        & \leq \left(t\kappa_{(a,f)} \varphi_{(a,f)}^{t-1} + (t-1)\kappa_{(a,f)} \varphi_{(a,f)}^{t-2} + \dots + 2\kappa_{(a,f)} \varphi_{(a,f)} + 1 \right)\Bar{c} ||B|| \sqrt{\gamma} \sqrt{2n_u \log \left( \frac{2n_ut}{\delta}\right)}\\
        & \leq \left( \left(1 + \kappa_{(a,f)} \varphi_{(a,f)} + \dots +  \kappa_{(a,f)} \varphi_{(a,f)}^{t-1}\right) + \left( \kappa_{(a,f)} \varphi_{(a,f)} + \dots +  \kappa_{(a,f)} \varphi_{(a,f)}^{t-1} \right) + \dots + \kappa_{(a,f)} \varphi_{(a,f)}^{t-1} \right)\\
        & \hspace{4mm} \Bar{c} ||B|| \sqrt{\gamma} \sqrt{2n_u \log \left( \frac{2n_ut}{\delta}\right)}\\
    \end{split}
\end{equation*}
\vspace{0mm}
\begin{equation}
    \begin{split}
        & \leq \Bar{c} ||B|| \sqrt{\gamma}\kappa_{(a,f)} \left(  \left( \varphi_{(a,f)}^{0} + \dots + \varphi_{(a,f)}^{t-1}\right) +  \left( \varphi_{(a,f)} + \dots + \varphi_{(a,f)}^{t-1}\right) + \dots + \varphi_{(a,f)}^{t-1} \right) \sqrt{2n_u \log \left( \frac{2n_ut}{\delta}\right)}\\
        & \leq \Bar{c} ||B|| \sqrt{\gamma}\kappa_{(a,f)} \left( \frac{1}{1 - \varphi_{(a,f)}} + \frac{\varphi_{(a,f)}}{1 - \varphi_{(a,f)}} + \dots \right)  \sqrt{2n_u \log \left( \frac{2n_ut}{\delta}\right)}\\
        & \leq \frac{ \Bar{c} ||B|| \sqrt{\gamma}\kappa_{(a,f)}}{(1 - \varphi_{(a,f)})^2}\sqrt{2n_u \log \left( \frac{2n_ut}{\delta}\right)},\\
    \end{split}
    \label{eq: bound on the 4th term in Delta_t}
\end{equation}
which holds with a probability of at least $1 - \delta$ after union bounding.

\subsection*{Bounding the norm of $\sum_{i = 0}^t \prod_{j = 0}^{t-i-1} a_{t-j} d_i \eta_i$:}

The bound on this term follows a similar procedure to that of the bound in \eqref{eq: bound on the 4th term in Delta_t}. It is straightforward to see that $||d_t|| \leq 1$ under the event $\mathcal{E}_{\mathbf{M}_{T_\text{w}}}$. Therefore, we have the following under the event $\mathcal{E}_{\mathbf{M}_{T_\text{w}}}$ for all $t \in [T_\text{w}, 2T_\text{w}-1]$:
\vspace{0mm}
\begin{equation}
    \left| \left|\sum_{i = 0}^t \prod_{j = 0}^{t-i-1} a_{t-j} d_i \eta_i \right| \right| \leq \frac{\sqrt{\gamma}\kappa_a}{1 - \varphi_a}\sqrt{2n_u \log \left( \frac{2n_ut}{\delta}\right)},
    \label{eq: bound on the 5th term in Delta_t}
\end{equation}
which holds with a probability of at least $1 - \delta$ after union bounding.

\subsection*{Bounding the norm of $\sum_{i = 0}^t \prod_{j = 0}^{t-i-1} a_{t-j} b_i e_i$:}
Similar to the other terms, there exists a bound on the norm of $b_t$ under the event $\mathcal{E}_{\mathbf{M}_{T_\text{w}}}$ for all $t \in [T_\text{w}, 2T_\text{w}-1]$ \cite[Lemma 4.1]{lale2020regret}:
\vspace{0mm}
\begin{equation*}
        ||b_t|| \leq 2\Xi(A) \beta_L(T_\text{w}) + 2\beta_A(T_\text{w}) \zeta + 2\beta_B(T_\text{w}) \Gamma \zeta \coloneqq \Bar{b}.
\end{equation*}
Observe that $e_t$ is a $\left(||C|| ||\Sigma||^{1/2} + \sigma_z \right)$-sub-Gaussian random variable \cite[Lemma 4.1]{lale2020regret}. Therefore, by using Lemma \ref{Lemma B.1}, with a probability of at least $1-\delta$, we have the following under event $\mathcal{E}_{\mathbf{M}_{T_\text{w}}}$ for all $t \in [T_\text{w}, 2T_\text{w}-1]$ after union bounding:
\vspace{0mm}
\begin{equation}
\begin{split}
    \left| \left| \sum_{i = 0}^t \prod_{j = 0}^{t-i-1} a_{t-j} b_i e_i \right| \right| & \leq \frac{\Bar{b} \kappa_a }{1 - \varphi_a} \left( ||C|| ||\Sigma||^{1/2} + \sigma_z \right) \sqrt{2n_y \log \left( \frac{2 n_y t}{\delta}\right)}.\\
\end{split}
\label{eq: bound on the 1st term in Delta_t}
\end{equation}

\subsection*{Bounding the norm of $\sum_{i = 1}^t \prod_{j = 0}^{t-i-1} a_{t-j} c_i \left(  \sum_{j = 0}^{i-1} \prod_{k = 0}^{i- 2 - j} f_{i-1-k} h_j e_j \right)$:}
Under the event $\mathcal{E}_{\mathbf{M}_{T_\text{w}}}$, we have for all $t \in [T_\text{w}, 2T_\text{w}-1]$ \cite[Lemma 4.1]{lale2020regret}:
\vspace{0mm}
\begin{equation*}
    \begin{split}
        ||h_t|| & \leq \Bigg(2\beta_A(T_\text{w}) + \kappa_\text{cl} \varphi_\text{cl} + 2\beta_B(T_\text{w}) \Gamma \Bigg) \zeta + 2||B|| \Gamma \beta_L(T_\text{w}) \coloneqq \Bar{h}.\\
    \end{split}
\end{equation*}
Similar to the bound in \eqref{eq: bound on the 4th term in Delta_t}, we have the following under event $\mathcal{E}_{\mathbf{M}_{T_\text{w}}}$ for all $t \in [T_\text{w}, 2T_\text{w}-1]$:
\vspace{0mm}
\begin{equation}
    \begin{split}
        &\left| \left| \sum_{i = 1}^t \prod_{j = 0}^{t-i-1} a_{t-j} c_i \left(  \sum_{j = 0}^{i-1} \prod_{k = 0}^{i- 2 - j} f_{i-1-k} h_j e_j \right) \right| \right| \leq \frac{\Bar{c} \Bar{h} \kappa_{(a,f)}}{(1 - \varphi_{(a,f)})^2} \left( ||C|| ||\Sigma||^{1/2} + \sigma_z \right) \sqrt{2n_y \log \left( \frac{2n_yt}{\delta}\right)},\\
    \end{split}
    \label{eq: bound on the 3rd term in Delta_t}
\end{equation}
which holds with a probability of at least $1 - \delta$ after union bounding.

\subsection*{Bounding the norm of $ \sum_{i = 1}^t \prod_{j = 0}^{t-i-1} a_{t-j} c_i \left( \sum_{j = 0}^{i-1} \prod_{k = 0}^{i- 2 - j} f_{i-1-k} g_j \Delta_j \right)$:}
It can be shown under the event $\mathcal{E}_{\mathbf{M}_{T_\text{w}}}$ that, for all $t \in [T_\text{w}, 2T_\text{w}-1]$, $||g_t|| \leq \Bar{g}$ exists \cite[Lemma 4.1]{lale2020regret}. With a similar argument as in the proof of \cite[Lemma 4.1]{lale2020regret}, it can be shown through the method of induction, that under the event $\mathcal{E}_{\mathbf{M}_{T_\text{w}}}$, the following holds with a probability of at least $1 - \delta$ for all $t \in [T_\text{w}, 2T_\text{w}-1]$:
\vspace{0mm}
\begin{equation}
\begin{split}
    & \left| \left| \hat{x}_{t|t-1, \Theta} - \hat{x}_{t|t-1, \hat{\Theta}_{t-1}} \right| \right| = \left| \left| \Delta_t \right| \right| \leq \Bar{\Delta}(t),\\
    \Bar{\Delta}(t) & = 10 \left( \frac{\Bar{b} \kappa_a}{1 - \varphi_a} + \frac{\Bar{c} \Bar{h} \kappa_{(a,f)}}{(1 - \varphi_{(a,f)})^2} \right)\left( ||C|| ||\Sigma||^{1/2} + \sigma_z \right)\sqrt{2n_y \log \left( \frac{2n_yt}{\delta}\right)}\\
    & \hspace{4mm} + 10 \left( \frac{\sqrt{\gamma}\kappa_a}{1 - \varphi_a} + \frac{\Bar{c} ||B|| \sqrt{\gamma} \kappa_{(a,f)}}{(1 - \varphi_{(a,f)})^2} \right)\sqrt{2n_u \log \left( \frac{2n_ut}{\delta}\right)}.
\end{split}
\label{eq: final bound on Delta_t}
\end{equation}

\subsection*{Putting things together:}
To recall from \eqref{eq: expanding x_hat_t}, we have
\vspace{0mm}
\begin{equation*}
    \begin{split}
        \hat{x}_{t|t,\hat{\Theta}_t} & = \sum_{i = 1}^t \prod_{j = 0}^{t-i-1} N_{t-j} \left( \hat{L}_{i}C\left( x_i - \hat{x}_{i|i-1,\Theta} \right) + \hat{L}_{i}C\left( \hat{x}_{i|i-1,\Theta} - \hat{x}_{i|i-1,\hat{\Theta}_{i-1}} \right) + \hat{B}_{i-1}\eta_{i-1} \right. \\
        & \left. \hspace{4mm} + \hat{L}_iz_i + \hat{L}_i\left( C - \hat{C}_i \right)\hat{B}_{i-1}\eta_{i-1}\right).
    \end{split}
\end{equation*}
From \eqref{eq: bound on the dynamics term of x_hat_t}, \eqref{eq: bound on process noise term of  x_hat_t - 1}, \eqref{eq: bound on process noise term of  x_hat_t - 2}, and \eqref{eq: final bound on Delta_t}, we have the following under the event $\mathcal{E}_{\mathbf{M}_{T_\text{w}}}$:
\begin{equation*}
    \begin{split}
        \left| \left| \hat{x}_{t|t,\hat{\Theta}_t} \right| \right| & \leq \left| \left| \sum_{i = 1}^t \prod_{j = 0}^{t-i-1} N_{t-j} \left( \hat{L}_{i}C\left( x_i - \hat{x}_{i|i-1,\Theta} \right) + \hat{L}_{i}C\left( \hat{x}_{i|i-1,\Theta} - \hat{x}_{i|i-1,\hat{\Theta}_{i-1}} \right) + \hat{B}_{i-1}\eta_{i-1} \right. \right. \right. \\
        & \left. \left. \left. \hspace{4mm} + \hat{L}_iz_i + \hat{L}_i\left( C - \hat{C}_i \right)\hat{B}_{i-1}\eta_{i-1}\right) \right| \right|\\
        & \leq \max_{1 \leq i \leq t} \left| \left| \left( \hat{L}_{i}C\left( x_i - \hat{x}_{i|i-1,\Theta} \right) + \hat{L}_{i}C\left( \hat{x}_{i|i-1,\Theta} - \hat{x}_{i|i-1,\hat{\Theta}_{i-1}} \right) + \hat{B}_{i-1}\eta_{i-1} \right. \right. \right. \\
        & \left. \left. \left. \hspace{4mm} + \hat{L}_iz_i + \hat{L}_i\left( C - \hat{C}_i \right)\hat{B}_{i-1}\eta_{i-1}\right) \right| \right| \left| \left| \sum_{i = 1}^t \prod_{j = 0}^{t-i-1} N_{t-j} \right| \right|\\
        & \leq \frac{\kappa_\text{cl}}{1 - \varphi_N}\max_{1 \leq i \leq t} \left| \left| \Bigg( \hat{L}_{i}C\left( x_i - \hat{x}_{i|i-1,\Theta} \right) + \hat{L}_{i}C\left( \hat{x}_{i|i-1,\Theta} - \hat{x}_{i|i-1,\hat{\Theta}_{i-1}} \right) + \hat{B}_{i-1}\eta_{i-1} \right. \right. \\
        &  \left. \left. \hspace{4mm} + \hat{L}_iz_i + \hat{L}_i\left( C - \hat{C}_i \right)\hat{B}_{i-1}\eta_{i-1}\Bigg) \right| \right| \leq \Bar{\mathcal{X}}(t),\\
    \end{split}
\end{equation*}
which holds with a probability of at least $1 - 3\delta$, where
\vspace{0mm}
\small
\begin{equation}
    \Bar{\mathcal{X}}(t) \coloneqq \frac{{\kappa_\text{cl}\Bigg( \zeta \left( ||C|| ||\Sigma||^{1/2} + \sigma_z \right) \sqrt{2n_x \log \left(\frac{2n_xt}{\delta}\right)} + \zeta||C||\Bar{\Delta}(t)\\ + \sqrt{\gamma}\left(\left| \left| B \right| \right| + 1 \right) \left( 1 + \zeta \right) \sqrt{2n_u\log \left( \frac{2n_ut}{\delta}\right)} \Bigg)}}{ 1 - \varphi_N }.
    \label{eq: bound on x_hat_t}
\end{equation}
\normalsize
Now, using \eqref{eq: bound on x_hat_t}, we can derive a bound on the norm of $\hat{x}_{t|t-1,\hat{\Theta}_{t-1}}$ as follows. From \eqref{eq: measurement and time update}, we have
\begin{equation*}
\begin{split}
    \hat{x}_{t|t-1,\hat{\Theta}_{t-1}} = \left(\hat{A}_{t-1} - \hat{B}_{t-1}\hat{K}_{t-1} \right) \hat{x}_{t-1|t-1, \hat{\Theta}_{t-1}} + \hat{B}_{t-1} \eta_{t-1},
\end{split}
\end{equation*}
then
\begin{equation}
    \begin{split}
        ||  \hat{x}_{t|t-1,\hat{\Theta}_{t-1}} || & \leq \left| \left| \hat{A}_{t-1} - \hat{B}_{t-1}\hat{K}_{t-1} \right| \right| \left| \left| \hat{x}_{t-1|t-1, \hat{\Theta}_{t-1}} \right| \right| + \left| \left| \hat{B}_{t-1} \right| \right| \left| \left| \eta_{t-1} \right| \right|\\
        & \leq \kappa_\text{cl} \varphi_\text{cl} \Bar{\mathcal{X}}(t) + \sqrt{\gamma}\left( ||B|| + 1 \right) \sqrt{2n_u \log \left( \frac{2n_ut}{\delta}\right)} \coloneqq X_\text{est,ac}(t),
    \end{split}
    \label{eq: bound on x_hat_t_t_1}
\end{equation}
which holds with a probability of at least $1 - 3\delta$ under the event $\mathcal{E}_{\mathbf{M}_{T_\text{w}}}$. Now to bound the norm of $u_t$, we recall the following from \eqref{eq: LBC policy naive}:
\vspace{0mm}
\begin{equation*}
    u_t = -\hat{K}_t \hat{x}_{t|t, \hat{\Theta}_t} + \eta_t,
\end{equation*}
then
\vspace{0mm}
\begin{equation}
        ||u_t|| \leq \Gamma \Bar{\mathcal{X}}(t) + \sqrt{\gamma}\sqrt{2n_u \log \left( \frac{2n_ut}{\delta}\right)} \coloneqq U_\text{ac}(t),
        \label{eq: bound on u_t}
\end{equation}
which holds with a probability of at least $1 - 3\delta$ under the event $\mathcal{E}_{\mathbf{M}_{T_\text{w}}}$. To derive a bound on the norm of $x_t$, we do the following:
\vspace{0mm}
\begin{equation*}
    \begin{split}
        x_t & = x_t - \hat{x}_{t|t-1,\Theta} + \hat{x}_{t|t-1,\Theta} - \hat{x}_{t|t-1,\hat{\Theta}_{t-1}} + \hat{x}_{t|t-1,\hat{\Theta}_{t-1}}\\
        & = x_t - \hat{x}_{t|t-1,\Theta} + \hat{x}_{t|t-1,\Theta} - \hat{x}_{t|t-1,\hat{\Theta}_{t-1}} + \left(\hat{A}_{t-1} - \hat{B}_{t-1}\hat{K}_{t-1} \right) \hat{x}_{t-1|t-1, \hat{\Theta}_{t-1}} + \hat{B}_{t-1} \eta_{t-1}.
    \end{split}
\end{equation*}
Then,
\vspace{0mm}
\begin{equation}
    \begin{split}
        ||x_t|| & \leq \left| \left| x_t - \hat{x}_{t|t-1,\Theta} \right| \right| + \left| \left| \hat{x}_{t|t-1,\Theta} - \hat{x}_{t|t-1,\hat{\Theta}_{t-1}} \right| \right| + \left| \left| \left(\hat{A}_{t-1} - \hat{B}_{t-1}\hat{K}_{t-1} \right) \right| \right| \left| \left| \hat{x}_{t-1|t-1, \hat{\Theta}_{t-1}} \right| \right|\\
        & \hspace{4mm} + \left| \left|  \hat{B}_{t-1} \right| \right| \left| \left|  \eta_{t-1} \right| \right|\\
        & \leq \left| \left| x_t - \hat{x}_{t|t-1,\Theta} \right| \right| + \left| \left| \hat{x}_{t|t-1,\Theta} - \hat{x}_{t|t-1,\hat{\Theta}_{t-1}} \right| \right| + \left| \left| \left(\hat{A}_{t-1} - \hat{B}_{t-1}\hat{K}_{t-1} \right) \right| \right| \left| \left| \hat{x}_{t-1|t-1, \hat{\Theta}_{t-1}} \right| \right|\\
        & \hspace{4mm} + \left( \left| \left| B \right| \right| + \left| \left| B - \hat{B}_{t-1} \right| \right| \right) \left| \left|  \eta_{t-1} \right| \right|\\
        & \leq ||\Sigma||^{1/2}\sqrt{2n_x \log \left(\frac{2n_x t}{\delta}\right)} + \Bar{\Delta}(t) + \kappa_\text{cl} \varphi_\text{cl} \Bar{\mathcal{X}}(t) + \sqrt{\gamma}\left( ||B|| + 1 \right) \sqrt{2n_u \log \left( \frac{2n_ut}{\delta}\right)} \\ & \coloneqq X_{ac}(t),
    \end{split}
    \label{eq: bound on x_t}
\end{equation}
which holds with a probability of at least $1 - 3\delta$ under the event $\mathcal{E}_{\mathbf{M}_{T_\text{w}}}$. Now for $y_t$, we have
\vspace{0mm}
\begin{equation*}
    \begin{split}
        y_t & = C \hat{x}_{t|t-1, \hat{\Theta}_{t-1}} + C \left( x_t -  \hat{x}_{t|t-1, \hat{\Theta}_{t-1}} \right) + z_t\\
        & = C \hat{x}_{t|t-1, \hat{\Theta}_{t-1}} + C \left( x_t -  \hat{x}_{t|t-1, \Theta} \right) + C \left( \hat{x}_{t|t-1, \Theta} -  \hat{x}_{t|t-1, \hat{\Theta}_{t-1}}\right) + z_t\\
        & = C \left(\hat{A}_{t-1} - \hat{B}_{t-1}\hat{K}_{t-1} \right)\hat{x}_{t-1|t-1, \hat{\Theta}_{t-1}} + C\hat{B}_{t-1}\eta_{t-1} + C \left( x_t -  \hat{x}_{t|t-1, \Theta} \right)\\
        & \hspace{4mm} + C \left( \hat{x}_{t|t-1, \Theta} -  \hat{x}_{t|t-1, \hat{\Theta}_{t-1}}\right) + z_t.\\
    \end{split}
\end{equation*}
Using similar analysis of $x_t$, we get the following bound for $y_t$:
\vspace{0mm}
\begin{equation}
\begin{split}
    ||y_t|| & \leq \kappa_\text{cl} \varphi_\text{cl} ||C|| \Bar{\mathcal{X}}(t) + \sqrt{\gamma}||C|| \left( 1 + ||B|| \right)\sqrt{2n_u \log \left( \frac{2n_ut}{\delta}\right)}\\
    & \hspace{4mm} + \left( ||C|| ||\Sigma||^{1/2} + \sigma_z \right) \sqrt{2n_x \log \left(\frac{2n_xt}{\delta}\right)} + ||C|| \Bar{\Delta}(t) \coloneqq Y_\text{ac}(t),
\end{split}
\label{eq: bound on y_t}
\end{equation}
which holds with a probability of at least $1 - 3\delta$ under the event $\mathcal{E}_{\mathbf{M}_{T_\text{w}}}$. By re-parameterizing $3\delta \rightarrow \delta$, the above bounds can be guaranteed with a probability of at least $1 - \delta$. Furthermore, it must be noted that $\Bar{\mathcal{X}}(t), X_\text{est,ac}(t), U_\text{ac}(t), X_\text{ac}(t), Y_\text{ac}(t) = \mathcal{O}(\sqrt{\log(t/\delta)})$ for the $0^\text{th}$ episode. Using this result, we now show that the condition on persistence of excitation will be satisfied for the $0^\text{th}$ episode.

\subsection*{Proof of Lemma \ref{Lemma: PE LBC} for the $0^\text{th}$ episode:}

The proof of Lemma \ref{Lemma: PE LBC} requires the truncated closed-loop noise evolution parameter $\mathcal{G}^\text{cl}$, which was defined in the beginning of the Appendix. This parameter captures how the noise and excitation signal sequences influence the input-output data in the vector $\phi$. This proof is an extension of the result in \cite[Lemma 3.1]{lale2021adaptive}, and the extension lies in accounting for the additive excitation signal. Conceptually, this proof follows similar steps as the one in \cite[Lemma 3.1]{lale2021adaptive} but, the terms involved are different, which makes analyzing the covariates $\sum_{i = T_\text{w}}^{t} \phi_i\phi_i^\top$ more challenging. 

From \eqref{eq:phi_def_shi}, we have
\vspace{0mm}
\begin{equation*}
    \mathbb{E}\left[\Bar{\phi}_t \Bar{\phi}_t^\top \right] = \mathcal{G}^\text{cl}_t \text{diag}(\sigma_w^2I,\sigma_z^2I, \sigma^2_{\eta_{t-2}}I, \sigma^2_{\eta_{t-1}}I,...,\sigma_w^2I,\sigma_z^2I, \sigma^2_{\eta_{t-2H-1}}I) {\mathcal{G}^\text{cl}_t}^\top  + 
        \begin{bmatrix}
            \mathbf{r}_{t-1}^\mathbf{c}\\.\\.\\.\\\mathbf{r}_{t-H}^\mathbf{c}
        \end{bmatrix}
        \begin{bmatrix}
            \mathbf{r}_{t-1}^\mathbf{c}\\.\\.\\.\\\mathbf{r}_{t-H}^\mathbf{c}
        \end{bmatrix}^\top.
\end{equation*}
This implies $\sigma_\text{min}\left(\mathbb{E}\left[\Bar{\phi}_t \Bar{\phi}_t^\top \right]\right) \geq \sigma_{\text{min}}^2(\mathcal{G}^\text{cl}_t)\text{min}\{\sigma_w^2,\sigma_z^2,\sigma_{\eta_{t-1}}^2\}$. To proceed, we require a lower bound on $\sigma_{\text{min}}(\mathcal{G}^\text{cl}_t)$. To obtain such a lower bound, the following perturbation bound exists:
\vspace{0mm}
\begin{equation*}
    || \mathcal{G}^\text{cl}_t - \mathcal{G}^\text{cl} || \leq \frac{\sigma_\text{c}}{2},
\end{equation*}
if $T_\text{w} \geq T_\mathcal{G}$ for some $T_\mathcal{G} > H$ and $t \in [T_\text{w}, 2T_\text{w}-1]$. Now, $T_\mathcal{G}$ can indeed be shown to exist. Given the structure of $\mathcal{G}^\text{cl}$, it holds that
\begin{equation*}
\begin{split}
    || \mathcal{G}^\text{cl}_t - \mathcal{G}^\text{cl} || & \leq \sum_{i = 1}^H || \Bar{\mathbf{G}}_{t-i} - \Bar{\mathbf{G}} ||\\
    & \leq \sum_{i = 1}^H || \mathbf{\hat{G}_1}^{(t-i)} - \mathbf{G_1} || + || \mathbf{\hat{\Psi}}_{t-i} \mathbf{\hat{G}_2}^{({t-i})} \mathbf{\hat{G}_3}^{({t-i-1})} - \mathbf{\Psi} \mathbf{G_2} \mathbf{G_3} || + \dots\\
    & \hspace{4mm} + ||\mathbf{\hat{\Psi}}_{t-i} \mathbf{\hat{G}_2}^{({t-i})} \mathbf{\hat{G}_2}^{({t-i}-1)} ... \mathbf{\hat{G}_2}^{({t-i}-H+1)}\mathbf{\hat{G}_3}^{({t-i}-H)} - \mathbf{\Psi} \mathbf{G_2}^{H}\mathbf{G_3} ||.
\end{split}
\end{equation*}
It must be noted that $\mathbf{\hat{\Psi}}_t$, $\mathbf{G_1}^{(t)}$, $\mathbf{G_2}^{(t)}$ and $\mathbf{G_3}^{(t)}$ are functions of $\hat{\Theta}_t$. Furthermore, in the $0^\text{th}$ episode, $\hat{\Theta}_t = \hat{\Theta}_{T_\text{w}}$. Therefore, from Lemma  \ref{Lemma: Markov parameter estimation error after warm-up} and Lemma \ref{Lemma: confidence bound on the system parameters}, a bound on $||\hat{\Theta}_t - \Theta||$ exists. As a consequence, it must be noted that $||\mathbf{\Psi}||$, $||\mathbf{G_1}||$, $||\mathbf{G_2}||$, $||\mathbf{G_3}||$, $||\mathbf{\hat{\Psi}}_t||$, $||\mathbf{G_1}^{(t)}||$, $||\mathbf{G_2}^{(t)}||$ and $||\mathbf{G_3}^{(t)}||$ can be upper bounded by a constant in the $0^\text{th}$ episode. Essentially, the basic problem that we are facing is the comparison of products of matrices. Therefore from Lemma \ref{Lemma B.13}, the time step $T_{\mathcal{G}}$ exists. 

One of the fundamental results of Weyl's inequalities on singular values is as follows:
\vspace{0mm}
\begin{equation}
    \sigma_j(X) + \sigma_j(Y) \leq \sigma_1(X + Y), \hspace{2mm} j = 1,2,...,\min\{m,n\},
    \label{eq: Weyl's inequality on singular values}
\end{equation}
which holds for any two matrices $X,Y \in \mathbb{R}^{m \times n}$. Taking $j = \min\{m,n\}$ and replacing $Y$ with $-Y$, we have
\vspace{0mm}
\begin{equation*}
    \begin{split}
        & \sigma_\text{min}(X) - \sigma_\text{min}(Y) \leq \sigma_\text{max}(X-Y)\\   
        \implies & \sigma_\text{min}(X) - \sigma_\text{max}(X - Y) \leq \sigma_\text{min}(Y).  
    \end{split}
\end{equation*}
Now taking $X = \mathcal{G}^\text{cl}$ and $Y = \mathcal{G}^\text{cl}_t$, we have
\vspace{0mm}
\begin{equation*}
    \begin{split}
        & \sigma_\text{min}(\mathcal{G}^\text{cl}) - \sigma_\text{max}(\mathcal{G}^\text{cl} - \mathcal{G}^\text{cl}_t) \leq \sigma_\text{min}(\mathcal{G}^\text{cl}_t)\\
        \implies & \sigma_\text{min}(\mathcal{G}^\text{cl}_t) \geq \sigma_\text{min}(\mathcal{G}^\text{cl}) - \sigma_\text{max}(\mathcal{G}^\text{cl}_t - \mathcal{G}^\text{cl})\\
        \implies & \sigma_\text{min}(\mathcal{G}^\text{cl}_t) \geq \frac{\sigma_\text{c}}{2}. 
    \end{split}
\end{equation*}
With the above result, we finally have 
\vspace{0mm}
\begin{equation}
    \begin{split}
        \sigma_\text{min}\left(\mathbb{E}\left[\Bar{\phi}_t \Bar{\phi}_t^\top \right]\right) & \geq \sigma_{\text{min}}^2(\mathcal{G}^\text{cl}_t)\text{min}\{\sigma_w^2,\sigma_z^2,\sigma_{\eta_{t-1}}^2\}\\
        & \geq \frac{\sigma_\text{c}^2}{4}\text{min}\{\sigma_w^2,\sigma_z^2,\sigma_{\eta_{t-1}}^2\}.\\
    \end{split}
    \label{eq: bound on min sigma of expected covariance in CL}
\end{equation}
Since singular values do not change under permutation, $\sigma_\text{min}\left(\mathbb{E}\left[\Bar{\phi}_t \Bar{\phi}_t^\top \right]\right) = \sigma_\text{min}\left(\mathbb{E}\left[\phi_t {\phi_t}^\top \right]\right)$. To recall, we need to derive a lower bound for $\sigma_\text{min}\left(\sum_{i = T_\text{w}}^{t} \phi_i\phi_i^\top \right)$. Firstly,  we will derive a bound on $\mid \mid \phi_t\mid\mid$. From \eqref{eq: bound on u_t} and \eqref{eq: bound on y_t}, we have
\vspace{0mm}
\begin{equation}
\begin{split}
    \mid\mid \phi_t \mid\mid & = \sqrt{\sum_{i = 1}^H ||y_{t-i}||^2 + ||u_{t-i}||^2}\\
    & \leq \sqrt{H \max_{1\leq i \leq H} (||y_{t-i}||^2 + ||u_{t-i}||^2)}\\
    & \leq \left(\sqrt{\max_{1\leq i \leq H}||y_{t-i}||^2} + \sqrt{\max_{1\leq i \leq H}||u_{t-i}||^2}\right)\sqrt{H}\\
    & \leq \underbrace{(Y_\text{ac}(t) + U_\text{ac}(t))}_{{\Upsilon_\text{ac}(t)}}\sqrt{H},
\end{split}
\label{eq: bounding 2-norm of phi_t CL}
\end{equation}
which holds with a probability of at least $1-\delta/2$ under the event $\mathcal{E}_{\mathbf{M}_{T_\text{w}}}$ after re-parameterising $\delta \rightarrow \delta/2$. The reason for considering $\delta/2$ will become apparent shortly. We will now apply Lemma \ref{Lemma B.2}. In the statement of Lemma \ref{Lemma B.2}, it becomes evident that the notations $\mathbf{X}_k = \mathbf{A}_k = \phi_i\phi_i^\top $, and the notation
\vspace{0mm}
\begin{equation*}
    \begin{split}
        \sigma^2 & = ||\sum_{i=T_\text{w}}^{t} (\phi_i\phi_i^\top )^2 ||\\
        & \leq (t - T_\text{w} + 1) \max_{T_\text{w}\leq i \leq t} \mid\mid \phi_i \mid\mid^2 \hspace{1mm}\mid\mid \phi_i^\top  \mid\mid^2\\
        & =  (t - T_\text{w} + 1) \max_{T_\text{w}\leq i \leq t} \mid\mid \phi_i \mid\mid^4.\\
    \end{split}
\end{equation*}
From Lemma \ref{Lemma B.2}, we can set 
\vspace{0mm}
\begin{equation*}
    \begin{split}
        & \frac{\delta}{2} = H(n_y + n_u) \hspace{1mm} \exp\left( \frac{-t^2}{8\sigma^2} \right)\\
        \implies & -\text{log}\left(\frac{\delta}{2H(n_y + n_u)}\right) = \frac{t^2}{8\sigma^2}\\
        \implies & \text{log}\left(\frac{2H(n_y + n_u)}{\delta}\right) = \frac{t^2}{8\sigma^2}\\
        \implies & t = 2\sqrt{2}\sigma\sqrt{\text{log}\left(\frac{2H(n_y + n_u)}{\delta}\right)}\\
        \implies & t \leq 2\sqrt{2(t - T_\text{w} + 1)}\max_{T_\text{w}\leq i \leq t} ||\phi_i||^2\sqrt{\text{log}\left(\frac{2H(n_y + n_u)}{\delta}\right)}.\\
    \end{split}
\end{equation*}
Finally by using Lemma \ref{Lemma B.2}, we have 
\vspace{0mm}
\begin{equation*}
    \lambda_\text{max}\left( \sum_{i=T_\text{w}}^{t} \phi_i\phi_i^\top  - \mathbb{E}\left[ \phi_i\phi_i^\top  \right]\right) \leq 2\sqrt{2(t - T_\text{w} + 1)}{\Upsilon_\text{ac}(t)}^2 H\sqrt{\text{log}\left(\frac{2H(n_y + n_u)}{\delta}\right)},
    \label{eq: Matrix Azuma inequality CL}
\end{equation*}
which holds with a probability of at least $1 - \delta$ under the event $\mathcal{E}_{\mathbf{M}_{T_\text{w}}}$. Since $\phi_i\phi_i^\top $ is a symmetric matrix, its singular values are the absolute values of its eigenvalues. Now using Weyl's inequality as described in \eqref{eq: Weyl's inequality on singular values}, we get
\vspace{0mm}
\begin{equation}
\begin{split}
    \sigma_{\text{min}}\left(\sum_{i = T_\text{w}}^{t} \phi_i\phi_i^\top \right) & \geq \sigma_{\text{min}}\left(\sum_{i = T_\text{w}}^{t} \mathbb{E}[\phi_i\phi_i^\top ]\right) - \left|\lambda_{\text{max}}\left(\sum_{i = T_\text{w}}^{t} \phi_i\phi_i^\top  - \mathbb{E}[\phi_i\phi_i^\top ]\right)\right|\\
    \implies \sigma_{\text{min}}\left(\sum_{i = T_\text{w}}^{t} \phi_i\phi_i^\top \right) & \geq (t - T_\text{w} + 1) \frac{\sigma_\text{c}^2}{4}\text{min}\{\sigma_w^2,\sigma_z^2,\sigma_{\eta_{t-1}}^2\}\\
    & \hspace{4mm} -  2\sqrt{2(t - T_\text{w} + 1)}{{\Upsilon_\text{ac}(t)}}^2 H\sqrt{\text{log}\left(\frac{2H(n_y + n_u)}{\delta}\right)},\\
\end{split}
\label{eq: bound on the persistency of excitation CL}
\end{equation}
which holds with a probability of at least $1 - \delta$ under the event $\mathcal{E}_{\mathbf{M}_{T_\text{w}}}$. Now we need to determine the minimum number of time steps to ensure $\sigma_{\text{min}}\left(\sum_{i = T_\text{w}}^{t} \phi_i\phi_i^\top \right) > 0$. Equating the RHS of \eqref{eq: bound on the persistency of excitation CL} to $0$ we get,
\vspace{0mm}
\begin{equation*}
    \begin{split}
         & 2\sqrt{2(t - T_\text{w} + 1)}{\Upsilon_\text{ac}(t)}^2 H\sqrt{\text{log}\left(\frac{2H(n_y + n_u)}{\delta}\right)} = (t - T_\text{w} + 1) \frac{\sigma_\text{c}^2}{4}\text{min}\{\sigma_w^2,\sigma_z^2,\sigma_{\eta_{t-1}}^2\} \\
         \implies & 8(t - T_\text{w} + 1) {\Upsilon_\text{ac}(t)}^4 H^2 \text{log}\left(\frac{2H(n_y + n_u)}{\delta}\right) = \frac{\sigma_\text{c}^4}{16}(t - T_\text{w} + 1)^2\text{min}\{\sigma_w^4,\sigma_z^4,\sigma_{\eta_{t-1}}^4\}\\
         \implies & (t - T_\text{w} + 1) = \frac{ 128 {\Upsilon_\text{ac}(t)}^4 H^2 \text{log}\left(\frac{2H(n_y + n_u)}{\delta}\right)}{\sigma_\text{c}^4\text{min}\{\sigma_w^4,\sigma_z^4,\sigma_{\eta_{t-1}}^4\}}.
    \end{split}
\end{equation*}
Therefore, for all $(t - T_\text{w} + 1) \geq T_\text{ac}(t)$ and $T_\text{w} \geq \Bar{T}^\text{alt,2}_\text{w}$ for $\Bar{T}^\text{alt,2}_\text{w}$ as defined in \eqref{eq: Alt LB on T_w}, where
\vspace{0mm}
\begin{equation}
    T_\text{ac}(t) = \frac{ 512 {\Upsilon_\text{ac}(t)}^4 H^2 \text{log}\left(\frac{2H(n_y + n_u)}{\delta}\right)}{\sigma_\text{c}^4\text{min}\{\sigma_w^4,\sigma_z^4,\sigma_{\eta_{t-1}}^4\}},
    \label{eq: Tac for PE CL}
\end{equation}
we have
\vspace{0mm}
\begin{equation}
\begin{split}
    & \sigma_{\text{min}}\left(\sum_{i = T_\text{w}}^{t} \phi_i\phi_i^\top \right) \geq \frac{ 128 {\Upsilon_\text{ac}(t)}^4 H^2 \text{log}\left(\frac{2H(n_y + n_u)}{\delta}\right)}{\sigma_\text{c}^2\text{min}\{\sigma_w^2,\sigma_z^2,\sigma_{\eta_{t-1}}^2\}} - \frac{ 64 {\Upsilon_\text{ac}(t)}^4 H^2 \text{log}\left(\frac{2H(n_y + n_u)}{\delta}\right)}{\sigma_\text{c}^2\text{min}\{\sigma_w^2,\sigma_z^2,\sigma_{\eta_{t-1}}^2\}}\\
    \implies & \sigma_{\text{min}}\left(\sum_{i = T_\text{w}}^{t} \phi_i\phi_i^\top \right) \geq (t - T_\text{w} + 1)\frac{\sigma_\text{c}^2\text{min}\{\sigma_w^2,\sigma_z^2,\sigma_{\eta_{t-1}}^2\}}{8},\\
\end{split}
\label{eq: final bound on the persistency of excitation CL}
\end{equation}
which holds with a probability of at least $1 - \delta$ under the event $\mathcal{E}_{\mathbf{M}_{T_\text{w}}}$. Recall from the previous subsection that $\Upsilon_\text{ac}(t)= \mathcal{O}(\sqrt{\log(t/\delta)})$ for the $0^\text{th}$ episode.

\subsection*{Proof of Lemma \ref{Lemma: Markov parameter estimation error during LBC phase} for the $0^\text{th}$ episode:}

Now with the above result, we address the Markov parameter estimation error bound after the $0^\text{th}$ episode. The proof of Lemma \ref{Lemma: Markov parameter estimation error during LBC phase} is an extension of an earlier result \cite[Th. 3.3]{lale2021adaptive}, where the extension requires accounting for the additive Gaussian excitatory signals. The structure of this proof is identical to the proof of \cite[Th. 3.3]{lale2021adaptive}, the only difference lies in the persistence of excitation condition \eqref{eq: final bound on the persistency of excitation CL}, which needs to be incorporated to derive a bound on the Markov parameter estimation error. This change is incremental and is presented here for the sake of completeness. For a detailed treatment of all the steps involved in this proof, refer to the proof of \cite[Th. 3.3]{lale2021adaptive}.

Recalling from \eqref{eq: input-output trajectory representation}, we have for a single input-output trajectory $\{y_t,u_t\}_{t = 0}^{t}$:
\vspace{0mm}
\begin{equation*}
    Y_t = \Phi_{t} {\mathbf{M}}^\top  + E_{t} + N_t,
\end{equation*}
where
\vspace{0mm}
\begin{equation*}
\begin{split}
\mathbf{M} & = 
\begin{bmatrix}
    CF & C\Bar{A}F & ... & C\Bar{A}^{H-1}F & CB & C\Bar{A}B & ... & C\Bar{A}^{H-1}B 
\end{bmatrix} \in \mathbb{R}^{n_y \times (n_y + n_u)H},\\
    Y_{t} & = 
\begin{bmatrix}
    y_H & y_{H+1} & ... & y_{t}
\end{bmatrix}^\top  \in \mathbb{R}^{(t-H+1) \times n_y},\\
\Phi_{t} & = 
\begin{bmatrix}
    \phi_H & \phi_{H+1} & ... & \phi_{t}
\end{bmatrix}^\top  \in \mathbb{R}^{(t-H+1) \times (n_y + n_u)H},\\
E_{t} & =
\begin{bmatrix}
    e_H & e_{H+1} & ... & e_{t}
\end{bmatrix}^\top  \in \mathbb{R}^{(t-H+1) \times n_y},\\
N_{t} & = 
\begin{bmatrix}
    C\Bar{A}^H\hat{x}_{0|-1,\Theta} & C\Bar{A}^H\hat{x}_{1|0,\Theta} & ... & C\Bar{A}^H\hat{x}_{t-H|t - H -1,\Theta}
\end{bmatrix}^\top  \in \mathbb{R}^{(t-H+1) \times n_y},
\end{split}
\end{equation*}
with $\Bar{A} = A - FC$. Further, recall from \eqref{eq: RLS to estimate Markov parameters in closed-loop} that 
\vspace{0mm}
\begin{equation*}
    \hat{\mathbf{M}}_t^\top  = ({\Phi_t}^\top \Phi_t + \lambda I)^{-1}{\Phi_t}^\top  Y_t.
\end{equation*}
This implies
\vspace{0mm}
\begin{equation*}
\begin{split}
       \hat{\mathbf{M}}_t & = \left[ ({\Phi_t}^\top \Phi_t + \lambda I)^{-1}{\Phi_t}^\top E_t + ({\Phi_t}^\top \Phi_t + \lambda I)^{-1}{\Phi_t}^\top N_t + {\mathbf{M}}^\top  -  \lambda({\Phi_t}^\top \Phi_t + \lambda I)^{-1}{\mathbf{M}}^\top  \right]^\top .
\end{split}
\end{equation*}
Now consider the following:
\vspace{0mm}
\begin{equation*}
    \begin{split}
        |\text{Tr}(X(\hat{\mathbf{M}}_t - \mathbf{M})^\top )|  & \leq |\text{Tr}(X({\Phi_t}^\top \Phi_t + \lambda I)^{-1}{\Phi_t}^\top E_t)| + |\text{Tr}(X({\Phi_t}^\top \Phi_t + \lambda I)^{-1}{\Phi_t}^\top N_t)|\\
        &\hspace{4mm} + \lambda |\text{Tr}(X({\Phi_t}^\top \Phi_t + \lambda I)^{-1}{\mathbf{M}}^\top )|,\
    \end{split}
\end{equation*}
where $X$ is some matrix. Let $M_1, M_2, M_3$ be three matrices. Using the property $|\text{Tr}(M_1 M_2 M_3^\top )| \leq \sqrt{\text{Tr}(M_1 M_2 M_1^\top )\text{Tr}(M_3 M_2 M_3^\top )}$ for a positive definite $M_2$ \cite{lale2021adaptive}, we have
\vspace{0mm}
\begin{equation*}
    \begin{split}
        |\text{Tr}(X(\hat{\mathbf{M}}_t - \mathbf{M})^\top )| & \leq \sqrt{\text{Tr}(X({\Phi_t}^\top \Phi_t + \lambda I)^{-1}X^\top )} \times\\
        & \hspace{4mm} \Bigg[ \sqrt{\text{Tr}(E_t^\top \Phi_t({\Phi_t}^\top \Phi_t + \lambda I)^{-1}{\Phi_t}^\top E_t)} + \sqrt{\text{Tr}(N_t^\top \Phi_t({\Phi_t}^\top \Phi_t + \lambda I)^{-1}{\Phi_t}^\top N_t)}\\
        & \hspace{4mm} + \lambda \sqrt{\text{Tr}(\mathbf{M}({\Phi_t}^\top \Phi_t + \lambda I)^{-1}{\mathbf{M}}^\top )} \Bigg].\\
    \end{split}
\end{equation*}
Substituting $X = (\hat{\mathbf{M}}_t - \mathbf{M})({\Phi_t}^\top \Phi_t + \lambda I)$ and $V_t = ({\Phi_t}^\top \Phi_t + \lambda I)$, we get
\vspace{0mm}
\begin{equation*}
\begin{split}
    |\text{Tr}((\hat{\mathbf{M}}_t - \mathbf{M})V_t(\hat{\mathbf{M}}_t - \mathbf{M})^\top )| & \leq \sqrt{\text{Tr}((\hat{\mathbf{M}}_t - \mathbf{M})V_t^\top (\hat{\mathbf{M}}_t - \mathbf{M})^\top )} \times\\
    & \hspace{4mm} \Bigg[ \sqrt{\text{Tr}(E_t^\top \Phi_tV_t^{-1}{\Phi_t}^\top E_t)} + \sqrt{\text{Tr}(N_t^\top \Phi_tV_t^{-1}{\Phi_t}^\top N_t)}\\
    & \hspace{4mm} +  \lambda \sqrt{\text{Tr}(\mathbf{M}V_t^{-1}{\mathbf{M}}^\top )} \Bigg].\\
\end{split}
\end{equation*}
Since $V_t$ is a symmetric positive definite matrix, the above expression reduces to
\vspace{0mm}
\begin{equation*}
\begin{split}
    \sqrt{\text{Tr}((\hat{\mathbf{M}}_t - \mathbf{M})V_t(\hat{\mathbf{M}}_t - {\mathbf{M}})^\top )} & \leq \Bigg[ \sqrt{\text{Tr}(E_t^\top \Phi_tV_t^{-1}{\Phi_t}^\top E_t)} + \sqrt{\text{Tr}(N_t^\top \Phi_tV_t^{-1}{\Phi_t}^\top N_t)}\\
    & \hspace{4mm} + \sqrt{\lambda}  \Bar{\mathbf{m}} \Bigg],\\
\end{split}
\end{equation*}
where $||\mathbf{M}||_\mathrm{F} \leq \Bar{\mathbf{m}}$. Now we provide bounds for each of the terms in the above expression.

\subsection*{Bounding $\sqrt{\textup{Tr}(E_t^\top \Phi_tV_t^{-1}{\Phi_t}^\top E_t)}$ \textup{\cite{lale2021adaptive}}:}

Since  $e_t$ is $||C\Sigma{C}^\top  + \sigma_z^2 I||$-sub-Gaussian vector, from Theorem \ref{Theorem B.1} we have:
\vspace{0mm}
\begin{equation}
    \sqrt{\text{Tr}(E_t^\top \Phi_tV_t^{-1}{\Phi_t}^\top E_t)} \leq \sqrt{n_y ||C\Sigma{C}^\top  + \sigma_z^2 I|| \log \left( \frac{\text{det}(V_t)^{1/2}\text{det}(V)^{-1/2}}{\delta}\right)},
    \label{eq: bound on the first term in Markov parameter bound}
\end{equation}
which holds with a probability of at least $1- \delta$. Here, $V = \lambda I$. For the sake of convenience, define the event $\mathcal{E}_{E_t}$:
\vspace{0mm}
\begin{equation*}
    \mathcal{E}_{E_t} \coloneqq \left\{ \sqrt{\text{Tr}(E_t^\top \Phi_tV_t^{-1}{\Phi_t}^\top E_t)} \leq \sqrt{n_y ||C\Sigma{C}^\top  + \sigma_z^2 I|| \log \left( \frac{\text{det}(V_t)^{1/2}\text{det}(V)^{-1/2}}{\delta}\right)} \right\},
\end{equation*}
which holds with a probability of at least $1- \delta$.

\subsection*{Bounding $\sqrt{\textup{Tr}(N_t^\top \Phi_tV_t^{-1}{\Phi_t}^\top N_t)}$\textup{\cite{lale2021adaptive}}:}
\vspace{0mm}
\begin{equation*}
\begin{split}
    \sqrt{\text{Tr}(N_t^\top \Phi_tV_t^{-1}{\Phi_t}^\top N_t)} & \leq \frac{1}{\sqrt{\lambda}} || N_t^\top \Phi_t||_\mathrm{F}\\
    & \leq \sqrt{\frac{n_y}{\lambda}} \left|\left| \sum_{i = H}^t \phi_i (C\bar{A}^H\hat{x}_{i-H|i-H-1,\Theta})^\top  \right|\right|.\\
\end{split}
\end{equation*}
The last inequality comes from the property, $\sigma_{\text{max}}(X) \leq ||X||_\mathrm{F} \leq \sqrt{\min\{m,n\}}\sigma_{\text{max}}(X)$, for any matrix $X \in \mathbb{R}^{m \times n}$. Now,
\vspace{0mm}
\begin{equation*}
\begin{split}
    \sqrt{\text{Tr}(N_t^\top \Phi_tV_t^{-1}{\Phi_t}^\top N_t)} & \leq \sqrt{\frac{n_y}{\lambda}} \left|\left| \sum_{i = H}^t \phi_i (C\bar{A}^H\hat{x}_{i-H|i-H-1,\Theta})^\top  \right|\right|\\
    & \leq  (t - H + 1) \sqrt{\frac{n_y}{\lambda}} \max_{H \leq i \leq t}\left|\left| \phi_i (C\bar{A}^H\hat{x}_{i-H|i-H-1,\Theta})^\top  \right|\right|\\
    & \leq t \sqrt{\frac{n_y}{\lambda}} ||C||c_{\nu}\nu^H \max_{H \leq i \leq t} ||\phi_i||\hspace{1mm} ||\hat{x}_{i-H|i-H-1,\Theta}||,\\
\end{split}
\end{equation*}
where $||(A - FC)^H|| \leq c_{\nu}\nu^H$ for some $\nu \in (0,1)$ and $c_{\nu} \geq 1$.  During the warm-up phase, i.e., $t \leq T_\text{w}-1$, we can bound $||\hat{x}_{t|t-1,\Theta}|| \leq X_\text{w}$, where $X_\text{w}$ is as defined in \eqref{eq: bound on state and input during warm-up}. During the warm-up phase, recall from Lemma \ref{Lemma: bound on state and input during warm-up} that $||\phi_t|| \leq \Upsilon_\text{w}\sqrt{H}$ with a probability of at least $1 - \delta/2$. Therefore, during the warm-up phase, we have
\vspace{0mm}
\begin{equation*}
    \max_{H \leq i \leq T_\text{w}-1} ||\phi_i||\hspace{1mm} ||\hat{x}_{i-H|i-H-1,\Theta}|| \leq \Upsilon_\text{w} X_\text{w} \sqrt{H},
\end{equation*}
which holds with a probability of at least $1 - \delta/2$. For the ease of comprehension, define an event $\mathcal{E}_{\phi,\text{warm}}$, where
\begin{equation*}
    \mathcal{E}_{\phi,\text{warm}} \coloneqq \left\{ ||\phi_t|| \leq \Upsilon_\text{w}\sqrt{H} \right\},
\end{equation*}
which holds with a probability of at least $1 - \delta/2$. Now, during the adaptive control phase, we have that $||\phi_t|| \leq {\Upsilon_\text{ac}(t)}\sqrt{H}$, with a probability of at least $1 - \delta/2$ under the event $\mathcal{E}_{\mathbf{M}_{T_\text{w}}}$ after re-parameterizing $\delta \rightarrow \delta/2$. Therefore, during the adaptive control phase, we have the following:
\vspace{0mm}
\begin{equation*}
    \max_{T_\text{w} \leq i \leq t} ||\phi_i||\hspace{1mm} ||\hat{x}_{i-H|i-H-1,\Theta}|| \leq {\Upsilon_\text{ac}(t)} X_\text{est,ac}(t) \sqrt{H},
\end{equation*}
which holds with a probability of at least $1 - \delta/2$, where $X_\text{est,ac}(t)$ is as defined in \eqref{eq: bound on x_hat_t_t_1}. Define an event $\mathcal{E}_{\phi,\text{ac}}$ given the event $\mathcal{E}_{\mathbf{M}_{T_\text{w}}}$:
\begin{equation*}
    \mathcal{E}_{\phi,\text{ac}} \coloneqq \left\{ ||\phi_t|| \leq {\Upsilon_\text{ac}(t)}\sqrt{H}  \right\},
\end{equation*}
which holds with a probability of at least $1 - \delta/2$. By setting $H \geq \Bar{H}$, where
\vspace{0mm}
\begin{equation}
\begin{split}
      \Bar{H} & \coloneqq \frac{\log\left(\sqrt{n_y/\lambda}||C||c_{\nu}\max \{\Upsilon_\text{w}X_\text{w}, {\Upsilon_\text{ac}(t)X_\text{est,ac}(t)} \} T^2\right)}{\log(1 / \nu)},
\end{split}
  \label{eq: value of H for CL}
\end{equation}
we have 
\vspace{0mm}
\begin{equation}
\begin{split}
    \sqrt{\text{Tr}(N_t^\top \Phi_tV_t^{-1}{\Phi_t}^\top N_t)} & \leq \frac{t}{T^2} \sqrt{H},
\end{split}
\label{eq: upper bound on bias term during LBC phase}
\end{equation}
which holds under the event $\mathcal{E}_{\phi,\text{warm}}\cap \mathcal{E}_{\phi,\text{ac}}$.

\subsection*{Putting things together:}

Now, combining \eqref{eq: bound on the first term in Markov parameter bound} and \eqref{eq: upper bound on bias term during LBC phase}, we obtain the following under the event  $\mathcal{E}_{\phi,\text{warm}}\cap \mathcal{E}_{\phi,\text{ac}} \cap \mathcal{E}_{E_t}$:
\vspace{0mm}
\begin{equation*}
\begin{split}
    & \sqrt{\text{Tr}((\hat{\mathbf{M}}_t - \mathbf{M})V_t(\hat{\mathbf{M}}_t - \mathbf{M})^\top )} \leq \sqrt{n_y ||\Sigma_e|| \log \left( \frac{\text{det}(V_t)^{1/2}\text{det}(V)^{-1/2}}{\delta}\right)} + \frac{t}{T^2} \sqrt{H} + \sqrt{\lambda} \Bar{\mathbf{m}}, \\
\end{split}
\end{equation*}
where $\Sigma_e = C\Sigma{C}^\top  + \sigma_z^2 I$. Now, 
\vspace{0mm}
\begin{equation*}
    \begin{split}
        & \text{Tr}((\hat{\mathbf{M}}_t - \mathbf{M})V_t(\hat{\mathbf{M}}_t - \mathbf{M})^\top ) \geq \sigma_{\text{min}}(V_t)||\hat{\mathbf{M}}_t - \mathbf{M}||^2_\mathrm{F}\\
        \implies & \sigma_{\text{min}}(V_t)||\hat{\mathbf{M}}_t - \mathbf{M}||^2_\mathrm{F} \leq \left(\sqrt{n_y ||\Sigma_e|| \log \left( \frac{\text{det}(V_t)^{1/2}\text{det}(V)^{-1/2}}{\delta}\right)} + \frac{t}{T^2} \sqrt{H} + \sqrt{\lambda}  \Bar{\mathbf{m}}\right)^2.\\
    \end{split}
\end{equation*}

Now we introduce the events corresponding to the persistence of excitation conditions for the warm-up phase and the adaptive control phase. For the warm-up phase, define an event $\mathcal{E}_{\text{PE, warm}}$ given the event $\mathcal{E}_{\phi,\text{warm}}$, where
\begin{equation*}
    \mathcal{E}_{\text{PE, warm}} \coloneqq \left\{  \sigma_{\text{min}}\left(\sum_{i = H}^{T_\text{w} - 1} \phi_i\phi_i^\top \right) \geq (T_\text{w} - H) \frac{\sigma_\text{o}^2\text{min}\{\sigma_w^2, \sigma_z^2, \sigma_u^2\}}{2} \right\},
\end{equation*}
which holds with a probability of at least $1 - \delta/2$ if $T_\text{w} \geq \max\{T_\text{o}, H\}$ (refer to proof of \cite[Lemma 3.1]{lale2021adaptive}). This event is a consequence of the result in Lemma \ref{Lemma: PE warm-up}.

For the adaptive control phase, we can use \eqref{eq: final bound on the persistency of excitation CL} to define an event $\mathcal{E}_{\text{PE, ac}}$ given the event $\mathcal{E}_{\phi,\text{ac}}$, where
\begin{equation*}
    \mathcal{E}_{\text{PE, ac}} \coloneqq \left\{ \sigma_{\text{min}}\left(\sum_{i = T_\text{w}}^{t} \phi_i\phi_i^\top \right) \geq (t - T_\text{w} + 1)\frac{\sigma_\text{c}^2\text{min}\{\sigma_w^2,\sigma_z^2,\sigma_{\eta_{t-1}}^2\}}{8} \right\},
\end{equation*}
which holds with a probability of at least $1 - \delta/2$ if $t - T_\text{w} + 1 \geq T_\text{ac}(t)$, where $T_\text{ac}(t)$ is as defined in \eqref{eq: Tac for PE CL}.

Therefore, under the event $\mathcal{E}_{\text{PE, warm}} \cap \mathcal{E}_{\text{PE, ac}}$, we have
\vspace{0mm}
\begin{equation*}
    \begin{split}
        \sigma_\text{min}(V_t) & = \sigma_\text{min}\left(\sum_{t = H}^{t} \phi_i {\phi_i}^\top + \lambda I\right)\\
        & \geq \sigma_\text{min}\left(\sum_{t = H}^{t} \phi_i {\phi_i}^\top \right)\\
        & \geq (T_\text{w} - H) \frac{\sigma_\text{o}^2\text{min}\{\sigma_w^2, \sigma_z^2, \sigma_u^2\}}{2} + (t - T_\text{w} + 1)\frac{\sigma_\text{c}^2\text{min}\{\sigma_w^2,\sigma_z^2,\sigma_{\eta_{t-1}}^2\}}{8}.
    \end{split}
\end{equation*}
Finally, from Lemma \ref{Lemma B.4}, we have for $t \in [T_\text{w}, 2T_\text{w}-1]$:
\vspace{0mm}
\begin{equation}
    \begin{split}
        & ||\hat{\mathbf{M}}_t - \mathbf{M}||_\mathrm{F} \leq \frac{\sqrt{n_y ||\Sigma_e|| \log \left( \frac{\text{det}(V_t)^{1/2}\text{det}(V)^{-1/2}}{\delta}\right)} + \frac{t}{T^2} \sqrt{H} + \sqrt{\lambda}\Bar{\mathbf{m}}}{\sqrt{ (T_\text{w} - H) \frac{\sigma_\text{o}^2\text{min}\{\sigma_w^2, \sigma_z^2, \sigma_u^2\}}{2} + (t - T_\text{w} + 1)\frac{\sigma_\text{c}^2\text{min}\{\sigma_w^2,\sigma_z^2,\sigma_{\eta_{t-1}}^2\}}{8}}}\\
        & \leq \frac{\sqrt{n_y ||\Sigma_e|| \left( \log(1/\delta) + \frac{H(n_u + n_y)}{2} \log \left( \frac{\lambda (n_u+n_y)H + (t - H + 1)\max\{\Upsilon_\text{w}^2, {\Upsilon_\text{ac}(t)}^2\}}{\lambda (n_u+n_y)H} \right) \right)} + \frac{t}{T^2} \sqrt{H} + \sqrt{\lambda} \Bar{\mathbf{m}}}{\sqrt{(T_\text{w} - H) \frac{\sigma_\text{o}^2\text{min}\{\sigma_w^2, \sigma_z^2, \sigma_u^2\}}{2} + (t - T_\text{w} + 1)\frac{\sigma_\text{c}^2\text{min}\{\sigma_w^2,\sigma_z^2,\sigma_{\eta_{t-1}}^2\}}{8}}}\\
        & \leq \frac{\sqrt{n_y ||\Sigma_e|| \left( \log(1/\delta) + \frac{H(n_u + n_y)}{2} \log \left( \frac{\lambda (n_u+n_y)H + T \max\{\Upsilon_\text{w}^2, {\Upsilon_\text{ac}(t)}^2\}}{\lambda (n_u+n_y)H} \right) \right)} + \frac{\sqrt{H}}{T} + \sqrt{\lambda} \Bar{\mathbf{m}}}{\sqrt{t - H + 1}\sqrt{ \min \left\{ \frac{\sigma_\text{o}^2\text{min}\{\sigma_w^2, \sigma_z^2, \sigma_u^2\}}{2}, \frac{\sigma_\text{c}^2\text{min}\{\sigma_w^2,\sigma_z^2,\sigma_{\eta_{t-1}}^2\}}{8}\right\}}},\\
    \end{split}
    \label{eq: Markov parameter estimation error during adaptive control phase}
\end{equation}
which holds under the event $\mathcal{E}_{\phi,\text{warm}}\cap \mathcal{E}_{\phi,\text{ac}} \cap \mathcal{E}_\text{PE, warm} \cap \mathcal{E}_\text{PE, ac} \cap \mathcal{E}_{E_t}$ that holds with a probability of at least $1 - 3\delta$. Now, $\delta$ can be re-parameterized to $\delta \rightarrow \delta/3$ to ensure that the above inequality holds with a  probability of at least $1 - \delta$. 

\subsection*{Proof for the rest of the episodes:}

From the above guarantee, we see that the parameter estimation error decreases after the $0^\text{th}$ episode since the denominator in \eqref{eq: Markov parameter estimation error during adaptive control phase} grows faster than the numerator. This implies that given the event $\mathcal{E}_{\mathbf{M}_{T_\text{w}}} $ in \eqref{eq:event_def_shi}, the following event
\begin{equation*}
    \mathcal{E}_{\mathbf{M}_{2T_\text{w}}} \coloneqq \left\{ || \mathbf{\hat{M}}_t - \mathbf{M} || \leq || \mathbf{\hat{M}}_{T_\text{w}} - \mathbf{M} || \leq 1, \text{ for } t \in [2T_\text{w}, 4T_\text{w} - 1] \text{ and }T_\text{w} \geq \Bar{T}_\text{w} \right\},
\end{equation*}
for the $1^\text{st}$ episode holds with a probability of at least $1 - \delta$. The time step $\Bar{T}_\text{w}$ is as defined in \eqref{eq: definition for LB on T_w}. Now we continue with the induction argument. We will assume that for the $k^\text{th}$ episode, 
\begin{equation*}
    \mathcal{E}_{\mathbf{M}_{l_k}} \coloneqq \left\{ || \mathbf{\hat{M}}_{t} - \mathbf{M} || \leq || \mathbf{\hat{M}}_{T_\text{w}} - \mathbf{M} || \leq 1, \text{ for } t \in [l_{k}, l_{k+1}) \text{ and }T_\text{w} \geq \Bar{T}_\text{w} \right\},
\end{equation*}
holds with a probability of at least $1 - \delta$. Now, for the ${k+1}^\text{th}$ episode, the above proofs for the bound on the norm of the system signals, then the persistence of excitation condition and finally, the bound on the norm of the Markov parameter estimation error after the $k^\text{th}$ episode, follows analogously. Therefore, we establish that for the ${k+1}^\text{th}$ episode, 
\begin{equation*}
    \mathcal{E}_{\mathbf{M}_{l_{k+1}}} \coloneqq \left\{ || \mathbf{\hat{M}}_{t} - \mathbf{M} || \leq || \mathbf{\hat{M}}_{T_\text{w}} - \mathbf{M} || \leq 1, \text{ for } t \in [l_{k+1}, l_{k+2}) \text{ and }T_\text{w} \geq \Bar{T}_\text{w} \right\},
\end{equation*}
holds with a probability of at least $1 - \delta$.  This implies that the Markov parameter estimation error is monotonically decreasing. It must be noted that for every episode, we have a new lower bound on $T_\text{w}$ since $T_\text{w} \geq T_\text{ac}(t)$. Therefore, we let $T_\text{w} \geq T_\text{ac}(T) = T_\text{ac}$. This implies that we can provide a uniform bound for the system signals for all $t \geq T_\text{w}$, where $\Bar{\mathcal{X}}(T), X_\text{est,ac}(T), U_\text{ac}(T), X_\text{ac}(T), Y_\text{ac}(T) = \Bar{\mathcal{X}}, X_\text{est,ac}, U_\text{ac}, X_\text{ac}, Y_\text{ac}$, respectively, with a slight abuse of notations. This concludes the proof.

\rem This inductive argument was adopted to address a `circular argument' that is evident in \cite{lale2021adaptive}. Instead of episode-wise analysis, the proof for the persistent excitation of all episodes, the one for the signal bound of all episodes, and the one for modeling error over all episodes are conducted separately in \cite{lale2021adaptive}. However, in the proof for the persistence of excitation condition \cite[Lemma 3.1]{lale2021adaptive}, the exact bounds on $||y_t||$ and $||u_t||$ are used when bounding $||\phi_t||$. But, providing the bounds on $||y_t||$ and $||u_t||$ requires a bound on $||\hat{x}_{t|t-1,\Theta} - \hat{x}_{t|t-1,\hat{\Theta}_{t-1}}||$, which requires the guarantee that the model parameter estimation error is monotonically decreasing. Recall that the persistence of excitation condition is required to establish that the model parameter estimation error is monotonically decreasing. Therefore, from a broader perspective, we can see that the proof methodology adopted to establish decreasing parameter estimation error in \cite{lale2021adaptive} inherently requires the assumption that the parameter estimation error is decreasing. This creates what we call a `circular argument'.

\hfill$\blacksquare$

\subsection{Proof of Theorem \ref{Theorem: regret LBC}}

To recall from \eqref{eq: cumulative regret}, we are trying to minimize the following definition of regret:
\vspace{0mm}
\begin{equation*}
    \begin{split}
        \mathrm{Regret}(T) & = \sum_{t = 0}^{T-1} c_t - TJ_* \text{, where}\\
        c_t & = y_t^\top Q y_t + u_t^\top R u_t.
    \end{split}
\end{equation*}
Since the adaptive control policy is deployed in an episodic fashion, as described in Algorithm \ref{algorithm: LQG-NAIVE}, the regret is also analyzed episode-wise, i.e., the cumulative difference between the (sub)optimal cost incurred by the adaptive control policy and the optimal long-term average expected cost $J_*$, is upper bounded for every episode. This bound is then summed over the number of episodes to obtain the final regret upper bound. The relation between the long-term average expected cost and the solution to a Lyapunov equation is given in \eqref{eq: result on LT avg cost - Lyapunov equation}. This relation is critical for establishing the regret upper bound. To recall, the estimated model parameter is maintained during the length of each episode. Therefore, for the sake of brevity, let us denote $\hat{\Theta}_{l_k} = \hat{\Theta}$, $\hat{K}_{l_k} = \hat{K}$, and $\sigma_{\eta_{l_k}}^2 = \sigma_{\eta}^2$, where $l_k$ is the time step at the start of the $k^\text{th}$ episode.

We will first decompose the cost. Consider the following decomposition of the cost at time step $t$:
\vspace{0mm}
\begin{equation}
    \begin{split}
        y_t^\top Q y_t + u_t^\top R u_t & = \left(C x_t + z_t \right)^\top Q \left(C x_t + z_t \right) + u_t^\top R u_t\\
        & = x_t^\top{C}^\top Q C x_t + u_t^\top R u_t+ 2z_t^\top QC x_t + z_t^\top Q z_t\\
        & = \underbrace{x_t^\top{C}^\top Q C x_t + \hat{x}_{t|t,\hat{\Theta}}^\top \hat{K}^\top R \hat{K} \hat{x}_{t|t,\hat{\Theta}}}_{c_{t,1}}  +  \underbrace{\eta_t^\top R \eta_t -2 \eta_t^\top R \hat{K} \hat{x}_{t|t,\hat{\Theta}} + 2z_t^\top QC x_t + z_t^\top Q z_t}_{c_{t,2}}.
    \end{split}
    \label{eq: cost decomposition at time t}
\end{equation}
We will upper bound $\sum_{t = 0}^{t-1}c_{t,1}$ and $\sum_{t=0}^{t-1} c_{t,2}$ separately. To avoid ambiguity in the present analysis, it must be noted that we are not analyzing the $0^\text{th}$ episode: the time step starts at $0$ just for the sake of convenience.

\subsection*{Upper bounding $\sum_{t = 0}^{t-1}c_{t,1}$}

With some abuse of notations in the proof of Lemma \ref{Lemma: PE LBC}, we redefine $\mathbf{\hat{G}_1}, \mathbf{\hat{G}_2}, \text{ and } \mathbf{\hat{G}_3}$ in the following. From \eqref{eq: measurement and time update}, we have
\vspace{0mm}
\begin{equation*}
    \begin{split}
        \underbrace{
        \begin{bmatrix}
            x_t \\ \hat{x}_{t|t,\hat{\Theta}}
        \end{bmatrix}}_{\Bar{x}_t}
        &
        =
        \underbrace{
        \begin{bmatrix}
            A & -B \hat{K}\\
            \hat{L} C A & \left(I - \hat{L} \hat{C} \right) \left( \hat{A} - \hat{B} \hat{K} \right) - \hat{L} C B \hat{K} 
        \end{bmatrix}}_{\mathbf{\hat{G}_1}}
        \begin{bmatrix}
            x_{t-1} \\ \hat{x}_{t-1|t-1,\hat{\Theta}} 
        \end{bmatrix}
        \\
        & \hspace{4mm}
        +
        \underbrace{
        \begin{bmatrix}
            I & 0\\
            \hat{L} C & \hat{L} 
        \end{bmatrix}}_{\mathbf{\hat{G}_2}}
        \underbrace{
        \begin{bmatrix}
            w_{t-1} \\ z_t \\ 
        \end{bmatrix}}_{\Bar{\epsilon}_{t-1}}
        +
        \underbrace{
        \begin{bmatrix}
            B\\
            \left( I - \hat{L} \hat{C} \right)\hat{B} + \hat{L} C B
        \end{bmatrix}}_{\mathbf{\hat{G}_3}}
        \eta_{t-1}\\
        \implies \Bar{x}_t & = \mathbf{\hat{G}_1} \Bar{x}_{t-1} + \mathbf{\hat{G}_2} \Bar{\epsilon}_{t-1} + \mathbf{\hat{G}_3}\eta_{t-1}.
    \end{split}
\end{equation*}
This implies,
\vspace{0mm}
\begin{equation*}
    \begin{split}
        \Bar{x}_t & = \mathbf{\hat{G}_1}^t \Bar{x}_{0} + \sum_{i = 0}^{t-1} \mathbf{\hat{G}_1}^{t-i-1}\mathbf{\hat{G}_2} \Bar{\epsilon}_{i} + \sum_{i = 0}^{t-1} \mathbf{\hat{G}_1}^{t-i-1}\mathbf{\hat{G}_3}\eta_{i}.
    \end{split}
\end{equation*}
For simplicity of exposition, let us define for $l \in \mathbb{N}$ and $j,l \geq i$ \cite{simchowitz2020naive},
\vspace{0mm}
\begin{equation}
    \begin{split}
        & \text{Col}_{i,j}(A) \coloneqq 
        \begin{bmatrix}
            \mathbb{I}_{i \geq 1} A^{i-1}\\ \mathbb{I}_{i \geq 2} A^{i-2}\\.\\.\\.\\ \mathbb{I}_{i \geq j} A^{i-j}
        \end{bmatrix}, \hspace{2mm}
        \text{Toep}_{i,j,l}(A) \coloneqq 
        \begin{bmatrix}
            A^i \mathbb{I}_{i \geq 0} & A^{i+1} \mathbb{I}_{i \geq -1} & ... & A^{i+l} \mathbb{I}_{i \geq -l}\\
            A^{i-1} \mathbb{I}_{i \geq 1} & A^{i} \mathbb{I}_{i \geq 0} & ... & A^{i+l-1} \mathbb{I}_{i \geq 1-l}\\
            .\\
            .\\
            .\\
            A^{i-j} \mathbb{I}_{i \geq j} & A^{i-j+1} \mathbb{I}_{i \geq j-1} & ... & A^{i+l-j} \mathbb{I}_{i \geq j-l}\\
        \end{bmatrix},\\
    \end{split}
    \label{eq: Definition of Toep, Col}
\end{equation}
and $\text{diag}_{t}(A) \coloneqq I_t \otimes A$, where $\mathbb{I}$ is the indicator function and with a slight abuse of notations, we have $I_t$ as the identity matrix with $t$ rows. This implies
\vspace{0mm}
\begin{equation*}
    \begin{split}
        \underbrace{
    \begin{bmatrix}
       \Bar{x}_{t-1} \\ \Bar{x}_{t-2} \\. \\. \\. \\ \Bar{x}_0 
    \end{bmatrix}}_{\Bar{x}_{[t-1:0]}} & = 
    \begin{bmatrix}
        \mathbf{\hat{G}_1}^{t-1} \\ \mathbf{\hat{G}_1}^{t-2} \\. \\. \\. \\ I
    \end{bmatrix} \Bar{x}_{0} + 
    \begin{bmatrix}
        0& I & \mathbf{\hat{G}_1} & \mathbf{\hat{G}_1}^2 & . & . & . & \mathbf{\hat{G}_1}^{t-2}\\
        0 & 0 & I & \mathbf{\hat{G}_1} & . & . & . & \mathbf{\hat{G}_1}^{t-3}\\
        . \\
        . \\
        . \\
        0 & 0 & 0 & 0 & . & . & . & I\\
        0 & 0 & 0 & 0 & . & . & . & 0\\
    \end{bmatrix}
    I_t \otimes \mathbf{\hat{G}_2} 
    \underbrace{
    \begin{bmatrix}
       \Bar{\epsilon}_{t-1} \\ \Bar{\epsilon}_{t-2} \\. \\. \\. \\ \Bar{\epsilon}_0 
    \end{bmatrix}}_{\Bar{\epsilon}_{[t-1:0]}}\\
    & + 
    \begin{bmatrix}
        0 & I & \mathbf{\hat{G}_1} & \mathbf{\hat{G}_1}^2 & . & . & . & \mathbf{\hat{G}_1}^{t-2}\\
        0 & 0 & I & \mathbf{\hat{G}_1} & . & . & . & \mathbf{\hat{G}_1}^{t-3}\\
        . \\
        . \\
        . \\
        0 & 0 & 0 & 0 & . & . & . & I\\
        0 & 0 & 0 & 0 & . & . & . & 0\\
    \end{bmatrix}
    I_t \otimes \mathbf{\hat{G}_3} 
    \underbrace{
    \begin{bmatrix}
       \eta_{t-1} \\ \eta_{t-2} \\. \\. \\. \\ \eta_0 
    \end{bmatrix}}_{\eta_{[t-1:0]}}\\
    \end{split}
\end{equation*}
\begin{equation}
\begin{split}
    \implies \Bar{x}_{[t-1:0]} & = \text{Col}_{t,t}(\mathbf{\hat{G}_1}) \Bar{x}_0 +  \text{Toep}_{-1,t-1,t-1}(\mathbf{\hat{G}_1})\text{diag}_{t}(\mathbf{\hat{G}_2})\Bar{\epsilon}_{[t-1:0]}\\
    & \hspace{4mm} + \text{Toep}_{-1,t-1,t-1}(\mathbf{\hat{G}_1})\text{diag}_{t}(\mathbf{\hat{G}_3})\eta_{[t-1:0]}.
\end{split}
\end{equation}

Finally, $\sum_{t = 0}^{T-1}c_{t,1}$ can be decomposed as:
\vspace{0mm}
\begin{equation*}
    \begin{split}
        \sum_{t = 0}^{t-1}c_{t,1} & = \sum_{t = 0}^{t-1} x_t^\top{C}^\top Q C x_t + \hat{x}_{t|t,\hat{\Theta}}^\top \hat{K}^\top R \hat{K} \hat{x}_{t|t,\hat{\Theta}}\\
        & = \sum_{t = 0}^{t-1} \Bar{x}_t^\top \mathbf{\Bar{W}} \Bar{x}_t\\
        & = \Bar{x}_{[t-1:0]}^\top \text{diag}_t(\mathbf{\Bar{W}}) \Bar{x}_{[t-1:0]}\\
        & = \Bar{x}_0^\top \underbrace{ \text{Col}_{t,t}^\top(\mathbf{\hat{G}_1}) \text{diag}_t(\mathbf{\Bar{W}}) \text{Col}_{t,t}(\mathbf{\hat{G}_1})}_{\Lambda_{\Bar{x}_0}} \Bar{x}_0\\
        & \hspace{4mm} + \Bar{\epsilon}_{[t-1:0]}^\top \underbrace{\text{diag}_{t}^\top(\mathbf{\hat{G}_2}) \text{Toep}_{-1,t-1,t-1}^\top(\mathbf{\hat{G}_1}) \text{diag}_t(\mathbf{\Bar{W}})\text{Toep}_{-1,t-1,t-1}(\mathbf{\hat{G}_1})\text{diag}_{t}(\mathbf{\hat{G}_2})}_{\Lambda_{\Bar{\epsilon}}} \Bar{\epsilon}_{[t-1:0]}\\
        & \hspace{4mm} + \eta_{[t-1:0]}^\top \underbrace{\text{diag}_{t}^\top(\mathbf{\hat{G}_3}) \text{Toep}_{-1,t-1,t-1}^\top(\mathbf{\hat{G}_1}) \text{diag}_t(\mathbf{\Bar{W}})\text{Toep}_{-1,t-1,t-1}(\mathbf{\hat{G}_1})\text{diag}_{t}(\mathbf{\hat{G}_3})}_{\Lambda_{\eta}}\eta_{[t-1:0]}\\
    \end{split}
\end{equation*}
\vspace{0mm}
\begin{equation}
    \begin{split}
        & \hspace{16mm} + 2 \Bar{\epsilon}_{[t-1:0]}^\top \underbrace{\text{diag}_{t}^\top(\mathbf{\hat{G}_2}) \text{Toep}_{-1,t-1,t-1}^\top(\mathbf{\hat{G}_1}) \text{diag}_t(\mathbf{\Bar{W}}) \text{Col}_{t,t}(\mathbf{\hat{G}_1}) }_{\Lambda_{\text{cross,1}}}\Bar{x}_0\\
        & \hspace{16mm} + 2\eta_{[t-1:0]}^\top \underbrace{\text{diag}_{t}^\top(\mathbf{\hat{G}_3}) \text{Toep}_{-1,t-1,t-1}^\top(\mathbf{\hat{G}_1}) \text{diag}_t(\mathbf{\Bar{W}})\text{Col}_{t,t}(\mathbf{\hat{G}_1}) }_{\Lambda_{\text{cross,2}}}\Bar{x}_0\\
        & \hspace{16mm} + 2\Bar{\epsilon}_{[t-1:0]}^\top \underbrace{\text{diag}_{t}^\top(\mathbf{\hat{G}_2}) \text{Toep}_{-1,t-1,t-1}^\top(\mathbf{\hat{G}_1}) \text{diag}_t(\mathbf{\Bar{W}})\text{Toep}_{-1,t-1,t-1}(\mathbf{\hat{G}_1})\text{diag}_{t}(\mathbf{\hat{G}_3})}_{\Lambda_{\text{cross,3}}}\eta_{[t-1:0]},
    \end{split}
    \label{eq: decomposition of the cost - part 1}
\end{equation}
where $\mathbf{\Bar{W}} = \text{diag}\left({C}^\top Q C, \hat{K}^\top R \hat{K} \right)$. We will now upper bound each of the above terms individually.

\subsection*{Bounding $\Bar{\epsilon}_{[t-1:0]}^\top \Lambda_{\Bar{\epsilon}} \Bar{\epsilon}_{[t-1:0]}$:}

This term can be upper-bounded using the Hanson-Wright inequality. To do so, firstly we require the following bounds. From Lemma \ref{Lemma B.12}, we have
\vspace{0mm}
\begin{equation*}
    \begin{split}
        \left| \left|\Lambda_{\Bar{\epsilon}} \right| \right| & \leq\left| \left| \mathbf{\hat{G}_2} \right| \right|^2 \left| \left| \mathbf{\Bar{W}} \right| \right| \left| \left| \text{Toep}_{-1,t-1,t-1}(\mathbf{\hat{G}_1}) \right| \right|^2\\
        & \leq \left| \left| \mathbf{\hat{G}_2} \right| \right|^2 \left| \left| \mathbf{\Bar{W}} \right| \right| \left| \left|\mathbf{\hat{G}_1} \right| \right|^2_{\mathcal{H}_ \infty}.\\
    \end{split}
\end{equation*}
The above bound depends on the estimated system parameter $\hat{\Theta}$. To use the Hanson-Wright inequality, we need a bound on $\left| \left|\Lambda_{\Bar{\epsilon}} \right| \right|$ that is not a random variable. Now,
\vspace{0mm}
\begin{equation*}
    \begin{split}
        \left| \left| \mathbf{\hat{G}_2} \right| \right| & \leq \left| \left| \mathbf{\hat{G}_2} - \mathbf{G_2} \right| \right| + \left| \left| \mathbf{G_2} \right| \right|\\
        & \leq \left| \left| \begin{bmatrix}
            0 & 0\\ (\hat{L} - L)C & (\hat{L} - L)
        \end{bmatrix} \right| \right| + \left| \left| \begin{bmatrix}
            I & 0\\ LC & L
        \end{bmatrix} \right| \right|\\
        & \leq c_{\mathbf{\hat{G}_2}},\\
    \end{split}
\end{equation*}
for some $c_{\mathbf{\hat{G}_2}} > 0$. Based on the proofs on Lemmas \ref{Lemma: bound on the states and outputs during LBC phase} - \ref{Lemma: Markov parameter estimation error during LBC phase}, we will define an event $\mathcal{E}_{\mathbf{M}}$, where
\vspace{0mm}
\begin{equation}
    \mathcal{E}_{\mathbf{M}} \coloneqq \left\{ || \mathbf{\hat{M}}_t - \mathbf{M} || \leq || \mathbf{\hat{M}}_{T_\text{w}} - \mathbf{M} || \leq 1, \text{ for all } t \geq T_\text{w} \geq \Bar{T}_\text{w} \right\},
\end{equation}
where $\mathbb{P}\{\mathcal{E}_{\mathbf{M}} \} \geq 1 - \delta$ and $\Bar{T}_\text{w}$ is as defined in \eqref{eq: definition for LB on T_w}. This event is a direct consequence of Lemma \ref{Lemma: Markov parameter estimation error during LBC phase}. 

Similarly, we can bound $\left| \left|\mathbf{\hat{G}_1} \right| \right|_{\mathcal{H}_ \infty} \leq c_{\mathbf{\hat{G}_1}}$ for some $c_{\mathbf{\hat{G}_1}} > 0$ under the event $\mathcal{E}_{\mathbf{M}}$. For Hanson-Wright inequality, we also require the following bound, which exists under the event $\mathcal{E}_{\mathbf{M}}$:
\vspace{0mm}
\begin{equation*}
    \begin{split}
         \left| \left|\Lambda_{\Bar{\epsilon}} \right| \right|_\mathrm{F} & = \sqrt{\text{Tr}(\Lambda_{\Bar{\epsilon}} \Lambda_{\Bar{\epsilon}}^\top)}\\
         & \lesssim \sqrt{t(n_x + n_y)} \coloneqq \Bar{c}_{\Lambda_{\Bar{\epsilon}}}.
    \end{split}
\end{equation*}
Consider the following:
\vspace{0mm}
\begin{equation*}
\begin{split}
    \left| \left| \mathbf{\Bar{W}} \right| \right| & = \left| \left| \begin{bmatrix} {C}^\top Q C & 0\\ 0 & \hat{K}^\top R \hat{K} \end{bmatrix} \right| \right|\\
    & \leq \left| \left| {C}^\top Q C +  \hat{K}^\top R \hat{K} \right| \right|\\
    & \leq \left| \left| C\right| \right|^2 \left| \left| Q \right| \right| + \Gamma^2 \left| \left| R \right| \right| \coloneqq c_{\mathbf{\Bar{W}}},\\
\end{split} 
\end{equation*}
which holds under the event $\mathcal{E}_{\mathbf{M}}$ for some $c_{\mathbf{\Bar{W}}} > 0$. Finally, from the Hanson-Wright inequality as defined in Theorem \ref{Theorem B.3 Hanson-Wright inequality}, we have
\vspace{0mm}
\begin{equation*}
    \mathbb{P}\left\{ \Bar{\epsilon}_{[t-1:0]}^\top \Lambda_{\Bar{\epsilon}} \Bar{\epsilon}_{[t-1:0]} - \mathbb{E}\left[ \Bar{\epsilon}_{[t-1:0]}^\top \Lambda_{\Bar{\epsilon}} \Bar{\epsilon}_{[t-1:0]} \right] > t \right\} \leq  2\exp \left[ -c \min \left( \frac{t^2}{a^4 \Bar{c}^2_{\Lambda_{\Bar{\epsilon}}}}, \frac{t}{a^2 c^2_{\mathbf{\hat{G}_2}}c^2_{\mathbf{\hat{G}_1}}c_{\mathbf{\Bar{W}}}}\right) \right],
\end{equation*}
where $c$ is an absolute positive constant. Since, $\Bar{\epsilon}_{[t-1:0]}$ consists of Gaussian random variables, the constant $a$ exists. Now we can simplify the above expression as follows. Let,
\vspace{0mm}
\begin{equation*}
    \begin{split}
        \frac{\delta}{2} & = \exp \left[ -c \min \left( \frac{t^2}{a^4 \Bar{c}^2_{\Lambda_{\Bar{\epsilon}}}}, \frac{t}{a^2c^2_{\mathbf{\hat{G}_2}}c^2_{\mathbf{\hat{G}_1}}c_{\mathbf{\Bar{W}}}}\right) \right]\\
        \frac{1}{c}\log\left( \frac{2}{\delta}\right) & = \min \left( \frac{t^2}{a^4 \Bar{c}^2_{\Lambda_{\Bar{\epsilon}}}}, \frac{t}{a^2c^2_{\mathbf{\hat{G}_2}}c^2_{\mathbf{\hat{G}_1}}c_{\mathbf{\Bar{W}}}}\right)\\
        \implies t & = a^2 \Bar{c}_{\Lambda_{\Bar{\epsilon}}} \sqrt{\frac{1}{c}\log\left( \frac{2}{\delta}\right)} \text{ or } t = \frac{a^2 c^2_{\mathbf{\hat{G}_2}}c^2_{\mathbf{\hat{G}_1}}c_{\mathbf{\Bar{W}}}}{c}\log\left( \frac{2}{\delta}\right).\\
    \end{split}
\end{equation*}

This implies with a probability of at least $1 - \delta$,
\vspace{0mm}
\begin{equation}
\begin{split}
    \Bar{\epsilon}_{[t-1:0]}^\top \Lambda_{\Bar{\epsilon}} \Bar{\epsilon}_{[t-1:0]} & \leq \mathbb{E}\left[  \Bar{\epsilon}_{[t-1:0]}^\top \Lambda_{\Bar{\epsilon}} \Bar{\epsilon}_{[t-1:0]}\right] + \mathcal{O}\left(\Bar{c}_{\Lambda_{\Bar{\epsilon}}}\sqrt{\log\left( \frac{2}{\delta}\right)} + c^2_{\mathbf{\hat{G}_2}}c^2_{\mathbf{\hat{G}_1}}c_{\mathbf{\Bar{W}}}\log\left( \frac{2}{\delta}\right)\right)\\
    \implies \Bar{\epsilon}_{[t-1:0]}^\top \Lambda_{\Bar{\epsilon}} \Bar{\epsilon}_{[t-1:0]} & \leq \text{Tr}\left(\Lambda_{\Bar{\epsilon}} \text{diag}_{t}\left( \begin{bmatrix} \sigma_w^2 I & 0 \\ 0 & \sigma_z^2 I \end{bmatrix} \right) \right)\\
    & \hspace{4mm} + \mathcal{O}\left(\Bar{c}_{\Lambda_{\Bar{\epsilon}}} \sqrt{\log\left( \frac{2}{\delta}\right)} + c^2_{\mathbf{\hat{G}_2}}c^2_{\mathbf{\hat{G}_1}}c_{\mathbf{\Bar{W}}}\log\left( \frac{2}{\delta}\right)\right).\\
\end{split}
\label{eq: regret bound - quadratic of noise}
\end{equation}

Now, we will focus on $\text{Tr}\left(\Lambda_{\Bar{\epsilon}} \text{diag}_{t}\left( \begin{bmatrix} \sigma_w^2 I & 0 \\ 0 & \sigma_z^2 I \end{bmatrix} \right) \right)$, where
\begin{equation*}
    \begin{split}
        & \text{Tr}\left(\Lambda_{\Bar{\epsilon}} \text{diag}_{t}\left( \begin{bmatrix} \sigma_w^2 I & 0 \\ 0 & \sigma_z^2 I \end{bmatrix} \right) \right)\\
        & = \text{Tr}\left( \text{diag}_{t}^\top(\mathbf{\hat{G}_2}) \text{Toep}_{-1,t-1,t-1}^\top(\mathbf{\hat{G}_1}) \text{diag}_t(\mathbf{\Bar{W}})\text{Toep}_{-1,t-1,t-1}(\mathbf{\hat{G}_1})\text{diag}_{t}(\mathbf{\hat{G}_2})\text{diag}_{t}\left( \begin{bmatrix} \sigma_w^2 I & 0 \\ 0 & \sigma_z^2 I \end{bmatrix} \right)\right).\\
    \end{split}
\end{equation*}

Now, consider the following:
\vspace{0mm}
\begin{equation*}
    \begin{split}
        & \text{Tr}\left( \text{Toep}_{-1,t-1,t-1}^\top(\mathbf{\hat{G}_1}) \text{diag}_t(\mathbf{\Bar{W}})\text{Toep}_{-1,t-1,t-1} (\mathbf{\hat{G}_1}) \right)\\
        & = \text{Tr}\left( \begin{bmatrix} \text{Col}_{0,t}(\mathbf{\hat{G}_1}) & ... & \text{Col}_{t-1,t}(\mathbf{\hat{G}_1}) \end{bmatrix}^\top  \text{diag}_t(\mathbf{\Bar{W}}) \begin{bmatrix} \text{Col}_{0,t}(\mathbf{\hat{G}_1}) & ... & \text{Col}_{t-1,t}(\mathbf{\hat{G}_1}) \end{bmatrix} \right)\\
        & = \sum_{i = 0}^{t-1} \text{Col}_{i,t}^\top (\mathbf{\hat{G}_1}) \text{diag}_t(\mathbf{\Bar{W}}) \text{Col}_{i,t}(\mathbf{\hat{G}_1})\\
        & \leq t \cdot \text{dlyap}\left( \mathbf{\hat{G}_1}, \mathbf{\Bar{W}}\right),
    \end{split}
\end{equation*}
where the last inequality comes from Corollary \ref{Corollary B.11} and $\text{dlyap}(\cdot)$ is as defined in Definition \ref{Definition B.1}. Let $\hat{S} = \text{dlyap}\left( \mathbf{\hat{G}_1}, \mathbf{\Bar{W}}\right)$. Further, we have the following relation. For any positive semi-definite matrices $X$, $Y$, and any matrix $P$, if $X \leq Y$ then, $P^\top X P \leq P^\top Y P \implies \text{Tr}(P^\top X P) \leq \text{Tr}(P^\top Y P)$. Further, for another diagonal matrix $Z$, $\text{Tr}(P^\top X P Z) = \text{Tr}(Z^{1/2} P^\top X P Z^{1/2}) \leq \text{Tr}(Z^{1/2} P^\top Y P Z^{1/2})$. Considering these relations, we have
\vspace{0mm}
\begin{equation*}
    \begin{split}
        & \text{Tr}\left(\Lambda_{\Bar{\epsilon}} \text{diag}_{t}\left( \begin{bmatrix} \sigma_w^2 I & 0 \\ 0 & \sigma_z^2 I \end{bmatrix} \right) \right)\\
        & = \text{Tr}\left( \text{diag}_{t}^\top(\mathbf{\hat{G}_2}) \text{Toep}_{-1,t-1,t-1}^\top(\mathbf{\hat{G}_1}) \text{diag}_t(\mathbf{\Bar{W}})\text{Toep}_{-1,t-1,t-1}(\mathbf{\hat{G}_1})\text{diag}_{t}(\mathbf{\hat{G}_2})\text{diag}_{t}\left( \begin{bmatrix} \sigma_w^2 I & 0 \\ 0 & \sigma_z^2 I \end{bmatrix} \right)\right)\\
        & \leq t \text{Tr}\left( \mathbf{\hat{G}_2}^\top \text{dlyap}\left( \mathbf{\hat{G}_1}, \mathbf{\Bar{W}}\right)\mathbf{\hat{G}_2}\begin{bmatrix} \sigma_w^2 I & 0 \\ 0 & \sigma_z^2 I \end{bmatrix}\right)\\
        & = t \text{Tr}\left( \mathbf{\hat{G}_2}^\top \hat{S}\mathbf{\hat{G}_2}\begin{bmatrix} \sigma_w^2 I & 0 \\ 0 & \sigma_z^2 I \end{bmatrix}\right)\\
        & = t J_s(\hat{\Theta}),
    \end{split}
\end{equation*}
where the last equality comes from \eqref{eq: result on LT avg cost - Lyapunov equation} and to recall, $J_s(\hat{\Theta})$ is an alternate formulation of the LQG control problem defined in \eqref{eq: LQG control problem with quadratic state}. Finally, from \eqref{eq: regret bound - quadratic of noise}, we have the following under the event $\mathcal{E}_{\mathbf{M}}$:
\vspace{0mm}
\begin{equation}
    \begin{split}
         \Bar{\epsilon}_{[t-1:0]}^\top \Lambda_{\Bar{\epsilon}} \Bar{\epsilon}_{[t-1:0]} & \leq t J_s(\hat{\Theta}) + \mathcal{O}\left( \sqrt{t(n_x + n_y)\log\left( \frac{2}{\delta}\right)} + \left( c^2_{\mathbf{\hat{G}_2}} c^2_{\mathbf{\hat{G}_1}}c_{\mathbf{\Bar{W}}} \right)\log\left( \frac{2}{\delta}\right) \right)\\
         & \leq t J_s(\hat{\Theta}) + \mathcal{O}\left[ \left(\sqrt{t(n_x + n_y)\log\left( \frac{2}{\delta}\right)} + \log\left( \frac{2}{\delta}\right) \right) \left( c^2_{\mathbf{\hat{G}_2}} c^2_{\mathbf{\hat{G}_1}} c_{\mathbf{\Bar{W}}} \right) \right],\\
    \end{split}
    \label{eq: cost part 1 - bound 1}
\end{equation}
which holds with a probability of at least $1 - \delta$.

\subsection*{Bounding $\eta_{[t-1:0]}^\top \Lambda_{\eta} \eta_{[t-1:0]}$}

This term can also be upper-bounded using the Hanson-Wright inequality. Under the event $\mathcal{E}_{\mathbf{M}}$, and from Lemma \ref{Lemma B.12}, we have
\vspace{0mm}
\begin{equation*}
    \begin{split}
        \left| \left|\Lambda_{\eta} \right| \right| & \leq \left| \left| \mathbf{\hat{G}_3} \right| \right|^2 \left| \left| \mathbf{\Bar{W}} \right| \right| \left| \left| \text{Toep}_{-1,t-1,t-1}(\mathbf{\hat{G}_1}) \right| \right|^2\\
        & \leq c^2_{\mathbf{\hat{G}_3}} c^2_{\mathbf{\hat{G}_1}} c_{\mathbf{\Bar{W}}},\\
    \end{split}
\end{equation*}
for some $c^2_{\mathbf{\hat{G}_3}} > 0$. Now from Theorem \ref{Theorem B.3 Hanson-Wright inequality}, we have 
\vspace{0mm}
\begin{equation}
\begin{split}
 \eta_{[t-1:0]}^\top \Lambda_{\eta} \eta_{[t-1:0]} & \leq \text{Tr}\left(\Lambda_{\eta} \text{diag}_{t}\left( \sigma_{\eta}^2 I \right) \right)  + \mathcal{O}\left( \sqrt{tn_u\log\left( \frac{2}{\delta}\right)} + c^2_{\mathbf{\hat{G}_3}} c^2_{\mathbf{\hat{G}_1}} c_{\mathbf{\Bar{W}}}\log\left( \frac{2}{\delta}\right)\right),\\
\end{split}
\label{eq: regret bound - quadratic of external excitation}
\end{equation}
which holds with a probability of at least $1 - \delta$. Now, we focus on $\text{Tr}\left(\Lambda_{\eta} \text{diag}_{t}\left( \sigma_{\eta}^2 I \right) \right)$. It is a standard fact that, for any positive semi-definite matrix $X$ and any matrix $Y$, $\text{Tr}(XY) \leq \text{Tr}(X)||Y||$. Now from Corollary \ref{Corollary B.11}, we have the following under the event $\mathcal{E}_{\mathbf{M}}$:
\vspace{0mm}
\begin{equation*}
\begin{split}
    \text{Tr}\left(\Lambda_{\eta} \text{diag}_{t}\left( \sigma_{\eta}^2 I \right) \right) & \leq 
    \text{Tr}\left(\Lambda_{\eta} \right) \sigma_{\eta}^2\\
    & = \text{Tr}\left( \text{diag}_{t}(\mathbf{\hat{G}_3})^\top \text{Toep}_{-1,t-1,t-1}(\mathbf{\hat{G}_1})^\top \text{diag}_t(\mathbf{\Bar{W}})\text{Toep}_{-1,t-1,t-1}(\mathbf{\hat{G}_1})\text{diag}_{t}(\mathbf{\hat{G}_3}) \right)\sigma_{\eta}^2\\
    & \leq t \text{Tr}\left( \mathbf{\hat{G}_3}^\top \text{dlyap}\left( \mathbf{\hat{G}_1}, \mathbf{\Bar{W}}\right)\mathbf{\hat{G}_3}\right)\sigma_{\eta}^2 \\
    & \leq t n_u\left| \left|  \mathbf{\hat{G}_3}^\top \hat{S}\mathbf{\hat{G}_3} \right| \right|\sigma_{\eta}^2 \\
    & \leq tn_u c^2_{\mathbf{\hat{G}_3}}  \left| \left| \hat{S} \right| \right|\sigma_{\eta}^2\\
    & \lesssim tn_u c^2_{\mathbf{\hat{G}_3}}  \left| \left| S \right| \right|\sigma_{\eta}^2.\\
\end{split}
\end{equation*}
In the above expression, $S$ depends on the true model parameters $\Theta$. In the literature, the perturbation bounds of the solutions to Lyapunov equations have been addressed as a function of $\tau(\hat{\Theta}, \Theta)$ \cite{lale2020regret_prior, mania2019certainty, simchowitz2020naive}. Therefore, a perturbation bound on $||\hat{S} - S||$ exists under the event $\mathcal{E}_{\mathbf{M}}$. This perturbation bound can be exploited to upper bound $||\hat{S}|| \leq c_S||S||$ for some $c_S > 0$. 

Finally, under the event $\mathcal{E}_{\mathbf{M}}$, and from \eqref{eq: regret bound - quadratic of external excitation}, we have
\vspace{0mm}
\begin{equation}
    \begin{split}
        \eta_{[t-1:0]}^\top \Lambda_{\eta} \eta_{[t-1:0]} &  \leq tn_u c^2_{\mathbf{\hat{G}_3}}  \left| \left| S \right| \right|\sigma_{\eta}^2 + \mathcal{O} \left( \sqrt{tn_u\log\left( \frac{2}{\delta}\right)} + \left(c^2_{\mathbf{\hat{G}_3}} c^2_{\mathbf{\hat{G}_1}} c_{\mathbf{\Bar{W}}} \right)\log\left( \frac{2}{\delta}\right)\right)\\
        &  \leq tn_u c^2_{\mathbf{\hat{G}_3}}  \left| \left| S \right| \right|\sigma_{\eta}^2 + \mathcal{O} \left[ \left( \sqrt{tn_u\log\left( \frac{2}{\delta}\right)} + \log\left( \frac{2}{\delta}\right)\right) \left( c^2_{\mathbf{\hat{G}_3}} c^2_{\mathbf{\hat{G}_1}} c_{\mathbf{\Bar{W}}}\right) \right],\\
    \end{split}
    \label{eq: cost part 1 - bound 2}
\end{equation}
which holds with a probability of at least $1 - \delta$.

\subsection*{Bounding $\Bar{x}_0^\top \Lambda_{\Bar{x}_0} \Bar{x}_0$:}

From Corollary \ref{Corollary B.11}, we have
\vspace{0mm}
\begin{equation}
    \begin{split}
        \Bar{x}_0^\top \Lambda_{\Bar{x}_0} \Bar{x}_0 & = \Bar{x}_0^\top \text{Col}_{t,t}^\top(\mathbf{\hat{G}_1}) \text{diag}_t(\mathbf{\Bar{W}}) \text{Col}_{t,t}^\top(\mathbf{\hat{G}_1}) \Bar{x}_0\\
        & \leq \Bar{x}_0^\top \text{dlyap}\left(\mathbf{\hat{G}_1}, \mathbf{\Bar{W}} \right) \Bar{x}_0\\
        & = \Bar{x}_0^\top \hat{S} \Bar{x}_0\\
        & \leq || \Bar{x}_0 ||^2 ||\hat{S}||\\
        & \lesssim || \Bar{x}_0 ||^2 ||S||\\
        & \lesssim \left( X_\text{ac}^2 + \Bar{\mathcal{X}}^2 \right) ||S||,\\
    \end{split}
    \label{eq: cost part 1 - bound 3}
\end{equation}
which holds with a probability of at least $1 - \delta$ under the event $\mathcal{E}_{\mathbf{M}}$. The last inequality comes from Lemma \ref{Lemma: bound on the states and outputs during LBC phase}. It must be noted that although $x_0 \sim \mathcal{N}(0, \Sigma)$ is assumed, we do not consider that for the bound in \eqref{eq: cost part 1 - bound 3} since $x_0$ here represents the state of the system at the start of an episode. For the sake of convenience, define the event $\mathcal{E}_x$ given $\mathcal{E}_{\mathbf{M}}$:
\begin{equation*}
    \mathcal{E}_x \coloneqq \left\{ ||x_t|| \leq X_\text{ac}, ||\hat{x}_{t|t,\hat{\Theta}}|| \leq \Bar{\mathcal{X}} \right\},
\end{equation*}
which holds with a probability of at least $1 - \delta$. The event $\mathcal{E}_x$ is a consequence of Lemma \ref{Lemma: bound on the states and outputs during LBC phase}. Under the event $\mathcal{E}_x$, the bound in \eqref{eq: cost part 1 - bound 3} becomes deterministic.

\subsection*{Bounding $2 \Bar{\epsilon}_{[t-1:0]}^\top \Lambda_{\text{cross,1}}\Bar{x}_0$:}

Firstly, notice that
\vspace{0mm}
\begin{equation*}
    \begin{split}
        \mathbb{E}\left[ 2 \Bar{\epsilon}_{[t-1:0]}^\top \Lambda_{\text{cross,1}}\Bar{x}_0 \right] & = 0,\\
        \text{Var}\left( 2 \Bar{\epsilon}_{[t-1:0]}^\top \Lambda_{\text{cross,1}}\Bar{x}_0 \right) & = 4\Bar{x}_0^\top \Lambda_{\text{cross,1}}^\top \text{diag}_t\left( \begin{bmatrix}
            \sigma^2_w I & 0\\ 0 & \sigma^2_z I
        \end{bmatrix}\right) \Lambda_{\text{cross,1}} \Bar{x}_0.\\
    \end{split}
\end{equation*}
It can be verified that there exists a matrix $X$ such that $XX^\top = \Lambda_{\Bar{\epsilon}}$ and a matrix $Y$ such that $YY^\top = \Lambda_{\Bar{x}_0}$ to obtain $\Lambda_{\text{cross,1}} = XY^\top$ \cite{simchowitz2020naive}. Then,
\vspace{0mm}
\begin{equation*}
\begin{split}
        \left| \left| \Lambda_{\text{cross,1}}\Bar{x}_0 \right| \right| & = \sqrt{\Bar{x}_0^\top YX^\top XY^\top \Bar{x}_0}\\
        & \leq \sqrt{\left| \left| X^\top X \right| \right| \cdot \Bar{x}_0^\top YY^\top \Bar{x}_0}\\
        & \leq \sqrt{\left| \left| \Lambda_{\Bar{\epsilon}} \right| \right| \cdot \Bar{x}_0^\top S \Bar{x}_0},
\end{split}
\end{equation*}
which holds since $\Lambda_{\Bar{\epsilon}}$ is symmetric positive semi-definite. The last inequality comes from \eqref{eq: cost part 1 - bound 3}. Now, from arithmetic mean-geometric mean inequality, we have the following under the event $\mathcal{E}_{\mathbf{M}} \cap \mathcal{E}_x$:
\vspace{0mm}
\begin{equation*}
    \begin{split}
        \left| \left| \Lambda_{\text{cross,1}}\Bar{x}_0 \right| \right| & \lesssim \left| \left| \Lambda_{\Bar{\epsilon}} \right| \right| + \Bar{x}_0^\top S \Bar{x}_0\\
        & \lesssim c^2_{\mathbf{\hat{G}_2}} c^2_{\mathbf{\hat{G}_1}} c_{\mathbf{\Bar{W}}}  + \Bar{x}_0^\top S \Bar{x}_0\\
        & \lesssim c^2_{\mathbf{\hat{G}_2}} c^2_{\mathbf{\hat{G}_1}} c_{\mathbf{\Bar{W}}} + || \Bar{x}_0 ||^2 ||S||\\
        & \lesssim c^2_{\mathbf{\hat{G}_2}} c^2_{\mathbf{\hat{G}_1}} c_{\mathbf{\Bar{W}}} + \left( X_\text{ac}^2 + \Bar{\mathcal{X}}^2 \right) ||S||.
    \end{split}
\end{equation*}
Now observe that $2 \Bar{\epsilon}_{[t-1:0]}^\top \Lambda_{\text{cross,1}}\Bar{x}_0$ is $c_{L_1}c^2_{\mathbf{\hat{G}_2}} c^2_{\mathbf{\hat{G}_1}} c_{\mathbf{\Bar{W}}} + \left( X_\text{ac}^2 + \Bar{\mathcal{X}}^2 \right) ||S||$ - Lipschitz, where $c_{L_1} > 0$ is some constant dependent on the true model parameters $\Theta$. Using Lemma \ref{Lemma B.5}, we have 
\vspace{0mm}
\begin{equation}
    \begin{split}
        & 2 \Bar{\epsilon}_{[t-1:0]}^\top \Lambda_{\text{cross,1}}\Bar{x}_0\\
        & \lesssim 2\sqrt{2\max\{\sigma^2_w,\sigma^2_z\}}\left( c^2_{\mathbf{\hat{G}_2}} c^2_{\mathbf{\hat{G}_1}} c_{\mathbf{\Bar{W}}} \sqrt{\log\left( \frac{2}{\delta} \right)} + \left( X_\text{ac}^2 + \Bar{\mathcal{X}}^2 \right) ||S|| \sqrt{\log\left( \frac{2}{\delta} \right)}\right)\\
        & \lesssim \max\{\sigma_w,\sigma_z\}\left( c^2_{\mathbf{\hat{G}_2}} c^2_{\mathbf{\hat{G}_1}} c_{\mathbf{\Bar{W}}} \sqrt{\log\left( \frac{2}{\delta} \right)} + \left( X_\text{ac}^2 + \Bar{\mathcal{X}}^2 \right) ||S|| \sqrt{\log\left( \frac{2}{\delta} \right)}\right),\\
    \end{split}
\end{equation}
which holds with a probability of at least $1 - \delta$ under the event $\mathcal{E}_{\mathbf{M}} \cap \mathcal{E}_x$.

\subsection*{Bounding $2 \eta_{[t-1:0]}^\top \Lambda_{\text{cross,2}}\Bar{x}_0$:}

In a similar fashion to the previous cross-term, we obtain the following bound under the event $\mathcal{E}_{\mathbf{M}} \cap \mathcal{E}_x$:
\vspace{0mm}
\begin{equation}
    \begin{split}
         & 2 \eta_{[t-1:0]}^\top \Lambda_{\text{cross,2}}\Bar{x}_0\\
         & \lesssim \sigma_\eta c^2_{\mathbf{\hat{G}_3}} c^2_{\mathbf{\hat{G}_1}} c_{\mathbf{\Bar{W}}} \sqrt{\log\left( \frac{2}{\delta} \right)} + \sigma_\eta\left( X_\text{ac}^2 + \Bar{\mathcal{X}}^2 \right) ||S|| \sqrt{\log\left( \frac{2}{\delta} \right)},\\
    \end{split}
    \label{eq: cost part 1 -  bound cross term 2}
\end{equation}
which holds with a probability of at least $1 - \delta$.

\subsection*{Bounding $2\Bar{\epsilon}_{[t-1:0]}^\top \Lambda_{\text{cross,3}} \eta_{[t-1:0]}$:}

Recalling from \eqref{eq: decomposition of the cost - part 1}, we have
\vspace{0mm}
\begin{equation*}
\begin{split}
    & 2\Bar{\epsilon}_{[t-1:0]}^\top \Lambda_{\text{cross,3}} \eta_{[t-1:0]}\\
    & = 2\Bar{\epsilon}_{[t-1:0]}^\top \text{diag}_{t}(\mathbf{\hat{G}_2})^\top \text{Toep}_{-1,t-1,t-1}(\mathbf{\hat{G}_1})^\top \text{diag}_t(\mathbf{\Bar{W}})\text{Toep}_{-1,t-1,t-1}(\mathbf{\hat{G}_1})\text{diag}_{t}(\mathbf{\hat{G}_3})\eta_{[t-1:0]}.
\end{split}
\end{equation*}
Now in a similar fashion as the previous cross terms, we have the following from the arithmetic mean-geometric mean inequality:
\vspace{0mm}
\begin{equation*}
    \begin{split}
        \left| \left| \Lambda_{\text{cross,3}} \eta_{[t-1:0]} \right| \right| & \leq \sqrt{\left| \left| \Lambda_{\Bar{\epsilon}} \right| \right| \cdot  \eta_{[t-1:0]}^\top \Lambda_{\eta} \eta_{[t-1:0]}}\\
        & \lesssim \left| \left| \Lambda_{\Bar{\epsilon}} \right| \right| +  \eta_{[t-1:0]}^\top \Lambda_{\eta} \eta_{[t-1:0]}.
    \end{split}
\end{equation*}
From \eqref{eq: cost part 1 - bound 2}, we have
\vspace{0mm}
\begin{equation*}
    \begin{split}
        & \left| \left| \Lambda_{\text{cross,3}} \eta_{[t-1:0]} \right| \right|\\
        & \lesssim c^2_{\mathbf{\hat{G}_2}} c^2_{\mathbf{\hat{G}_1}} c_{\mathbf{\Bar{W}}}  + tn_u c^2_{\mathbf{\hat{G}_3}} \left| \left| S \right| \right|\sigma_{\eta}^2 + \left( \sqrt{tn_u\log\left( \frac{2}{\delta}\right)} + \log\left( \frac{2}{\delta}\right)\right) \left(c^2_{\mathbf{\hat{G}_3}} c^2_{\mathbf{\hat{G}_1}} c_{\mathbf{\Bar{W}}} \right),
    \end{split}
\end{equation*}
which holds with a probability of at least $1 - \delta$ under the event $\mathcal{E}_{\mathbf{M}} \cap \mathcal{E}_x$. The above bound is indeed deterministic since $t$ in the above bound represents the number of time steps in a particular episode, which is known a priori. Now using Lemma \ref{Lemma B.5} on the entire term, we have
\vspace{0mm}
\begin{equation}
    \begin{split}
        & 2\Bar{\epsilon}_{[t-1:0]}^\top \Lambda_{\text{cross,3}} \eta_{[t-1:0]}\\
        & \lesssim \max\{\sigma_w,\sigma_z\} c^2_{\mathbf{\hat{G}_2}} c^2_{\mathbf{\hat{G}_1}}c_{\mathbf{\Bar{W}}} \sqrt{\log\left( \frac{2}{\delta} \right)} + tn_u c^2_{\mathbf{\hat{G}_3}} \left| \left| S \right| \right|\sigma_{\eta}^2 \max\{\sigma_w,\sigma_z\} \sqrt{\log\left( \frac{2}{\delta} \right)}\\
        & \hspace{4mm} + \left( \sqrt{tn_u}\log\left( \frac{2}{\delta}\right) + \log^2\left( \frac{2}{\delta}\right)\right) \left( c^2_{\mathbf{\hat{G}_3}} c^2_{\mathbf{\hat{G}_1}} c_{\mathbf{\Bar{W}}} \right) \max\{\sigma_w,\sigma_z\},
    \end{split}
\end{equation}
which holds with a probability of at least $1 - 2\delta$ under the event $\mathcal{E}_{\mathbf{M}} \cap \mathcal{E}_x$.

\subsection*{Putting things together:}

Finally, we have the following bound under the event $\mathcal{E}_{\mathbf{M}} \cap \mathcal{E}_x$:
\vspace{0mm}
\begin{equation}
    \begin{split}
         \sum_{t = 0}^{t-1}c_{t,1} & \lesssim  t J_s(\hat{\Theta}) + tn_u c^2_{\mathbf{\hat{G}_3}}  \left| \left| S \right| \right|\sigma_{\eta}^2 \left( 1 + \max\{\sigma_w,\sigma_z\} \sqrt{\log\left( \frac{2}{\delta} \right)} \right)\\
         & \hspace{4mm} + \max\{\sigma_w,\sigma_z, \sigma_\eta \} \left( X_\text{ac}^2 + \Bar{\mathcal{X}}^2 \right) ||S||\left(1 + \sqrt{\log\left( \frac{2}{\delta} \right)}\right)\\
         & \hspace{4mm} + \left(\sqrt{\Bigg( t(n_x + n_y) + \max\{\sigma^2_w,\sigma^2_z\}\Bigg)\log\left( \frac{2}{\delta}\right)} + \log\left( \frac{2}{\delta}\right) \right) \left( c^2_{\mathbf{\hat{G}_2}} c^2_{\mathbf{\hat{G}_1}} c_{\mathbf{\Bar{W}}}\right)\\
         & \hspace{4mm} + \left( \sqrt{tn_u\log\left( \frac{2}{\delta}\right)} + \left( \sqrt{tn_u\max\{\sigma^2_w,\sigma^2_z\}} +1 \right)\log\left( \frac{2}{\delta}\right) \right)\left(c^2_{\mathbf{\hat{G}_3}} c^2_{\mathbf{\hat{G}_1}} c_{\mathbf{\Bar{W}}} \right)\\
         & \hspace{4mm} + \left( \sigma_\eta \sqrt{\log\left( \frac{2}{\delta} \right)} + \max\{\sigma_w,\sigma_z\}\log^2\left( \frac{2}{\delta}\right)\right) \left( c^2_{\mathbf{\hat{G}_3}} c^2_{\mathbf{\hat{G}_1}} c_{\mathbf{\Bar{W}}} \right),\\
    \end{split}
    \label{eq: cost part 1}
\end{equation}
which holds with a probability of at least $1 - 6\delta$.

\subsection*{Upper bounding $\sum_{t = 0}^{t-1}c_{t,2}$}

To recall, from \eqref{eq: cost decomposition at time t} we have
\vspace{0mm}
\begin{equation*}
    \sum_{t = 0}^{t-1}c_{t,2} = \eta_t^\top R \eta_t - 2 \eta_t^\top R \hat{K} \hat{x}_{t|t,\hat{\Theta}} + 2z_t^\top Q C x_t + z_t^\top Q z_t.
\end{equation*}

\subsection*{Upper bounding $\sum_{t = 0}^{t-1} 2z_t^\top QC x_t$:}

Bounding this term follows analogously to the previous cross-terms. Under the event $\mathcal{E}_x$, we have:
\vspace{0mm}
\begin{equation*}
\begin{split}
    \left| \left| \text{diag}_t(QC) x_{[t-1:0]}\right| \right| & \leq \sqrt{||Q|| \cdot x_{[t-1:0]}^\top \text{diag}_t({C}^\top Q C) x_{[t-1:0]}}\\
    & \leq \sqrt{||Q|| \cdot ||\text{diag}_t({C}^\top Q C)|| \cdot ||x_{[t-1:0]}||^2}\\
    & \leq \sqrt{||Q|| \cdot ||{C}^\top Q C||} \cdot \sqrt{t}X_\text{ac}.\\
\end{split}
\end{equation*}
Now from Lemma \ref{Lemma B.5}, we have
\vspace{0mm}
\begin{equation}
    \begin{split}
        \sum_{t = 0}^{t-1} 2z_t^\top QC x_t & \lesssim \sqrt{t}\sigma_z ||Q||\cdot ||C|| X_\text{ac} \sqrt{\log \left( \frac{2}{\delta}\right)},
    \end{split}
    \label{eq: cost part 2 - bound 1}
\end{equation}
which holds with a probability of at least $1 - \delta$ under the event $\mathcal{E}_x$.

\subsection*{Upper bounding $ - \sum_{t=0}^{t-1} 2 \eta_t^\top R \hat{K} \hat{x}_{t|t,\hat{\Theta}}$:}

Let us define, $\hat{x}_{[t-1:0]} \coloneqq \begin{bmatrix} \hat{x}_{t-1|t-1,\hat{\Theta}}^\top &...& \hat{x}_{0|0,\hat{\Theta}}^\top \end{bmatrix}^\top$. Bounding this term is again similar to the previous cross-terms. Under the event $\mathcal{E}_x$, we have:
\vspace{0mm}
\begin{equation*}
    \begin{split}
        \left| \left| \text{diag}_t(R\hat{K}) \hat{x}_{[t-1:0]} \right| \right| & \leq \sqrt{ \left| \left| R \right| \right| \cdot \hat{x}_{[t-1:0]}^\top \text{diag}_t(\hat{K}^\top R\hat{K}) \hat{x}_{[t-1:0]} }\\
        & \leq \sqrt{ \left| \left| R \right| \right| \cdot \left| \left| \text{diag}_t(\hat{K}^\top R\hat{K}) \right| \right| \cdot || \hat{x}_{[t-1:0]} ||^2}\\
        & \leq \left| \left| R \right| \right| \Gamma \sqrt{t}\Bar{\mathcal{X}}.
    \end{split}
\end{equation*}
From Lemma \ref{Lemma B.5}, we have
\vspace{0mm}
\begin{equation}
    \begin{split}
         \sum_{t=0}^{t-1} - 2 \eta_t^\top R \hat{K} \hat{x}_{t|t,\hat{\Theta}} & = 2 (-\eta_{[t-1:0]})^\top \text{diag}_t(R\hat{K}) \hat{x}_{[t-1:0]}\\
         & \lesssim \sigma_\eta \sqrt{t} \left| \left| R \right| \right| \Gamma\Bar{\mathcal{X}} \sqrt{\log \left( \frac{2}{\delta}\right)},
    \end{split}
    \label{eq: cost part 2 - bound 2}
\end{equation}
which holds with a probability of at least $1-\delta$ under the event $\mathcal{E}_{\mathbf{M}} \cap \mathcal{E}_x$.

\subsection*{Upper bounding $\sum_{t=0}^{t-1} \eta_t^\top R \eta_t$:}

From Hanson-Wright inequality (refer to Theorem \ref{Theorem B.3 Hanson-Wright inequality}), we have
\vspace{0mm}
\begin{equation}
    \begin{split}
        & \sum_{t=0}^{t-1} \eta_t^\top R \eta_t = \eta_{[t-1:0]}^\top \text{diag}_t(R) \eta_{[t-1:0]}\\
        & \leq t n_u\sigma^2_\eta \text{Tr} \left( R\right) + \mathcal{O} \left( \sqrt{tn_u\log \left( \frac{2}{\delta} \right)} + ||R|| \log \left( \frac{2}{\delta} \right) \right)\\
        & \leq t n_u\sigma^2_\eta \text{Tr} \left( R\right) + \mathcal{O} \left( \sqrt{tn_u\log \left( \frac{2}{\delta} \right)} + \log \left( \frac{2}{\delta} \right) \right)||R||,\\
    \end{split}
    \label{eq: cost part 2 - bound 3}
\end{equation}
which holds with a probability of at least $1-\delta$.

\subsection*{Upper bounding $\sum_{t=0}^{t-1} z_t^\top Q z_t$:}

From Hanson-Wright inequality, we have
\vspace{0mm}
\begin{equation}
    \begin{split}
        & \sum_{t=0}^{t-1} z_t^\top Q z_t = z_{[t-1:0]}^\top \text{diag}_t(Q) z_{[t-1:0]}\\
        & \leq tn_y\sigma^2_z \text{Tr} \left( Q\right) + \mathcal{O} \left( \sqrt{tn_x\log \left( \frac{2}{\delta} \right)} + ||Q|| \log \left( \frac{2}{\delta} \right) \right)\\
        & \leq t n_y\sigma^2_z \text{Tr} \left( Q\right) + \mathcal{O} \left( \sqrt{tn_x\log \left( \frac{2}{\delta} \right)} + \log \left( \frac{2}{\delta} \right) \right)||Q||,\\
    \end{split}
    \label{eq: cost part 2 - bound 4}
\end{equation}
which holds with a probability of at least $1-\delta$.

\subsection*{Putting things together:}

We have the following bound under the event $\mathcal{E}_x \cap \mathcal{E}_{\mathbf{M}}$:
\vspace{0mm}
\begin{equation}
    \begin{split}
        \sum_{t = 0}^{t-1}c_{t,2} & \lesssim \sqrt{t}\sigma_z ||Q||\cdot ||C|| X_\text{ac} \sqrt{\log \left( \frac{2}{\delta}\right)} + \sigma_\eta \sqrt{t} \left| \left| R \right| \right| \Gamma\Bar{\mathcal{X}} \sqrt{\log \left( \frac{2}{\delta}\right)}\\
        & \hspace{4mm} + t n_u\sigma^2_\eta \text{Tr} \left( R\right) +  \left( \sqrt{tn_u\log \left( \frac{2}{\delta} \right)} + \log \left( \frac{2}{\delta} \right) \right) ||R||\\
        & \hspace{4mm} + t n_y\sigma^2_z \text{Tr} \left( Q\right) + \left( \sqrt{tn_x\log \left( \frac{2}{\delta} \right)} + \log \left( \frac{2}{\delta} \right) \right)||Q||,
    \end{split}
    \label{eq: cost part 2}
\end{equation}
which holds with a probability of at least $1-4\delta$.

\subsection*{Final upper-bound on the cumulative cost:}

Combining \eqref{eq: cost part 1} and \eqref{eq: cost part 2}, we have
\vspace{0mm}
\begin{equation}
    \begin{split}
        & \sum_{t=0}^{t-1}  y_t^\top Q y_t + u_t^\top R u_t\\
        & \lesssim  t J_s(\hat{\Theta}) + tn_u c^2_{\mathbf{\hat{G}_3}}  \left| \left| S \right| \right|\sigma_{\eta}^2 \left( 1 + \max\{\sigma_w,\sigma_z\} \sqrt{\log\left( \frac{2}{\delta} \right)} \right)\\
         & \hspace{4mm} + \max\{\sigma_w,\sigma_z, \sigma_\eta \} \left( X_\text{ac}^2 + \Bar{\mathcal{X}}^2 \right) ||S||\left(1 + \sqrt{\log\left( \frac{2}{\delta} \right)}\right)\\
         & \hspace{4mm} + \left(\sqrt{\Bigg( t(n_x + n_y) + \max\{\sigma^2_w,\sigma^2_z\}\Bigg)\log\left( \frac{2}{\delta}\right)} + \log\left( \frac{2}{\delta}\right) \right) \left( c^2_{\mathbf{\hat{G}_2}} c^2_{\mathbf{\hat{G}_1}} c_{\mathbf{\Bar{W}}}\right)\\
         & \hspace{4mm} + \left( \sqrt{tn_u\log\left( \frac{2}{\delta}\right)} + \left( \sqrt{tn_u\max\{\sigma^2_w,\sigma^2_z\}} +1 \right)\log\left( \frac{2}{\delta}\right) \right)\left(c^2_{\mathbf{\hat{G}_3}} c^2_{\mathbf{\hat{G}_1}} c_{\mathbf{\Bar{W}}} \right)\\
         & \hspace{4mm} + \left( \sigma_\eta \sqrt{\log\left( \frac{2}{\delta} \right)} + \max\{\sigma_w,\sigma_z\}\log^2\left( \frac{2}{\delta}\right)\right) \left( c^2_{\mathbf{\hat{G}_3}} c^2_{\mathbf{\hat{G}_1}} c_{\mathbf{\Bar{W}}} \right)\\
         & \hspace{4mm} + \sqrt{t}\sigma_z ||Q||\cdot ||C|| X_\text{ac} \sqrt{\log \left( \frac{2}{\delta}\right)} + \sigma_\eta \sqrt{t} \left| \left| R \right| \right| \Gamma\Bar{\mathcal{X}} \sqrt{\log \left( \frac{2}{\delta}\right)}\\
        & \hspace{4mm} + t n_u\sigma^2_\eta \text{Tr} \left( R\right) +  \left( \sqrt{tn_u\log \left( \frac{2}{\delta} \right)} + \log \left( \frac{2}{\delta} \right) \right) ||R||\\
        & \hspace{4mm} + t n_y\sigma^2_z \text{Tr} \left( Q\right) + \left( \sqrt{tn_x\log \left( \frac{2}{\delta} \right)} + \log \left( \frac{2}{\delta} \right) \right)||Q|| \coloneqq c_{\text{cost},k},
    \end{split}
    \label{eq: upper bound on the cumulative cost over one episode}
\end{equation}
which holds with a probability of at least $1-10\delta$ under the event $\mathcal{E}_{\mathbf{M}} \cap \mathcal{E}_x$. Essentially, $c_{\text{cost},k}$ is the upper bound on the cumulative cost incurred during the $k^\text{th}$ episode. Further, there exists an event $\mathcal{E}_\text{cost}$, which holds with a probability of at least $1 - 10\delta$ such that, on $\mathcal{E}_{\mathbf{M}} \cap \mathcal{E}_x \cap \mathcal{E}_\text{cost}$ the following bound holds:
\vspace{0mm}
\begin{equation*}
     \sum_{t=0}^{t-1}  y_t^\top Q y_t + u_t^\top R u_t \lesssim c_{\text{cost},k}.
\end{equation*}
Now, $\mathcal{E}_{\mathbf{M}} \cap \mathcal{E}_x \cap \mathcal{E}_\text{cost}$ holds with a probability of at least $1 - 12\delta$. We can re-parametrise $\delta \rightarrow \frac{\delta}{12\log_2(T / T_\text{w})}$. The reason for doing so will become apparent in the following section. The bound in \eqref{eq: upper bound on the cumulative cost over one episode} holds for one episode. This bound can be used to determine the upper bound on the cumulative cost for the entire horizon by summing the established bound in \eqref{eq: upper bound on the cumulative cost over one episode} over the number of episodes, i.e., by taking the union bound. This will be addressed in the following section. 

\subsection*{Regret upper bound}

To recall, the regret is defined as such:
\vspace{0mm}
\begin{equation*}
    \mathrm{Regret}(T) = \sum_{t=0}^{T-1}  (y_t^\top Q y_t + u_t^\top R u_t - J_*).
\end{equation*}
To recall, the system parameters are estimated at the start of each episode. Hence, the system parameter being used during the $k^\text{th}$ episode is denoted as $\Hat{\Theta}_{l_k}$, where $l_k$ is the time step at the start of the $k^{\text{th}}$ episode. Further, as a reminder, the number of time steps in each episode is double the previous episode. Hence, we can approximate the number of episodes to be $\lfloor \log_2(T/T_\text{w}) \rfloor$. From \eqref{eq: upper bound on the cumulative cost over one episode}, we have
\begin{equation*}
    \begin{split}
        \mathrm{Regret}(T) & \lesssim \sum_{k=0}^{\lfloor \log_2(T/T_\text{w}) \rfloor-1} l_k \left(J_s(\hat{\Theta}_{l_k}) -  J_*\right) + l_k n_u c^2_{\mathbf{\hat{G}_3}}  \left| \left| S \right| \right|\sigma_{\eta_{l_k}}^2 \left( 1 + \max\{\sigma_w,\sigma_z\} \sqrt{\log\left( \frac{2}{\delta} \right)} \right)\\
         & \hspace{4mm} + \max\{\sigma_w,\sigma_z, \sigma_{\eta_{l_k}} \} \left( X_\text{ac}^2 + \Bar{\mathcal{X}}^2 \right) ||S||\left(1 + \sqrt{\log\left( \frac{2}{\delta} \right)}\right)\\
         & \hspace{4mm} + \left(\sqrt{\Bigg( l_k(n_x + n_y) + \max\{\sigma^2_w,\sigma^2_z\}\Bigg)\log\left( \frac{2}{\delta}\right)} + \log\left( \frac{2}{\delta}\right) \right) \left( c^2_{\mathbf{\hat{G}_2}} c^2_{\mathbf{\hat{G}_1}} c_{\mathbf{\Bar{W}}}\right)\\
         & \hspace{4mm} + \left( \sqrt{l_k n_u\log\left( \frac{2}{\delta}\right)} + \left( \sqrt{l_k n_u\max\{\sigma^2_w,\sigma^2_z\}} +1 \right)\log\left( \frac{2}{\delta}\right) \right)\left(c^2_{\mathbf{\hat{G}_3}} c^2_{\mathbf{\hat{G}_1}} c_{\mathbf{\Bar{W}}} \right)\\
    \end{split}
\end{equation*}
\vspace{0mm}
\begin{equation}
    \begin{split}
         & \hspace{4mm} + \left( \sigma_{\eta_{l_k}} \sqrt{\log\left( \frac{2}{\delta} \right)} + \max\{\sigma_w,\sigma_z\}\log^2\left( \frac{2}{\delta}\right)\right) \left( c^2_{\mathbf{\hat{G}_3}} c^2_{\mathbf{\hat{G}_1}} c_{\mathbf{\Bar{W}}} \right)\\
         & \hspace{4mm} + \sqrt{l_k}\sigma_z ||Q||\cdot ||C|| X_\text{ac} \sqrt{\log \left( \frac{2}{\delta}\right)} + \sigma_{\eta_{l_k}} \sqrt{l_k} \left| \left| R \right| \right| \Gamma\Bar{\mathcal{X}} \sqrt{\log \left( \frac{2}{\delta}\right)}\\
        & \hspace{4mm} + l_k n_u\sigma^2_{\eta_{l_k}} \text{Tr} \left( R\right) +  \left( \sqrt{l_k n_u\log \left( \frac{2}{\delta} \right)} + \log \left( \frac{2}{\delta} \right) \right) ||R||\\
        & \hspace{4mm} + l_k n_y\sigma^2_z \text{Tr} \left( Q\right) + \left( \sqrt{l_k n_x\log \left( \frac{2}{\delta} \right)} + \log \left( \frac{2}{\delta} \right) \right)||Q||,
    \end{split}
    \label{eq: penultimate regret upper bound}
\end{equation}
which holds under the event $\mathcal{E}_{\mathbf{M}} \cap \mathcal{E}_x \cap \mathcal{E}_\text{cost}$. Now, we will refine the above bound. To recall, $\sigma^2_{\eta_{l_k}} = \frac{\gamma}{\sqrt{l_k}}$. Furthermore, based on the result in Lemma \ref{Lemma: Markov parameter estimation error during LBC phase} and Lemma \ref{Lemma: confidence bound on the system parameters}, we have the following result \cite[Th. 4]{mania2019certainty}:
\vspace{0mm}
\begin{equation*}
    \begin{split}
        J(\hat{\Theta}_k) -  J_* & = J_s(\hat{\Theta}_k) - J_s(\Theta)\\
        & \leq c_{\Theta} \tau(\hat{\Theta}_t, \Theta)^2,\\
    \end{split}
\end{equation*}
where $c_{\Theta} > 0$ is some constant dependent on the true model parameter $\Theta$. The above bound holds when $\tau(\hat{\Theta}_t, \Theta) \leq \epsilon_J$, which is satisfied when $t \geq T_\text{w} \geq T_J$ for some $T_J > H$. Now, we have
\vspace{0mm}
\begin{equation*}
    \begin{split}
        J_s(\hat{\Theta}_k) -  J_* & = J_s(\hat{\Theta}_k) - J_s(\Theta) - \sigma^2_z n_y \text{Tr}(Q)\\
        & \lesssim c_{\Theta}\left(\frac{1}{\sqrt{l_k}} \right)^2 -  \sigma^2_z n_y \text{Tr}(Q),\\
    \end{split}
\end{equation*}
where the last inequality comes from Lemma \ref{Lemma: Markov parameter estimation error during LBC phase}. Combining the above result with the bound in \eqref{eq: penultimate regret upper bound}, we obtain the final regret upper bound:
\begin{equation*}
    \begin{split}
        \mathrm{Regret}(T) & \lesssim \lfloor \log_2(T/T_\text{w}) \rfloor c_{\Theta} + \sqrt{T} \gamma n_u c^2_{\mathbf{\hat{G}_3}}  \left| \left| S \right| \right|\left( 1 + \max\{\sigma_w,\sigma_z\} \sqrt{\log\left( \frac{2}{\delta} \right)} \right)\\
         & \hspace{4mm} + \lfloor \log_2(T/T_\text{w}) \rfloor \max\{\sigma_w,\sigma_z, T^{-1/4}\}\left( X_\text{ac}^2 + \Bar{\mathcal{X}}^2 \right) ||S||\left(1 + \sqrt{\log\left( \frac{2}{\delta} \right)}\right)\\
         & \hspace{4mm} + \left(\sqrt{\Bigg( T(n_x + n_y) + \max\{\sigma^2_w,\sigma^2_z\}\Bigg)\log\left( \frac{2}{\delta}\right)} + \log\left( \frac{2}{\delta}\right) \right) \left( c^2_{\mathbf{\hat{G}_2}} c^2_{\mathbf{\hat{G}_1}} c_{\mathbf{\Bar{W}}}\right)\\
         & \hspace{4mm} + \left( \sqrt{T n_u\log\left( \frac{2}{\delta}\right)} + \left( \sqrt{T n_u\max\{\sigma^2_w,\sigma^2_z\}} +1 \right)\log\left( \frac{2}{\delta}\right) \right)\left(c^2_{\mathbf{\hat{G}_3}} c^2_{\mathbf{\hat{G}_1}} c_{\mathbf{\Bar{W}}} \right)\\
        & \hspace{4mm} + \lfloor \log_2(T/T_\text{w}) \rfloor \left( \sqrt{n_u \gamma} T^{-1/4} \sqrt{\log\left( \frac{2}{\delta} \right)} + \max\{\sigma_w,\sigma_z\}\log^2\left( \frac{2}{\delta}\right)\right) \left( c^2_{\mathbf{\hat{G}_3}} c^2_{\mathbf{\hat{G}_1}} c_{\mathbf{\Bar{W}}} \right)\\
    \end{split}
\end{equation*}
\vspace{0mm}
\begin{equation}
    \begin{split}
         & + \sqrt{T}\sigma_z ||Q||\cdot ||C|| X_\text{ac} \sqrt{\log \left( \frac{2}{\delta}\right)} + \sqrt{T \gamma} \left| \left| R \right| \right| \Gamma\Bar{\mathcal{X}} \sqrt{\log \left( \frac{2}{\delta}\right)}\\
        & + \sqrt{T}\gamma n_u \text{Tr} \left( R\right) +  \left( \sqrt{T n_u\log \left( \frac{2}{\delta} \right)} + \log \left( \frac{2}{\delta} \right) \right) ||R||\\
        & + \left( \sqrt{T n_x\log \left( \frac{2}{\delta} \right)} + \log \left( \frac{2}{\delta} \right) \right)||Q||,
    \end{split}
    \label{eq: final regret upper bound}
\end{equation}
which holds under the event $\mathcal{E}_{\mathbf{M}} \cap \mathcal{E}_x \cap \mathcal{E}_\text{cost}$. The regret bound in \eqref{eq: final regret upper bound} suggests that $\mathrm{Regret}(T) = \Tilde{\mathcal{O}}(\sqrt{T})$ with a probability of at least $1 - \delta$. This concludes the proof. \hfill$\blacksquare$

\subsection{Technical background related to the warm-up phase} \label{appendix: technical background warm-up phase}

\lem[\cite{lale2020regret_prior}] \label{Lemma: bound on state and input during warm-up} Let $\Xi(A)$ be as defined in Definition \ref{Definition: Phi(A)}. For any $\delta \in (0,1/6)$, with a probability of at least $1 - \delta/6$, the following bounds hold when controlling the system as defined in \eqref{eq: LQG state space equation} with $u_t \sim \mathcal{N}(0, \sigma^2_uI)$ for all $t \in [0,T_\text{w}-1]$:
\vspace{0mm}
\begin{equation}
    \begin{split}
        ||x_t|| \leq X_\text{w},\hspace{1mm} ||u_t|| \leq U_\text{w}, \hspace{1mm} ||z_t|| \leq Z,
    \end{split}
\end{equation}
\vspace{0mm}
where
\vspace{0mm}
\begin{equation}
    \begin{split}
        X_\text{w} & \coloneqq (\sigma_w + \sigma_u ||B||)\frac{\Xi(A)\rho({A})}{\sqrt{1 - \rho(A)^2}}\sqrt{2n_x\text{ log}(12n_xT_\text{w}/\delta)},\\
        U_\text{w} & \coloneqq \sigma_u \sqrt{2n_u\text{ log}(12n_uT_\text{w}/\delta)},\\
        Z & \coloneqq \sigma_z \sqrt{2n_y\text{ log}(12n_yT_\text{w}/\delta)}.
    \end{split}
    \label{eq: bound on state and input during warm-up}
\end{equation}
As a consequence of the above bounds, we have 
\vspace{0mm}
\begin{equation}
    \begin{split}
        ||\phi_t|| & \leq \underbrace{(||C|| X_\text{w} + Z + U_\text{w})}_{\Upsilon_\text{w}}\sqrt{H},
    \end{split}
    \label{eq: Definition of Upsilon_w}
\end{equation}
which holds with a probability of at least $1 - \delta/2$ with $\delta \in (0,1/2)$, for all $t \in [H, T_\text{w} - 1]$.

\lem[\cite{lale2020regret}] \label{Lemma: PE warm-up}
For some $\sigma_\text{o} > 0$ and for some $T_\text{o} > H$, if the warm-up duration $T_\text{w} \geq T_\text{o}$, then for all $t \in [T_\text{o}, T_\text{w} - 1]$, and for any $\delta \in (0,1)$, with a probability of at least $1 - \delta$, we have
\vspace{0mm}
\begin{equation}
    \sigma_{\text{min}}\left(\sum_{i = H}^{T_\text{w} - 1} \phi_i\phi_i^\top \right) \geq (T_\text{w} - H) \frac{\sigma_\text{o}^2\text{min}\{\sigma_w^2, \sigma_z^2, \sigma_u^2\}}{2},
    \label{eq: PE warm-up}
\end{equation}
where
\vspace{0mm}
\begin{equation}
    T_\text{o} \coloneqq \frac{ 32 \Upsilon_w^4 H \text{log}\left(\frac{2H(n_y + n_u)}{\delta}\right)}{\sigma_\text{o}^4\text{min}\{\sigma_w^4,\sigma_z^4,\sigma_u^4\}}.
    \label{eq: To for warm-up}
\end{equation}

\lem[\cite{lale2021adaptive}] \label{Lemma: Markov parameter estimation error after warm-up} The initial estimate of the Markov parameters after the warm-up period of $T_\text{w} \geq \max\{H, T_\text{o}\}$ time steps, $\mathbf{\hat{M}}_{T_\text{w}}$, obeys the following bound with a probability of at least $1 - \delta$ for $\delta \in (0,1)$:
\vspace{0mm}
\begin{equation}
    ||\mathbf{\hat{M}}_{T_\text{w}} - \mathbf{M}|| \leq \frac{\mathrm{poly}(n_y, n_u, H, \Sigma_e, \Upsilon_\text{w})}{\sqrt{\sigma_\text{min}\left(V_{T_\text{w}} \right)}} = \tilde{\mathcal{O}}\left( \frac{1}{\sqrt{T_\text{w}}} \right),
\end{equation}
where $\Upsilon_\text{w}$ is as defined in \eqref{eq: Definition of Upsilon_w}. Moreover, if $T_\text{w} \geq T_{\mathbf{M}}$, then 
\vspace{0mm}
\begin{equation}
    ||\mathbf{\hat{M}}_{T_\text{w}} - \mathbf{M}|| \leq 1,
\end{equation}
where $T_\mathbf{M} = R_\text{warm}^2$,
\begin{equation*}
    R_\text{warm} \coloneqq \frac{\sqrt{2n_y||\Sigma_e||\log\left( \frac{\text{det}(V_{T_\text{w}})^{1/2}}{\delta \text{det}(\lambda I)^{1/2}} \right)} + \Bar{\mathbf{m}} \sqrt{2\lambda} + \frac{\sqrt{2H}}{T_\text{w}}}{\sigma_\text{o}\min\{\sigma_w, \sigma_z, \sigma_u\}},
\end{equation*}
and $\Bar{\mathbf{m}}$ is a positive constant satisfying $||\mathbf{M}||_\mathrm{F} \leq \Bar{\mathbf{m}}$.

\lem[\cite{lale2021adaptive}] \label{Lemma: confidence bound on the system parameters} Let $\mathcal{H} = \begin{bmatrix} \mathcal{H}_\mathbf{F} & \mathcal{H}_\mathbf{G} \end{bmatrix}$ be the concatenation of two Hankel matrices obtained from $\mathbf{M}$. The notations $\mathcal{H}_\mathbf{F}$ and $\mathcal{H}_\mathbf{G}$ have analogous expressions to the definition in \eqref{eq: Hankel matrix definition} but with the true model parameter $\Theta$. Let Assumption \ref{assumption:setS} hold, then we have $\mathcal{H}$ to be a rank-$n_x$ matrix. Let $\tilde{A},\tilde{B},\tilde{C},\tilde{L}$ be the similarity transformed model parameters obtained from $\mathbf{M}$ by using Algorithm \ref{algorithm: SYSID}. At time step $t$, let $\hat{A}_t, \hat{B}_t, \hat{C}_t, \hat{L}_t$ be the estimated model parameters obtained from $\mathbf{\hat{M}}_t$ via Algorithm \ref{algorithm: SYSID}. Then, for a given choice of $H \geq \Bar{H}$, there exists a unitary matrix $\mathbf{T} \in \mathbb{R}^{n_x \times n_x}$ such that, $\Tilde{\Theta} = (\Tilde{A},\Tilde{B},\Tilde{C},\Tilde{L}) \in (\mathcal{C}_A(t) \times \mathcal{C}_B(t) \times \mathcal{C}_C(t) \times \mathcal{C}_L(t))$, where
\vspace{0mm}
\begin{equation}
    \begin{split}
        & \mathcal{C}_A(t) = \left\{\ A'\in \mathbb{R}^{n_x \times n_x}: || \hat{A}_t - \mathbf{T}^\top A'\mathbf{T}|| \leq \beta_A(t)  \right\},\\
        & \mathcal{C}_B(t) = \left\{\ B'\in \mathbb{R}^{n_x \times n_u}: || \hat{B}_t - \mathbf{T}^\top B'|| \leq \beta_B(t)  \right\},\\
        & \mathcal{C}_C(t) = \left\{\ C'\in \mathbb{R}^{n_y \times n_x}: || \hat{C}_t - C'\mathbf{T}|| \leq \beta_C(t)  \right\},\\
        & \mathcal{C}_L(t) = \left\{\ L'\in \mathbb{R}^{n_u \times n_y}: || \hat{L}_t - \mathbf{T}^\top L'|| \leq \beta_L(t)  \right\},\\
    \end{split}
\end{equation}
where for some positive constants $c_A, c_{L,1}$ and $c_{L,2}$ that depend on the true model parameter $\Theta$, we have
\vspace{0mm}
\begin{equation}
    \begin{split}
        & \beta_A(t) = c_A \left( \frac{\sqrt{n_xH}(||\mathcal{H}|| + \sigma_{n_x}(\mathcal{H}))}{\sigma^2_{n_x}(\mathcal{H})}\right) || \hat{\mathbf{M}}_t - \mathbf{M}||,\\
        & \beta_B(t) = \beta_C(t) = \sqrt{\frac{20n_xH}{\sigma_{n_x}(\mathcal{H})}} || \hat{\mathbf{M}}_t - \mathbf{M}||,\\
        & \beta_L(t) = \frac{c_{L,1} ||\mathcal{H}||}{\sqrt{\sigma_{n_x}(\mathcal{H})}} \beta_A(t) + c_{L,2} \frac{\sqrt{n_xH}(||\mathcal{H}|| + \sigma_{n_x}(\mathcal{H}))}{\sigma^{3/2}_{n_x}(\mathcal{H})}|| \hat{\mathbf{M}}_t - \mathbf{M}||.\\
    \end{split}
\end{equation}

The following result can be found in the proof of \cite[Th. 3.4]{lale2021adaptive} and holds when the model parameter estimation error is monotonically decreasing.

\coro \label{Corollary: confidence bound on the system parameters after warm-up} For some $T_A, T_B > H$ and for a given choice of $H \geq \Bar{H}$, there exists a unitary matrix $\mathbf{T} \in \mathbb{R}^{n_x \times n_x}$ such that, with a probability of at least $1 - \delta$ for $\delta \in (0,1)$, we have
\vspace{0mm}
\begin{equation}
    \begin{split}
        & \mathcal{C}_A(T_\text{w}) = \left\{\ A'\in \mathbb{R}^{n_x \times n_x}: || \hat{A}_t - \mathbf{T}^\top A'\mathbf{T}|| \leq \beta_A(T_\text{w})  \right\},\\
        & \mathcal{C}_B(T_\text{w}) = \left\{\ B'\in \mathbb{R}^{n_x \times n_u}: || \hat{B}_t - \mathbf{T}^\top B'|| \leq \beta_B(T_\text{w})  \right\},\\
        & \mathcal{C}_C(T_\text{w}) = \left\{\ C'\in \mathbb{R}^{n_y \times n_x}: || \hat{C}_t - C'\mathbf{T}|| \leq \beta_C(T_\text{w})  \right\},\\
        & \mathcal{C}_L(T_\text{w}) = \left\{\ L'\in \mathbb{R}^{n_u \times n_y}: || \hat{L}_t - \mathbf{T}^\top L'|| \leq \beta_L(T_\text{w})  \right\},\\
    \end{split}
\end{equation}
\vspace{0mm}
where 
\vspace{0mm}
\begin{equation}
    \begin{split}
        & \beta_A(T_\text{w}) = \frac{\sigma_{n_x}(A)}{2} \text{ if }T_\text{w} \geq T_A,\\
        & \beta_B(T_\text{w}) = \beta_C(T_\text{w}) = 1 \text{ if }T_\text{w} \geq T_B,\\
        & \beta_L(T_\text{w}) = \frac{c_{L,1} ||\mathcal{H}||}{\sqrt{\sigma_{n_x}(\mathcal{H})}} \beta_A(T_\text{w}) + c_{L,2} \frac{\sqrt{n_xH}(||\mathcal{H}|| + \sigma_{n_x}(\mathcal{H}))}{\sigma^{3/2}_{n_x}(\mathcal{H})} \text{ if }T_\text{w} \geq T_A.\\
    \end{split}
\end{equation}

\subsection{Technical background related to LQG-IF2E} \label{appendix: technical background LQG-IF2E}

\lem[\cite{ziemann2021uninformative}]{(Chain rule for Fisher information)} \label{Lemma: Chain rule for Fisher information} For  random variables $x \in \mathbb{R}^n$ and $y \in \mathbb{R}^m$, and a density $p$, consider the following FIM for a bivariate density $p_\theta (x,y)$ parameterized by $\theta$:
\vspace{0mm}
\begin{equation*}
    \Bar{I}_{p(x,y)}(\theta) = \int_{\mathbb{R}^m} \int_{\mathbb{R}^n} \nabla_\theta \log p_\theta(x,y) \left( \nabla_\theta \log p_\theta(x,y) \right)^\top p_\theta(x,y) dx dy.
\end{equation*}
Define the conditional FIM as 
\vspace{0mm}
\begin{equation*}
    \Bar{I}_{p(x|y)}(\theta) = \int_{\mathbb{R}^m} \int_{\mathbb{R}^n} \nabla_\theta \left( \log p_\theta(x|y) \right) \left[ \nabla_\theta \left( \log p_\theta(x|y) \right) \right]^\top p_\theta(x|y)dx\hspace{1mm}p_\theta(y)dy.
\end{equation*}
Then
\vspace{0mm}
\begin{equation*}
    \Bar{I}_{p(x,y)}(\theta) = \Bar{I}_{p(x|y)}(\theta) + \Bar{I}_{p(y)}(\theta),
\end{equation*}
assuming that $\nabla_\theta \log p_\theta(x|y)$ and $\nabla_\theta \log p_\theta(y)$ have mean zero.

\lem[\cite{ziemann2022regret}] \label{Lemma: FIM of a Gaussian random variable} Let $\mu: \mathbb{R}^{d_\theta} \rightarrow \mathbb{R}^d$ and $V: \mathbb{R}^{d_\theta} \rightarrow \mathbb{R}^{d \times d}$, with $V > 0$  for all $\theta \in \mathbb{R}^{d_\theta}$, and define 
\vspace{0mm}
\begin{equation*}
    \gamma_\theta(x) = \frac{1}{\sqrt{(2\pi)^d \text{det}(V(\theta))}}\exp \left( -\frac{1}{2} (x - \mu(\theta))^\top {V(\theta)}^{-1} (x - \mu(\theta)) \right).
\end{equation*}
Then
\vspace{0mm}
\begin{equation*}
    \Bar{I}_\gamma(\theta) = \left( \verb|D|_\theta \mu(\theta) \right)^\top V(\theta)^{-1} \left( \verb|D|_\theta \mu(\theta) \right) + \frac{1}{2} \left( \verb|D|_\theta \verb|vec| \left( V(\theta) \right) \right)^\top \left(I \otimes {V(\theta)}^{-2} \right) \verb|D|_\theta \verb|vec| \left(V(\theta)\right).
\end{equation*}

\lem \label{Lemma: FIM} For the approximate model \eqref{eq: input-output trajectory representation}, the FIM under any policy $\pi$ after collecting the observations $\{y_i\}_{i = 0}^{t}$ and $\{u_i\}_{i = 0}^{t-1}$ for $t \geq H$, is given by 
\vspace{0mm}
\begin{equation*}
    I_{H,t} = \sum_{i = H}^{t} \mathbb{E}\left[ \phi_i {\phi_i}^\top \otimes \Sigma_e^{-1} \right].
\end{equation*}
\begin{proof}
This proof is an extension from the state measurement case addressed in \cite{ziemann2022regret} to the partial observability case. From Lemma \ref{Lemma: Chain rule for Fisher information}, we have for the observations $\{y_i\}_{i = 0}^{t}$ and $\{u_i\}_{i = 0}^{t-1}$
\vspace{0mm}
\begin{equation*}
    I_{H,t} = \Bar{I}_{p\left( \{y_i\}_{i=H}^{t}, \{\phi_i\}_{i=H}^{t}\right)} = \sum_{i = H}^{t} \mathbb{E}\left[ \mathcal{L}_i(\mathbf{M})\right],
\end{equation*}
where $p\left( \{y_i\}_{i=H}^{t}, \{\phi_i\}_{i=H}^{t}\right)$ is a multivariate density function for the sequences $\{ y_i \}_{i = H}^{t}$ and $\{ \phi_i \}_{t = H}^{t}$. Further,
\vspace{0mm}
\begin{equation*}
\begin{split}
    \mathcal{L}_t(\mathbf{M}) = \int_{\mathbb{R}^{n_y}} \nabla_{\mathbf{M}} \log p\left( \Bar{y} - \mathbf{M}\phi_t - e_t\right) \left( \nabla_{\mathbf{M}} \log p\left( \Bar{y} - \mathbf{M}\phi_t - e_t\right) \right)^\top\\
    \cdot p\left( \Bar{y} - \mathbf{M}\phi_t - e_t\right)d\Bar{y},
\end{split}
\end{equation*}
where $\Bar{y}$ is a dummy variable for integration. Notice that $y_t|\{(y_i, u_i)\}_{i = 0}^{t-1} \sim \mathcal{N}(\mathbf{M} \phi_t, \Sigma_e)$, where $\Sigma_e = C\Sigma C^\top + \sigma^2_z I$. Therefore, from Lemma \ref{Lemma: FIM of a Gaussian random variable}, with $y_t = \mathbf{M} \phi_t + e_t$, we have
\vspace{0mm}
\begin{equation*}
    \begin{split}
        \mathcal{L}_t(\mathbf{M}) & =  \mathbb{E} \left[\left( \Big(\texttt{D}_{\mathbf{M}}\mathbf{M}\right)\phi_t \Big)^\top \Sigma_e^{-1} \left(\texttt{D}_{\mathbf{M}}\mathbf{M}\right)\phi_t \right]\\
         & = \mathbb{E} \left[ \left( \nabla_{\mathbf{M}}\texttt{vec}(\mathbf{M}) \right)^\top \left( \phi_t {\phi_t}^\top \otimes \Sigma_e^{-1} \right) \nabla_{\mathbf{M}}\texttt{vec}(\mathbf{M}) \right]\\
        & = \mathbb{E} \left[ \phi_t {\phi_t}^\top \otimes \Sigma_e^{-1} \right].\\
    \end{split}
\end{equation*}
\vspace{0mm}
Then
\vspace{0mm}
\begin{equation*}
     I_{H,t} = \sum_{i = H}^{t} \mathbb{E} \left[ \phi_i {\phi_i}^\top \otimes \Sigma_e^{-1}\right].
\end{equation*}
This concludes the proof.  
\end{proof}

\subsection{Technical Tools} \label{appendix B: technical lemmas}

\subsection*{Definitions:}

\defn[\cite{simchowitz2020naive}]{(Discrete Lyapunov equation)} \label{Definition B.1} Let $X,Y \in \mathbb{R}^{m \times m}$ with $Y =Y^\top$ and $\rho(X) < 1$. We let $\mathcal{T}_X[P] \coloneqq X^\top P X + Y$, and let $\text{dlyap}(X,Y)$ denote the unique positive semi-definite solution $\mathcal{T}_X[P] = P$.

\defn[\cite{mania2019certainty}] \label{Definition: Phi(A)}  For any square matrix $M$, define the following:
\begin{equation*}
   \Xi(M) \coloneqq \text{sup}_{\tau \geq 0} \frac{||{M}^\tau||}{\rho_M^\tau},
\end{equation*}
for some $\rho_M \geq \rho(M) $. Moreover, if $M$ is Schur stable, a $\rho_M \in (0,1)$ can be chosen such that $ \Xi(M) \geq 1$ and $\Xi(M)$ is finite.

\subsection*{Probabilistic bounds:}

\lem[\cite{abbasi2011regret}] \label{Lemma B.1}
Let $v \in \mathbb{R}^d$ be an entry-wise $M$-sub-Gaussian random variable. Then with probability of at least $1-\delta$, $||v|| \leq M\sqrt{2d\text{ log}(2d/\delta)}$.

\lem{(Gaussian concentration inequality)} \label{Lemma B.5} Let $X = \begin{bmatrix} X_1,...,X_n \end{bmatrix}^\top$ be a vector with i.i.d. standard Gaussian entries and $G: \mathbb{R}^n \rightarrow \mathbb{R}$ a $\mathcal{L}$-Lipschitz function ($|G(x) - G(y)| \leq \mathcal{L} ||x-y||$, for all $x,y \in \mathbb{R}^n$). Then, for every $t \geq 0$ 
\vspace{0mm}
\begin{equation*}
    \mathbb{P}\left\{ |G(X) - \mathbb{E}[G(X)]| \geq t\right\} \leq 2\exp\left( \frac{-t^2}{2\mathcal{L}^2} \right).
\end{equation*}

\lem[\cite{tropp2012user}] \label{Lemma B.2}
 Consider a self-adjoint matrix martingale $\{\mathbf{Y}_k: k = 1,..,n\}$ in dimension $d$, and let $\{\mathbf{X}_k\}$ be the associated difference equation. Consider also a fixed sequence $\{\mathbf{A}_k\}$ of self-adjoint matrices that satisfy
\vspace{0mm}
\begin{equation*}
    \mathbb{E}_{k-1}\mathbf{X}_k = 0 \text{ and } \mathbf{X}_k^2 \leq \mathbf{A}_k^2 \text{ almost surely.}
\end{equation*}
Compute the variance parameter
\vspace{0mm}
\begin{equation*}
    \sigma^2 := \mid\mid \sum_k \mathbf{A}_k^2 \mid\mid.
\end{equation*}
Then for all $t\geq 0$,
\vspace{0mm}
\begin{equation*}
    \mathbb{P}\left\{\lambda_\text{max}(\mathbf{Y}_n - \mathbb{E}\mathbf{Y}_n)\geq t\right\} \leq d\hspace{1mm}e^{\frac{-t^2}{8\sigma^2}}.
\end{equation*}

\thm[\cite{abbasi2011regret}] \label{Theorem B.1}
Let $(\mathcal{F}_t;k \geq 0)$ be a filtration, $(m_k; k \geq 0)$ be an $\mathbb{R}^d$ - valued stochastic process adapted to $(\mathcal{F}_k)$, $(f_k; k \geq 1)$ be a real-valued martingale difference process adapted to $(\mathcal{F}_k)$. Assume that $f_k$ is conditionally sub-Gaussian with constant $M$. Consider the martingale
\vspace{0mm}
\begin{equation*}
    \mathcal{T}_t = \sum_{k=1}^t f_k m_{k-1}
\end{equation*}
and the matrix-valued processes for $V > 0$
\vspace{0mm}
\begin{equation*}
    V_t = \sum_{k=1}^t m_{k-1}m_{k-1}^T, \hspace{4mm} \Bar{V}_t = V + V_t, \hspace{2mm} t \geq 0.
\end{equation*}
Then for any $0 < \delta < 1$, with probability $1-\delta$,
\vspace{0mm}
\begin{equation*}
    \mathcal{T}_t^\top \Bar{V}_t^{-1} \mathcal{T}_t \leq 2M^2 \log \left( \frac{\text{det}(\Bar{V}_t)^{1/2}\text{det}(V)^{-1/2}}{\delta}\right) \hspace{2mm} \forall t \geq 0.
\end{equation*}

\thm[\cite{rudelson2013hanson}]{(Hanson-Wright inequality)} \label{Theorem B.3 Hanson-Wright inequality} Let $X = (X_1,...,X_n) \in \mathbb{R}^n$ be a random vector with independent components $X_i$ which satisfy $\mathbb{E}[x_i] = 0$ and $||X_i||_{\psi_2} \leq k$ for all $i = 1,..,n$, where $||.||_{\psi_2} = \sup_{p \geq 1} p^{-1/2}(\mathbb{E}[.]^p)^{1/p}$ is the sub-Gaussian norm. Let $A$ be an $n \times n$ matrix. Then, for every $t \geq 0$,
\vspace{0mm}
\begin{equation*}
    \mathbb{P}\left\{ |X^\top AX - \mathbb{E}{X^\top AX}| > t \right\} \leq 2 \exp \left[ -c \min \left( \frac{t^2}{k^4 ||A||^2_\mathrm{F}}, \frac{t}{k^2||A||}\right) \right],
\end{equation*}
where $c$ is a positive absolute constant.

\subsection*{Perturbation bounds:}

\begin{lem} \label{Lemma B.13}
For $i \in \{1,\dots, n\}$ and $\hat{W}_i$, $W_i \in \mathbb{R}^{m_i \times m_{i+1}}$, if $\|W_i\|, \|\hat{W}_i\|$ are upper bounded by a constant, and if $\|W_i - \hat{W}_i\| = \mathcal{O}(f(\cdot))$, where $f(\cdot)$ is any arbitrary function, then it holds that $\| W_1 \cdots W_n -\hat{W}_1 \cdots \hat{W}_n \| = \mathcal{O}(f(\cdot))$.
\end{lem}

\begin{proof}
The above result can be proved by induction. When $n=1$, it trivially holds. Assume that when $n=k$, we have $\| W_1 \cdots W_k -\hat{W}_1 \cdots \hat{W}_k \| = \mathcal{O}(f(\cdot))$. Then $\| W_1 \cdots W_k W_{k+1} -\hat{W}_1 \cdots \hat{W}_k \hat{W}_{k+1} \|  \leq \| (W_1 \cdots W_k - \hat{W}_1 \cdots \hat{W}_k) W_{k+1} +\hat{W}_1 \cdots \hat{W}_k (W_{k+1} -\hat{W}_{k+1}) \| = \mathcal{O}(f(\cdot))$, which proves the property for $n=k+1$ and thus proves the result by induction.
    
\end{proof}

\lem[\cite{mania2019certainty}] \label{Lemma: general matrix perturbation bound} Let $M$ be an arbitrary matrix in $\mathbb{R}^{n \times n}$ and let $\Xi(M)$ be as defined in Definition \ref{Definition: Phi(A)}. Then, for all $\tau \geq 1$ and real matrices $\Delta$ of appropriate dimensions we have
\vspace{0mm}
\begin{equation*}
    || (M + \Delta)^\tau || \leq \Xi(M) \left( \Xi(M) ||\Delta|| + \rho_M \right)^k,
\end{equation*}
where $\rho_M \geq \rho(M)$.

\thm[\cite{simchowitz2020naive}] \label{Theorem: perturbation bound on K} Let $(A,B)$ be a stabilizable system and $P$ be the solution of the DARE \eqref{eq: DARE to compute K} formulated on $(A,B)$. Given an alternate pair of matrices $(\hat{A}, \hat{B})$, let $\epsilon = \max\{||\hat{A} - A||, ||\hat{B} - B||\}$. Then if $\epsilon \leq \frac{1}{54||P||^5}$,
\begin{enumerate}
    \item $||P(\hat{A}, \hat{B})|| \leq 1.0835||P||$ and
    \item $||B\left(K(\hat{A}, \hat{B}) - K\right)|| < \frac{1}{5||P||^{3/2}}$.
\end{enumerate}

\subsection*{Bounds related to matrix functions:}

\lem[\cite{abbasi2011improved}]{(Regularised design matrix Lemma)} \label{Lemma B.4}
When the covariates $z_t z_t^\top$ satisfies $||z_t|| \leq D$, with some $D > 0$ with probability 1 then,
\vspace{0mm}
\begin{equation*}
    \log \frac{\text{det}(V_t)}{\text{det}(\lambda I)} \leq d \log \left( \frac{\lambda d + t D^2}{\lambda d} \right),
\end{equation*}
where $V_t = \lambda I + \sum_{i = 1}^t z_i z_i^\top$ for $z_i \in \mathbb{R}^d$.

\lem[\cite{simchowitz2020naive}] \label{Lemma B.11}
Let $\text{Toep}_{i,j,l}(X)$ and $\text{Col}_{i,j}(X)$ be as defined in \eqref{eq: Definition of Toep, Col} for $X \in \mathbb{R}^{m \times m}$. For any $i \leq j,l$, and for $Y \in \mathbb{R}^{m \times m}$, and $\text{diag}_{j-i}(Y)$ denoting a $j-i$ block matrix with blocks $Y$ on the diagonal, we have the bound
\vspace{0mm}
\begin{equation*}
    \text{Tr}\left(\text{Col}_{i,j}(X)^\top \text{diag}_{j-i}(Y)\text{Col}_{i,j}(X) \right) \leq \text{Tr}(\text{dlyap}(X,Y)).
\end{equation*}

The following result is a straightforward extension of Lemma \ref{Lemma B.11}.

\coro \label{Corollary B.11} Let $\text{Toep}_{i,j,l}(X)$ and $\text{Col}_{i,j}(X)$ be as defined in \eqref{eq: Definition of Toep, Col} for $X \in \mathbb{R}^{m \times m}$. For any $i \leq j,l$, and for $Y \in \mathbb{R}^{m \times m}$, and $\text{diag}_{j-i}(Y)$ denoting a $j-i$ block matrix with blocks $Y$ on the diagonal, we have the bound
\vspace{0mm}
\begin{equation*}
    \text{Col}_{i,j}(X)^\top \text{diag}_{j-i}(Y)\text{Col}_{i,j}(X)  \leq \text{dlyap}(X,Y).
\end{equation*}

\lem[\cite{simchowitz2020naive}] \label{Lemma B.12}  Let $\text{Toep}_{i,j,l}(X)$ and $\text{Col}_{i,j}(X)$ be as defined in \eqref{eq: Definition of Toep, Col} for $X \in \mathbb{R}^{m \times m}$. For any $i \leq j,l$, we have $||\text{Col}_{i,j}(X)|| \leq ||\text{Toep}_{i,j,l}(X)|| \leq ||X||_{\mathcal{H}_\infty}$.

\lem[\cite{shi2023suboptimality}] \label{Lemma: decay of closed loop matrix} Given $(A,B)$ a stabilizable pair, an alternative pair $(\hat{A}, \hat{B})$ satisfying $\max\{||\hat{A} - A||, ||\hat{B} - B||\} \leq \epsilon$, where $\epsilon < 1/(8||P||^2)$, and any positive integer $k$, we have 
\begin{equation*}
    ||(\hat{A} - K(\hat{A}, \hat{B})\hat{B})^k|| \leq \sqrt{\alpha_\epsilon} \sqrt{\beta_*}(\gamma(\epsilon))^k,
\end{equation*}
where $\gamma(\epsilon) \triangleq \sqrt{1 - \alpha_\epsilon^{-1}\beta_*^{-1}}$, $\beta_* = \sigma_\text{max}(P) / \sigma_\text{min}(Q)$ and $\alpha_\epsilon \triangleq (1 - 8||P||^2 \epsilon)^{-1/2}$, $P$ is the solution of the DARE \eqref{eq: DARE to compute K} formulated on $(A,B)$.

\end{document}